\documentclass[11pt,fullpage,letterpaper]{article} 
  \usepackage[margin=1in]{geometry}
  \usepackage{listings} 
\usepackage{mathrsfs} 
\usepackage{kpfonts,comment}
\usepackage[T1]{fontenc}

\usepackage{kz_style}

 \newcommand{\myversion}{2}
\newenvironment{proofsketch}{\begin{proof}[Proof Sketch]}{\end{proof}}

\title{
\LARGE  \bf  Joint Communication-Control Strategy   Optimization 
with Partially Nested Information Structures: \\ 
The Linear-Quadratic Case}  
\author{ 
Haoyi You$^\dag$ \and \qquad\qquad  Kaiqing Zhang\thanks{The authors are affiliated with the University of Maryland, College Park, MD, USA, 20742. Emails: 
        {\tt\small \{yuriiyou,~kaiqing\}@umd.edu}.}
        }
\date{}

\usepackage{setspace}
\renewcommand{\thesection}{\Roman{section}}

\usepackage{titlesec}

\titleformat{\section}
  {\centering\normalfont\Large\bfseries}
  {\thesection.}{0.5em}{}

\begin{document} 
\maketitle

\begin{abstract}%
In this paper, we formalize a joint communication-control strategy optimization (JCCO) problem in multi-agent linear systems with quadratic costs, {under the common-information-based (CIB) framework from decentralized stochastic control.} 
For computational tractability, we focus on such JCCO problems with partially nested (PN) information structures (ISs). 
In particular, with a baseline communication protocol that leads to a PN IS, we establish a series of conditions under which the partial nestedness is preserved under the (additional) communication strategies to be optimized, while violating them may cause {nonlinearity} or non-existence of the optimal control strategies, with open-loop communication strategies.  
Under these conditions, 
we then develop a dynamic-programming-based approach to compute the optimal control strategies of JCCO with open-loop communication strategies, which yields a set of closed-form Riccati Equations.  As a byproduct of independent interest, such an approach also offers a new way to solve decentralized linear-quadratic control with PN ISs and \emph{output} feedback, under the CIB framework and the aforementioned conditions.  
Finally, we extend such an approach to JCCOs with closed-loop communication strategies, yielding a more tractable dynamic program than an infinite-dimensional CIB-belief-based one. 
\end{abstract}

\section{Introduction}
The design of communication protocols jointly with control has been extensively studied in both the control theory and multi-agent learning literature, 
from different perspectives and under different models. For example, 
\cite{zhang2006communication,peng2013event} investigated the design of communication channels with capacity constraints in networked control systems; 
\cite{matni2015communication,matni2016regularization} studied communication architecture and control strategy co-design via norm minimization; \cite{yuksel2013jointly,fu2012lack,maity2021optimal} investigated the joint optimization of the quantization and control policies, with the quantized signals being communicated to the controller; and 
\cite{2016learningtocommunicate,learningpropogation} developed heuristic learning-to-communicate  algorithms to jointly learn the control and communication {strategies} for accumulated  reward maximization.

In this paper, we formalize a joint communication-control strategy optimization (JCCO) problem in multi-agent linear systems with quadratic costs. Motivated by the empirical studies \cite{2016learningtocommunicate,learningpropogation}, we model communication as some \emph{actions/strategies} to be optimized \emph{jointly} with the control strategies, under an \emph{optimal control}  formulation to minimize some accumulated costs. The agents may follow some fixed \emph{baseline}  communication protocols for information sharing, and then share their private information through \emph{additional sharing}, following their communication strategies. We formalize the problem under the common-information-based (CIB) framework from decentralized stochastic control \cite{ashutosh2013team,CIBLQgames}, which has also been adopted in several recent studies on control-communication strategy co-optimization \cite{ashutoshcommunicate1,ashutoshcommunicate2,LTC}. In comparison, \cite{ashutoshcommunicate1,ashutoshcommunicate2} focused specifically on {sharing instantaneous observations and on systems} with \emph{decoupled} state dynamics. Our formalism in \cite{LTC} is more general and allows the sharing of private-information {histories in discrete-space problems, together with computational and sample complexity analyses}. 
We here focus on an instantiation of our  formalism in \cite{LTC} in the continuous-space, linear-quadratic (LQ) setting, which yields fundamentally different technical challenges as detailed below. 

To facilitate tractable computation of the optimal strategies, 
we focus on deriving \emph{dynamic-programming} (DP)-based approaches to solve  JCCO. To this end, we concentrate on such co-optimization problems with partially nested (PN)
information structures (ISs) \cite{ho1972team}. Indeed, with a \emph{fixed}  open-loop communication strategy (and thus a fixed IS), partial nestedness is known to be a favorable IS that not only ensures the \emph{linearity} of an optimal control strategy, but also yields a 
\emph{convex reformulation} of the decentralized LQ control problem \cite{ho1972team}. However, as pointed out in  \cite{QCstatefeedback}, the reformulated convex program can be too large to be computationally efficient, which precisely motivated the DP-based solution therein. Moreover, the DP-based approach can provide more insights into the structure of the optimal controller \cite{QCstatefeedback}, by identifying the proper \emph{sufficient statistics} for decision-making \cite{QCstatefeedback,ashutosh2013team,AA,nayyar2015structural}.  

Yet, the results in  \cite{QCstatefeedback,nayyar2015structural} do not apply to our setting, as they focused on the special cases with {factorized} states, and with either {state-feedback} \cite{QCstatefeedback} or {multi-tree} coupling and communication graphs \cite{nayyar2015structural}.  Under the CIB framework \cite{ashutosh2013team}, \cite{AA} considered the output-feedback setting, and showed that when control strategies are \emph{restricted to be linear}, and the private-information component of the strategy is \emph{fixed}, the resulting problem is linear-quadratic-Gaussian (LQG) control, which can be solved via Riccati Equations. However, it remains unclear how to solve the overall decentralized LQ control problem by further optimizing over the private-information component. 
Since given a \emph{fixed} communication strategy, our JCCO problem reduces to a decentralized LQ control one, we need to advance these results in order to fully address JCCO. 
We thus make the following contributions: 

\paragraph{Contributions.}  (i) We formalize \emph{joint communication-control strategy optimization} as an optimal control problem with linear systems and quadratic costs, under the common-information-based framework  \cite{ashutosh2013team}. (ii) For better computational tractability, 
we focus on JCCO with 
\emph{partially nested} baseline ISs  \cite{ho1972team}, and identify structural conditions under which partial nestedness is preserved under additional information sharing, while violating them may cause nonlinearity or non-existence of the optimal control strategies, under  \emph{open-loop}  communication strategies. (iii) We then develop a closed-form dynamic program, i.e., a set of Riccati Equations, to solve for the optimal control strategy with respect to {fixed} {open-loop}  communication strategies in JCCO. This program is also of independent interest: it may be viewed as a new way to solve decentralized linear-quadratic control with PN ISs and \emph{output} feedback, under the CIB framework and the aforementioned structural conditions, which thus advances the results in \cite{AA}.    
The key is to identify a novel connection between the (strictly) PN IS \cite{ho1972team} and the strategy-independent common-information-based belief (SI-CIB) condition \cite{CIBLQgames}  in this linear-quadratic setting. (iv) We then extend such an approach to JCCO with \emph{closed-loop} communication strategies, yielding reduced-dimensional sufficient statistics and a more tractable dynamic program than an infinite-dimensional belief-based one, when applying the CIB framework \cite{ashutosh2013team,CIBLQgames} 
directly. 

\section{Preliminaries} 
\noindent\textbf{Notation.}  Random variables are denoted by bold upper case letters, and their realizations are denoted by the corresponding non-bold upper case letters. 
For any matrix $X$, we use 
$X^\dag$ to denote the 
pseudo-inverse of $X$ if $X\succeq 0$. For any two integers $0\le a<b$, we denote $[a:b]:=\{a,a+1,\cdots,b\}$, and denote $[a]=[1:a]$. 
For any vector space $\cX$, we denote by $\cP(\cX)$  the space of all probability distributions over $\cX$.

\subsection{Problem formulation}

For a team of $n>1$ agents, a \emph{joint communication-control  strategy  optimization} problem with \emph{output feedback} can be described by the following tuple: $\cD=\langle H, \cX,\{\cY_i\}_{i\in[n]},\{\cU_i\}_{i\in[n]}, \{\cM_{i,h}\}_{i\in[n],h\in[H]}, \{A_h\}_{h\in[H]},
\{B_{i,h}\}_{i\in[n],h\in[H]}, \{E_{i,h}\}_{i\in[n],h\in[H]}, \{Q_h^1\}_{h\in[H+1]},\{Q_h^2\}_{h\in[H]},\allowbreak\{\cK_h\}_{h\in[H]}\rangle$, where $H$ is the time horizon and   $\bX_h\in\cX=\RR^{d_x}$ is the state. At each timestep $ h\in[H]$, each agent $i\in[n]$ receives a noisy observation $\bY_{i,h}\in\cY_i=\RR^{d_y^i}$ of the state $\bX_h$, and chooses a control action $\bU_{i,h}\in\cU_i=\RR^{d_u^i}$. 
At timestep $h\in[H]$, 
we denote by $\bU_h=\begin{bmatrix}
    \bU_{1,h}^\top&\bU_{2,h}^\top&\cdots&\bU_{n,h}^\top
\end{bmatrix}^\top$ the joint control action of all the $n$ agents, and by $\cU=\RR^{\sum_{i=1}^nd_u^i}$ the joint control action space; we denote by $\bY_h=\begin{bmatrix}
    \bY_{1,h}^\top&\bY_{2,h}^\top&\cdots&\bY_{n,h}^\top
\end{bmatrix}^\top$ the joint observation, and by $\cY=\RR^{\sum_{i=1}^nd_y^i}$ the joint observation space. 
The system evolves as follows for each timestep $h\in[H]$:
\vspace{-6pt}
\begin{equation}\label{equ:system_dynamics}
        \bX_{h+1}=A_h\bX_h+\sum_{i=1}^nB_{i,h}\bU_{i,h}+\bW_{0,h}, \qquad\bY_{i,h}=E_{i,h}\bX_h+\bW_{i,h}, \forall i\in[n],
\end{equation}
where for each $i\in[0:n], h\in[H]$, $\bW_{i,h}$ 
is a Gaussian random variable with distribution $\cN(\bm{0}, \Sigma_{i,h})$ for some covariance matrix $\Sigma_{i,h}\succeq 0$, while $A_h, E_{i,h}$  are matrices of appropriate dimensions. We define $B_h:=\begin{bmatrix}
    B_{1,h}&B_{2,h}&\cdots&B_{n,h}
\end{bmatrix}$ and $E_h:=\begin{bmatrix}
    E_{1,h}^\top&E_{2,h}^\top&\cdots&E_{n,h}^\top
\end{bmatrix}^\top$. 
Here, we assume that the initial state $\bX_1$ follows a Gaussian distribution $\cN(\bm{0}, \Sigma_{1})$  with covariance matrix $\Sigma_{1}\succeq 0$, and that $\bX_1$ and $\{\bW_{i,h}\}_{i\in[0:n],h\in[H]}$ are mutually independent. 

At timestep $h\in[H]$, each agent will share part of her information with other agents. The shared information $\bZ_h:=\bZ_h^b\cup \bZ_h^a$ consists of two parts, the \emph{baseline-sharing} part $\bZ_h^b$, which originates  from some existing sharing protocol, and the \emph{additional-sharing}  part $\bZ_h^a$, which is \emph{decided/learned} by agents, with joint additional-sharing information $\bZ_h^a:=\cup_{i=1}^n \bZ_{i,h}^a$. The baseline-sharing part is introduced for generality (i.e., $\bZ_h^b$ may be set as  $\emptyset$), and some benign information structures to be introduced later may require a certain amount of baseline sharing; see \cite{liu2023tractable,LTC} for concrete examples.  
At timestep $h$, the common information among all the agents is thus defined as the union of all the shared information so far: $\bC_{h^-}=\cup_{t=1}^{h-1}\bZ_t\cup \bZ_h^b, \bC_{h^+}=\cup_{t=1}^{h}\bZ_t$, where  $\bC_{h^-}$ and $\bC_{h^+}$ denote the common information \emph{before} and \emph{after}  additional sharing, respectively. The private information of agent $i$ at timestep $h$ {before} and {after}  additional sharing is denoted by $\bP_{i,h^-},\bP_{i,h^+}$, respectively, where $\bP_{i,h^-}\subseteq\{\bY_{1:h},\bU_{1:h-1}\}\backslash\bC_{h^-},\bP_{i,h^+}\subseteq\{\bY_{1:h},\bU_{1:h-1}\}\backslash\bC_{h^+}$. Then, we define  $\bI_{i,h^-}:=\bP_{i,h^-}\cup \bC_{h^-}, \bI_{i,h^+}:=\bP_{i,h^+}\cup\bC_{h^+}$ as the  information available to agent $i$ before and after additional sharing, respectively. We denote by $\bP_{h^-}:=[\bP_{1,h^-}^\top~\cdots~\bP_{n,h^-}^\top]^\top$ the joint private information at timestep $h$ before additional sharing, and similarly by $\bP_{h^+}$ the joint private information after additional sharing.
We denote by $\cC_{h^-}, \cC_{h^+},\cP_{i,h^-},\cP_{i,h^+},\cP_{h^-},\cP_{h^+}, \cI_{i,h^-},\cI_{i,h^+}$ the sets of all possible values of the respective random variables $\bC_{h^-}, \bC_{h^+},\bP_{i,h^-},\bP_{i,h^+},\bP_{h^-},\bP_{h^+},\bI_{i,h^-},\bI_{i,h^+}$.

At each timestep $h$, in addition to the control action, each agent $i$ needs to choose a communication action $\bM_{i,h}\in\cM_{i,h}$ to determine what information $\bZ_{i,h}^a$ she will share, as specified below. We use $\bM_{h}:=(\bM_{1,h},\cdots,\bM_{n,h})\in\cM_h$ to denote the joint communication action 
at timestep $h$. 

We assume that the evolution of common and private information in $\cD$ follows the rules below, where we adopt the convention that any quantity at timestep $0$ is empty/null. 
\begin{assumption}[Information Evolution]\label{ass: evolution rule} ~
\begin{enumerate}[(a)] 
        \item (Baseline sharing). The baseline sharing evolves as $\bZ_h^b=\chi_h(\bP_{(h-1)^+},\bU_{h-1},\bY_{h})$ for some fixed projection function $\chi_h$.
        \item (Additional sharing). For each $i\in[n]$,  the additional sharing evolves as $\bZ_{i,h}^a=\phi_{i,h}(\bM_{i,h},\bP_{i,h^-})$ for some function $\phi_{i,h}$, where for each realization $M_{i,h}\in \cM_{i,h}$, $\phi_{i,h}(M_{i,h}, \cdot)$ is a fixed projection function. The joint additional sharing $\bZ_h^a$ is thus generated by $\bZ_h^a=\phi_h(\bM_h,\bP_{h^-})$ for some  function $\phi_h$.
        \item For each $i\in[n], \bP_{i,h^-}=\zeta_{i,h}(\bP_{i,(h-1)^+},\bU_{i,h-1}, \bY_{i,h})$  for some fixed projection function $\zeta_{i,h}$, and the joint private information thus evolves as $\bP_{h^-}=\zeta_{h}(\bP_{(h-1)^+},\bU_{h-1},\bY_h)$ for some projection function $\zeta_{h}$.
        \item For each $i\in[n], \bP_{i,h^+}=\bP_{i,h^-}\backslash\bZ_{i,h}^a$.
        \item For each $i\in[n]$ and $h\in[H]$, $\bY_{i,h}\in\bI_{i,h^-}$ and $\bI_{i,h^-}\subseteq\bI_{i,h^+}$; for $h\in[H-1]$, $\bI_{i,h^+}\subseteq\bI_{i,(h+1)^-}$.
    \end{enumerate}
\end{assumption}

Note that (a) and (c) on baseline sharing follow from those in \cite{CIBLQgames,liu2023tractable,LTC}. (b) and (d) on additional sharing dictate how the communication action affects the additional sharing. For example, a common choice of $(\cM_{i,h},\phi_{i,h})$ is that $\cM_{i,h}=\{0,1\}^{\max\limits_{P_{i,h^-}\in \cP_{i,h^-}}|P_{i,h^-}|}$,  where $|P_{i,h^-}|$ denotes the number of elements in $P_{i,h^-}$, and for any $P_{i,h^-}\in\cP_{i,h^-},M_{i,h}\in\cM_{i,h}$, $\phi_{i,h}(\bM_{i,h}=M_{i,h},\bP_{i,h^-}=P_{i,h^-})$ consists of the $k$-th element $(k\in[|P_{i,h^-}|])$ of $P_{i,h^-}$ if and only if the $k$-th element of $M_{i,h}$ is 1. Condition (e) means that each agent has full memory of the available information, and has a \emph{closed-loop} (baseline) information structure with $\bY_{i,h}\in \bI_{i,h^-}$, which are standard assumptions also made in  \cite{LTC,yuksel2023stochastic}.

\subsection{Objective and strategies}
At each timestep $h\in[H]$, agents will incur two types of costs, communication cost $\kappa_h$ and control cost $c_h$. The control cost inherits the cost of a decentralized LQG problem without communication, see e.g., \cite{AA,QCstatefeedback},  and is defined as $c_h=\bX_h^\top Q_h^1\bX_h+\bU_h^\top Q_h^2\bU_h$, with a terminal cost of $c_{H+1}=\bX_{H+1}^\top Q_{H+1}^1 \bX_{H+1}$.  Here, for every $h\in[H]$, $Q_h^1,Q_h^2\succeq0$ and $Q_{H+1}^1\succeq0$ are symmetric matrices of appropriate dimensions. 
The communication cost is determined by the additional sharing and is defined as $\kappa_h=\cK_h(\bM_h)$ for some function $\cK_h:\cM_h\rightarrow \RR_+$, depending on the communication action $\bM_h$ chosen from a finite set $\cM_h$.

At each timestep $h$, each agent $i$  needs to choose a communication action $\bM_{i,h}$ and a control action $\bU_{i,h}$,  based on her communication strategy $g_{i,h}^m$ and control strategy $g_{i,h}^a$, respectively. 
The class of \emph{open-loop} communication strategies is defined as $\cG_{i,h}^m:=\{g_{i,h}^m:\{\ast\}\rightarrow \cM_{i,h}\}$. We identify each such strategy with its single value $g_{i,h}^m(\ast)$, so $\bM_{i,h}$ is pre-selected. In this case, we consider the control strategies $g_{i,h}^a\in \cG_{i,h}^a:=\{g_{i,h}^a:\cI_{i,h^+}\rightarrow \cU_{i}\}$. Also, we may extend the setting to the class of \emph{closed-loop} communication strategies,  $\cG_{i,h}^m:=\{g_{i,h}^m:\cI_{i,h^-}\times \cM_{1:h-1}\rightarrow \cM_{i,h}\}$ where $\bM_{i,h}$ is thus chosen based on $\bI_{i,h^-}$ and $\bM_{1:h-1}$. Following \cite{LTC,ashutoshcommunicate1,ashutoshcommunicate2}, in this closed-loop setting, the communication action  $\bM_{h}$ is also assumed to be (implicitly) shared, since it can be inferred from \emph{what has been shared}.  
Hence, in this case, we consider the control strategies $g_{i,h}^a\in \cG_{i,h}^a:=\{g_{i,h}^a:\cI_{i,h^+}\times \cM_{1:h}\rightarrow \cU_{i}\}$ for generality. 
We will focus on the open-loop communication strategies throughout and extend our approach to the closed-loop ones in \S\ref{sec: closed-loop}. 
We use $g_h^m:=(g_{1,h}^m,\cdots,g_{n,h}^m)$ and $g_h^a:=(g_{1,h}^a,\cdots,g_{n,h}^a)$  to denote the joint communication and control strategies, respectively.  
We denote by $\cG_h^a,  \cG_h^m$ the spaces of these joint  strategies.

The total cost (objective) of all the agents is defined as
\begin{equation*}
    J_{\cD}(g_{1:H}^m,g_{1:H}^a):=\EE\left[\sum_{h=1}^H(c_h+\kappa_h)+c_{H+1}\bigggiven g_{1:H}^m,g_{1:H}^a\right],
\end{equation*}
which is the expected accumulated sum of the control and communication costs over $H$ timesteps. With this objective, we define the solution concept of \emph{team optimality} below. 
\begin{definition}
    \label{def: optimal_strategy_open-loop}
    Given a JCCO problem $\cD$,  a strategy $(g_{1:H}^{m,\ast},g_{1:H}^{a,\ast})$ is  \emph{team-optimal} if $\forall g_{1:H}^m\in\cG_{1:H}^m,g_{1:H}^a\in\cG_{1:H}^a$, 
    \begin{equation*}
        J_{\cD}(g_{1:H}^m,g_{1:H}^a)\ge  J_{\cD}(g_{1:H}^{m,\ast},g_{1:H}^{a,\ast}),
   \end{equation*}
   \normalsize
   with $g_{1:H}^{m,\ast}, g_{1:H}^{a,\ast}$ being referred to as the \emph{optimal communication} and \emph{control}  strategies, respectively.
\end{definition}

Note that the team-optimal strategies depend on the (communication) strategy space being used (i.e., open-loop versus closed-loop), which will be clear later from the context.

\subsection{Information structures}

Information structure is a well-studied notion in decentralized stochastic control that captures \emph{who knows what and when}  \cite{witsenhausen1971information,aditya2012information}.  For any JCCO  problem $\cD$, as the additional sharing via communication will also affect the IS and is not determined beforehand, when discussing the IS of $\cD$, we mean 
that of the problem \emph{without additional sharing}, which is essentially a decentralized LQG control problem, as formally defined below.   

\begin{definition}[Decentralized LQG Induced by JCCO]
For a JCCO problem $\cD$, we call $\check{\cD}$ the
\emph{decentralized LQG control problem induced by $\cD$}. The information
structure (IS) of $\cD$ is defined to be that of $\check{\cD}$.
\end{definition}

\begin{definition}[(Strictly) Partially Nested JCCO]
A JCCO problem $\cD$ is called \emph{partially nested} (PN) if its
induced decentralized LQG control problem $\check{\cD}$ is PN
\cite{ho1972team}. Namely, for any $i_1,i_2\in[n]$ and $h_1,h_2\in[H]$
with $h_1<h_2$, if $\bU_{i_1,h_1}$ influences $\bI_{i_2,h_2^-}$ in
$\check{\cD}$, then
$\bI_{i_1,h_1^-}\subseteq\bI_{i_2,h_2^-}$. It is called
\emph{strictly PN} if, under the same condition,
$\bI_{i_1,h_1^-}\cup\{\bU_{i_1,h_1}\}
\subseteq\bI_{i_2,h_2^-}$.
\end{definition}

\section{Linearity and Structural   Conditions for JCCO with Open-Loop Communication Strategies}
\label{sec: hardness and assumptions}

We now focus on the setting with \emph{open-loop}  communication strategies. Given any such fixed $g^m_{1:H}$ (and thus pre-selected $M_{1:H}$), the problem becomes a   decentralized LQG control one. 
However, it is known that the optimal control strategy of a decentralized LQG control problem may be \emph{nonlinear} in general \cite{witsenhausen1968counterexample}, making it computationally challenging. Fortunately, with partially nested ISs, the optimal control strategies can be linear 
\cite{ho1972team}\footnote{For completeness, we provide Corollary~\ref{cor:degenerate-PN-linearity} to establish the extension needed here for possibly singular Gaussian covariances and positive-semidefinite state and control cost matrices $\{Q_h^1\}_{h\in[H+1]}$ and $\{Q_h^2\}_{h\in[H]}$.}. Hence, it is natural to ask: when it comes to JCCO (with open-loop $g^m_{1:H}$), when would the optimal control strategy $g^a_{i,h}$ still be \emph{linear}, i.e., a linear function of the available information $\bI_{i,h^+}$, for all  $i\in[n],h\in[H]$?

To answer this question, first,  
if the IS of JCCO (i.e., the IS of baseline sharing) is not even PN, then in general it is hard for the IS of the decentralized LQG to be PN 
\emph{after} the optimal  additional sharing, especially with (high)  communication costs. This may thus further break the linearity of the optimal control strategies. 
To formalize this intuition, we introduce the following result, whose  proof, together with other omitted proofs for this section,  can be found in \S \ref{sec: appendix_hardness_result_proof}.  
\begin{lemma}
    There exists a JCCO problem $\cD$ with open-loop communication strategies that is not PN but satisfies Assumptions \ref{ass: evolution rule}, \ref{ass: useless action}, and \ref{ass: non-degeneracy}, such that the optimal control strategy is \emph{nonlinear}. 
    \label{lemma: nonqc}
\end{lemma}
{
\ifnum\myversion=1
\begin{proofsketch}
In the proof, we can construct a $\cD$ with high enough communication cost, but satisfies Assumptions \ref{ass: evolution rule}, \ref{ass: useless action}, and \ref{ass: non-degeneracy}. Then, the optimal communication strategy leads to no additional sharing, and $\cD$ reduces to Witsenhausen's counterexample \cite{witsenhausen1968counterexample}, whose  optimal strategy is nonlinear.
\end{proofsketch}
\fi
}

Due to Lemma \ref{lemma: nonqc},  we focus on PN JCCO problems, in order to obtain linearity of the optimal control strategies and thus computational tractability. However, even when the baseline sharing is PN, the additional sharing may \emph{break} the PN IS, and thus cause computational hardness again.    
Specifically, suppose agent $i_1$'s control action $\bU_{i_1,h_1}$  does not influence another agent $i_2$'s information $\bI_{i_2,h_2}$ in baseline sharing, but her control action $\bU_{i_1,h_1}$ is shared with others via additional sharing. Then, $\bU_{i_1,h_1}$ starts to influence $\bI_{i_2,h_2}$  and the PN IS may thus break. Here, the \emph{non-influence} of $\bU_{i_1,h_1}$ in baseline sharing may come from two cases: 1)  $\bU_{i_1,h_1}$ does not influence the underlying \emph{state} $\bX_{h_1+1}$; or 2) $\bU_{i_1,h_1}$ does influence the underlying state $\bX_{h_1+1}$, but the effect is \emph{not observed} by other agents. To avoid the impact of such non-influential actions, 
we make the following two assumptions, followed by justifications of their necessity regarding the \emph{existence/linearity} of the optimal control strategies.    

For case 1) above, 
we make the following assumption that such useless actions will \emph{not} be used or shared later.
In fact, for the linear systems considered in Equation~\eqref{equ:system_dynamics}, a useless control action will not even appear as a random variable in the system evolution, and it is natural to disregard it in the available information of later agents. Moreover, we highlight that such an assumption has been implicitly made in the common examples in decentralized stochastic control \cite{QCstatefeedback,liu2023tractable}.

\begin{assumption} \label{ass: useless action}
    For any $i\in[n],h\in[H]$, if $B_{i,h}=\mathbf{0}$, then $\bU_{i,h}\notin \bI_{j,(h')^-}, \bU_{i,h}\notin \bI_{j,(h')^+},{\forall h'\in[H]\text{ with }h'>h,\ j\in[n]}$. 
\end{assumption}

The significance of Assumption \ref{ass: useless action} in JCCO is reflected in two respects. First, it helps to preserve the PN IS (as to be shown formally in Theorem \ref{thm: preserve PN}). Second, without it, the optimal control strategy may either not exist or be nonlinear, as shown in the following lemma.

\begin{lemma}
    There exist PN JCCO problems $\cD_1,\cD_2$ with open-loop communication strategies that satisfy Assumptions \ref{ass: evolution rule}, \ref{ass: non-degeneracy} but not  Assumption \ref{ass: useless action},  such that the optimal control strategy of $\cD_1$ is \emph{nonlinear}, and the team-optimal strategy of $\cD_2$ \emph{does not exist}.
    \label{lemma: useless action}
\end{lemma}

{
\ifnum\myversion=1
\begin{proofsketch}
In the proof, we can construct JCCO problems  where there is no communication cost for sharing some specific \emph{useless action}. 
Then, the optimal communication strategy will share such an action, which will be designed to encode (and thus implicitly share) some observation, as a \emph{compression} of the observation. The nonlinearity result then stems from the fact that there is no \emph{linear} bijection between a high-dimensional vector space and a compressed lower-dimensional one. The non-existence result originates from the fact that there is no optimal compression function that minimizes the second moment of the compressed information (i.e., the useless action).
\end{proofsketch}
\fi
}

For case 2) above, the PN IS may break due to the \emph{non-informativeness} of agents' observations. We thus make the following assumption.

\begin{assumption}
\label{ass: non-degeneracy}
    For any $h\in[H-1], i\in[n]$, if $B_{i,h}\neq \mathbf{0}$, then $E_{-i,h+1}B_{i,h}\neq \mathbf{0}$,  
    where $E_{-i,h+1}:=\begin{bmatrix}
        E_{1,h+1}^\top&\cdots&E_{i-1,h+1}^\top&E_{i+1,h+1}^\top&\cdots&E_{n,h+1}^\top
    \end{bmatrix}^\top$.
\end{assumption}

Assumption \ref{ass: non-degeneracy} means that, for any $i\in[n]$ and $h\in[H-1]$, if $B_{i,h}\neq\mathbf{0}$, then $E_{j,h+1}B_{i,h}\neq\mathbf{0}$ for at least one $j\neq i$. Thus, it does not require $E_{-i,h+1}$ to have full rank or the other agents to jointly observe the whole state $\bX_{h+1}$. For instance, if $B_{i,h}=\mathbf{0}$, then no condition is imposed on $E_{-i,h+1}$.
As shown below, Assumption \ref{ass: non-degeneracy} is necessary to ensure the linearity of the optimal control strategy of $\cD$.

\begin{lemma}\label{lemma: nondegeneracy}
    There exists a PN JCCO problem $\cD$ with open-loop communication strategies  that satisfies  Assumptions \ref{ass: evolution rule}, \ref{ass: useless action} but not Assumption \ref{ass: non-degeneracy},  such that the optimal control strategy of $\cD$ is \emph{nonlinear}.
\end{lemma}

{\ifnum\myversion=1
\begin{proofsketch}
In the proof, we construct a JCCO $\cD$ based on Witsenhausen's counterexample, but change the observation of the second agent to be null, which can only happen without Assumption \ref{ass: non-degeneracy}, then $\cD$ has PN IS. However, we allow such an agent to obtain such an observation via additional sharing, making $\cD$ reduce to the Witsenhausen's counterexample, where the optimal control strategy is nonlinear.
\end{proofsketch}
\fi
}

We have thus justified the importance of the above assumptions and the partial nestedness of baseline sharing for JCCO problems, by showing that missing any one of them (while keeping others) may cause the nonlinearity (or even non-existence) of the optimal control strategy.  

On the other hand, for any fixed open-loop communication strategy $g_{1:H}^m\in \cG_{1:H}^m$, we denote by $\cD(g_{1:H}^m)$ the subproblem of finding the optimal \emph{control}  strategy with respect to $g_{1:H}^m$, where 
    we add $~\bar{\phantom{x}}~$ to the notation in  $\cD(g_{1:H}^m)$ to distinguish it from the notation in the original JCCO problem.   We remove $+$ in the timestep index for all the notation in $\cD(g_{1:H}^m)$ for convenience, as the IS is now \emph{fixed} and only the control problem will be analyzed. Then, we can show that if $\cD$ satisfies all the assumptions above, the optimal control strategy of $\cD$ is linear. Indeed,  under any fixed open-loop  $g_{1:H}^m\in \cG_{1:H}^m$, the subproblem $\cD(g_{1:H}^m)$ can be shown to be a \emph{decentralized LQG} problem with a partially nested information structure (see \S \ref{sec: DecLQG} for a formal definition).
\begin{theorem}
    Consider a PN JCCO problem $\cD$ that satisfies Assumptions \ref{ass: evolution rule}, \ref{ass: useless action}, and \ref{ass: non-degeneracy}. For any open-loop communication strategy $g_{1:H}^m\in \cG_{1:H}^m$,  the subproblem $\cD(g_{1:H}^m)$ is a decentralized LQG control  problem with a PN IS, and its information evolution rules satisfy:  for each   $i\in[n]$ and $h\in[H]$, 
        \begin{align*}
            &\bar{\bC}_{h}=\bar{\bC}_{h-1}\cup\bar{\bZ}_h,~~~ \bar{\bZ}_h=\bar{\chi}_h(\bar{\bP}_{h-1},\bar{\bU}_{h-1},\bar{\bY}_{h}), ~~~\bar{\bY}_{i,h}\in \bar{\bI}_{i,h},\notag\\
            &\bar{\bP}_{i,h}=\bar{\zeta}_{i,h}(\bar{\bP}_{i,h-1},\bar{\bU}_{i,h-1},\bar{\bY}_{i,h}),\quad\bar{\bI}_{i,h-1
            }\subseteq \bar{\bI}_{i,h}, 
        \end{align*}
    \normalsize
    for some projection functions $\bar{\chi}_h$ and $\{\bar{\zeta}_{i,h}\}_{i\in[n]}$. 
    Furthermore, there exists  an   optimal joint  strategy $(g_{1:H}^{m,\ast},g_{1:H}^{a,\ast})$ of $\cD$  with open-loop communication strategies such that $g_{1:H}^{a,\ast}$ is a collection of \emph{linear}  functions. 
    \label{thm: preserve PN}
\end{theorem}

Theorem \ref{thm: preserve PN} provides a way to solve $\cD$ by finding the optimal control strategy $g_{1:H}^{a,\ast}$ with respect to any fixed open-loop communication strategy $g_{1:H}^m$, i.e., solving 
$\cD(g_{1:H}^m)$, and then optimizing over the open-loop communication strategies $g_{1:H}^m$. We will thus  investigate how to solve the problem $\cD(g_{1:H}^m)$ next.

\ifnum\myversion=1
\begin{proofsketch}
    For any two agents $(i_1,h_1), (i_2,h_2), h_1\le h_2$ in $\cD(g_{1:H}^m)$. If $\overline{B}_{i_1,h_1}=\mathbf{0}$, then agent $(i_1,h_1)$ cannot influence the states. Also, due to Assumption \ref{ass: useless action}, $\overline{\bU}_{i_1,h_1}$ cannot lie in any information. Therefore, agent $(i_1,h_1)$ cannot influence agent $(i_2,h_2)$. If $B_{i_1,h_1}\neq\mathbf{0}$, then from Assumptions \ref{ass: non-degeneracy} and \ref{ass: evolution rule}(e),  agent $(i_1,h_1)$ will influence state $\bX_{h_1+1}$ and further influence $\bI_{j,h_1+1}$ in $\cD$ for some $j\neq i_1$ even there is no additional sharing. Partial nestedness gives $\bI_{i_1,h_1^-}\subseteq\bI_{j,h_1+1}$. Since $j\neq i_1$, every label private to agent $i_1$ can enter agent $j$'s information only through common information due to Assumption \ref{ass: evolution rule}, then it holds that $\bI_{i_1,h_1^-}\subseteq \bC_{h_1+1}$. Moreover, every label in $\overline{\bI}_{i_1,h_1}\backslash\bI_{i_1,h_1^-}$ was revealed through additional sharing and is common. Then, $\overline{\bI}_{i_1,h_1}\subseteq \overline{\bC}_{h_1+1}\subseteq\overline{\bI}_{i_2,h_2}$. 
    Therefore, $\cD(g_{1:H}^m)$ has PN IS.

    Also, one can verify that $\cD(g_{1:H}^m)$ satisfies the information evolution rule with projection functions defined by
    \begin{align*}
        &\bar{\chi}_h(\bar{\bP}_{h-1},\bar{\bU}_{h-1},\bar{\bY}_{h}):=\chi_h(\bar{\bP}_{h-1},\bar{\bU}_{h-1},\bar{\bY}_{h})\cup\\
        &\qquad\qquad\phi_h(M_h,\zeta_h(\bar{\bP}_{h-1},\bar{\bU}_{h-1},\bar{\bY}_{h})),\\
        &\bar{\zeta}_{i,h}(\bar{\bP}_{i,h-1},\bar{\bU}_{i,h-1},\bar{\bY}_{i,h}):=\zeta_{i,h}(\bar{\bP}_{i,h-1},\bar{\bU}_{i,h-1},\bar{\bY}_{i,h})\\
        &\qquad\qquad\backslash\phi_{i,h}(M_{i,h},\zeta_{i,h}(\bar{\bP}_{i,h-1},\bar{\bU}_{i,h-1},\bar{\bY}_{i,h})).
        \end{align*}
\end{proofsketch}
\fi

\section{Dynamic Programming for   Decentralized LQG with Partially Nested Information Structures} 
\label{sec: open_loop}

{Since $\cD(g_{1:H}^m)$ is PN under the assumptions in Theorem \ref{thm: preserve PN}, it now suffices to develop an algorithm that can solve PN decentralized LQG control problems. For notational convenience, in this section, we will omit the superscript $a$ for the control strategies $\bar g^a_{1:H}$ for $\cD(g_{1:H}^m)$. 

The algorithm to be introduced can be viewed as a dynamic-programming approach for solving a class of PN  decentralized LQG problems with \emph{output} feedback  under the {common-information-based} framework \cite{ashutosh2013team,ashutosh2013game}, 
which thus advances the results in \cite{AA} and may 
be of independent interest.} Detailed proofs of results in this section are deferred to \S \ref{sec: appendix_open_loop_proof}. 
Specifically, under PN IS, the known dynamic-programming-based approaches and sufficient statistics in \cite{QCstatefeedback,nayyar2015structural} do not apply here, as they focused on the special settings with \emph{factorized} states, and with either \emph{state-feedback}  \cite{QCstatefeedback} or \emph{multi-tree}  coupling and
communication graphs \cite{nayyar2015structural}.

{On the other hand, under the {common-information-based}  framework \cite{ashutosh2013team}, 
\cite{AA} showed that} {if agents are \emph{restricted} to using \emph{linear} control strategies, then the decentralized control problem can be reformulated as a \emph{centralized LQG} one when \emph{fixing} the private-information component of the strategies. 
However, it is not clear how to further \emph{optimize} the private-information component, which induces a  non-convex
optimization problem in general \cite{AA}. 
Furthermore, 
\cite{CIBLQgames} showed  that, 
if additionally the CIB beliefs from \cite{ashutosh2013team} are  \emph{strategy independent} 
\cite[Section 2.4]{CIBLQgames}, then the CIB beliefs can be further simplified to a finite-dimensional conditional mean of the state and private information in this linear-quadratic setting\footnote{Note that \cite{CIBLQgames} focused on a \emph{game} setting with strategic agents, instead of the \emph{team} setting in \cite{AA,QCstatefeedback} and the present paper.}.}

Inspired by \cite{CIBLQgames}, and to address the non-convexity issue in \cite{AA}, we exploit the PN IS to obtain a finite-dimensional, conditional-mean-based sufficient statistic, by establishing a new connection between the PN IS  and the SI-CIB condition. 

To this end, we adapt our technique in \cite{LTC} to this linear-quadratic setting, while addressing several fundamentally different  challenges. First, we expand the PN IS into a \emph{strictly} PN one, where the \emph{actions} of agents who affect some other agent will also be known by the affected agent (in addition to their \emph{information}). More formally, we construct a new problem $\tilde{\cD}(g_{1:H}^m)$, whose elements are almost identical to those in ${\cD}(g_{1:H}^m)$, except that the agents' information is now expanded as follows: for each  $h\in[H]$
\begin{equation}
\begin{aligned}
\tilde{\bC}_h=\bar{\bC}_h\cup \{\bar{\bU}_{i,t}\given i\in[n],t<h, \bar{\bI}_{i,t}\subseteq \bar{\bC}_h, \bar{B}_{i,t}\neq \mathbf{0}\},\qquad 
    \forall j\in[n], \tilde{\bP}_{j,h}=\bar{\bP}_{j,h}\backslash\{\bar{\bU}_{j,t}\given t<h, \bar{\bI}_{j,t}\subseteq \bar{\bC}_h\},
    \label{eq: expansion}
\end{aligned}
\end{equation}
where we add $~\tilde{}~$ to the notation in $\tilde{\cD}(g_{1:H}^m)$, and the incremental common information is defined by $\tilde{\bZ}_h:=\tilde{\bC}_h\backslash\tilde{\bC}_{h-1},\forall h\in[H]$.  
Through such an expansion, we can leverage the new problem $\tilde{\cD}(g_{1:H}^m)$ to solve the original $\cD(g_{1:H}^m)$, as formalized in the following lemma.

\begin{lemma}
    Let $\cD$ be PN and satisfy  Assumptions \ref{ass: evolution rule}, \ref{ass: useless action}, and \ref{ass: non-degeneracy}, and  for any fixed open-loop $g_{1:H}^m\in\cG_{1:H}^m$, let $\cD(g_{1:H}^m)$ and $\tilde{\cD}(g_{1:H}^m)$ be the decentralized LQG control problems   constructed above. Then, there exists a function $\varphi$ such that for any linear optimal control strategy $\tilde{g}_{1:H}^{\ast}$ of the problem $\tilde{\cD}(g_{1:H}^m)$, $\bar{g}_{1:H}^{\ast}=\varphi(\tilde{g}_{1:H}^{\ast},\cD(g_{1:H}^m))$ is a linear optimal control strategy of $\cD(g_{1:H}^m)$, and  
    $J_{\cD(g_{1:H}^m)}(\bar{g}_{1:H}^{\ast})=J_{\tilde{\cD}(g_{1:H}^m)}(\tilde{g}_{1:H}^{\ast})$, where $J_{\cD(g_{1:H}^m)}(\bar{g}_{1:H}^{\ast})$ is the expected accumulated cost of $\cD(g_{1:H}^m)$ under the strategy $\bar{g}_{1:H}^{\ast}$, and $J_{\tilde{\cD}(g_{1:H}^m)}(\tilde{g}_{1:H}^{\ast})$ is defined similarly (cf. \S\ref{sec: DecLQG} for the formal definitions).
    \label{lemma: equivalence of PN and sPN}
\end{lemma}

Before proceeding further, we introduce some additional notation in $\tilde{\cD}(g_{1:H}^m)$. 
For any $h\in[H]$, we define  $\tilde{\bS}_h:=\begin{bmatrix}\tilde{\bX}_h^\top& \tilde{\bP}_{1,h}^\top& \cdots& \tilde{\bP}_{n,h}^\top
\end{bmatrix}^\top$ and use $\tilde{\cS}_h:=\tilde{\cX}\times \tilde{\cP}_{1,h}\times\cdots\times \tilde{\cP}_{n,h}$ to denote the space of $\tilde{\bS}_h$. Then, we use $\tilde{\bB}_h\in\cP(\tilde{\cS}_h)$ with $\tilde{\bB}_h(\db\tilde{\bS}_h):=\PP^{\tilde{\cD}(g_{1:H}^m)}(\db\tilde{\bS}_h\given \tilde{\bC}_h,\tilde{g}_{1:h-1})$ to denote the conditional probability measure
of $\tilde{\bS}_h$ given the common information $\tilde{\bC}_h$ and the past strategies $\tilde{g}_{1:h-1}$. Hence, we have $\tilde{\bB}_h=\Pi_h(\tilde{\bC}_h,\tilde{g}_{1:h-1})$ for some functional $\Pi_h$. The benefit of solving $\tilde{\cD}(g_{1:H}^m)$ instead of $\cD(g_{1:H}^m)$ is {that the former satisfies the strategy-independence condition for CIB beliefs under the stated assumptions, as shown below.} 
\begin{theorem}
    \label{thm: expansion}
    If $\cD$ is PN and satisfies Assumptions  \ref{ass: evolution rule}, \ref{ass: useless action}, and \ref{ass: non-degeneracy}, then for any fixed open-loop $g_{1:H}^m\in \cG_{1:H}^m$, $\tilde{\cD}(g_{1:H}^m)$ constructed as above is (strictly) PN and satisfies the \emph{SI-CIB condition}: for any $h\in[H]$, any realization $\tilde{C}_h\in\tilde{\cC}_h$, and any two control strategies $\tilde{g}_{1:h-1}, \tilde{g}_{1:h-1}'$ that can reach $\tilde{C}_h$, it holds that $\Pi_h(\tilde{C}_h,\tilde{g}_{1:h-1})=\Pi_h(\tilde{C}_h,\tilde{g}_{1:h-1}')$.
    Meanwhile, $\tilde{\cD}(g_{1:H}^m)$ has the following information evolution rules: for each  $i\in[n]$ and $h\in[H]$, 
    \begin{align}
        \label{equ: DecLQG_info_evolution}
            \tilde{\bC}_{h}&=\tilde{\bC}_{h-1}\cup\tilde{\bZ}_h, \tilde{\bZ}_h=\tilde{\chi}_h(\tilde{\bP}_{h-1},\tilde{\bU}_{h-1},\tilde{\bY}_{h}), \\
            \tilde{\bP}_{i,h}&=\tilde{\zeta}_{i,h}(\tilde{\bP}_{i,h-1},\tilde{\bU}_{i,h-1},\tilde{\bY}_{i,h}),\tilde{\bY}_{i,h}\in \tilde{\bI}_{i,h},\tilde{\bI}_{i,h-1
            }\subseteq \tilde{\bI}_{i,h}, \notag
    \end{align}
    for some projection functions $\tilde{\chi}_h$ and $\{\tilde{\zeta}_{i,h}\}_{i\in[n]}$. 
\end{theorem}

Furthermore, we define the \emph{prescription} \cite{ashutosh2013team} of each agent $i$ at timestep $h$ as $\tilde{\gamma}_{i,h}:\tilde{\cP}_{i,h}\rightarrow \tilde{\cU}_{i}$. 
We denote by $\tilde{\gamma}_{h}=(\tilde{\gamma}_{1,h}, \cdots, \tilde{\gamma}_{n,h})$ the joint prescription  of all the agents,  and by $\tilde{\Gamma}_{i,h}$ and $\tilde{\Gamma}_h$ the spaces of agent $i$'s and the joint prescriptions at timestep $h$, respectively. For notational convenience, we write $\tilde{\gamma}_h(\tilde{\bP}_h)=\begin{bmatrix}
    \tilde{\gamma}_{1,h}(\tilde{\bP}_{1,h})^\top&\cdots&\tilde{\gamma}_{n,h}(\tilde{\bP}_{n,h})^\top
\end{bmatrix}^\top$.

Following \cite{CIBLQgames}, under the SI-CIB condition, the belief $\tilde{\bB}_h$ in $\tilde{\cD}(g_{1:H}^m)$ is a Gaussian distribution with mean $\tilde{\bTheta}_h=\begin{bmatrix}
    \EE[\tilde{\bX}_h\given \tilde{\bC}_h]^\top&\EE[\tilde{\bP}_{1,h}\given \tilde{\bC}_h]^\top& \cdots& \EE[\tilde{\bP}_{n,h}\given \tilde{\bC}_h]^\top
\end{bmatrix}^\top$ and covariance $\tilde{\Sigma}_h$ that evolve over time, as formalized below.

\begin{lemma}[Adapted from Lemmas 2.2 \& 2.3 in \cite{CIBLQgames}] 
    Consider a decentralized LQG problem $\tilde{\cD}(g_{1:H}^m)$ constructed by the strict expansion in Theorem~\ref{thm: expansion}, that has information evolution rules in Equation  \eqref{equ: DecLQG_info_evolution} and satisfies the SI-CIB condition. Then the CIB belief $\tilde{\bB}_h$ admits a 
    Gaussian distribution $\cN(\tilde{\bTheta}_h,\tilde{\Sigma}_h)$, and the conditional mean $\tilde{\bTheta}_h$ and the conditional covariance $\tilde{\Sigma}_h$ of 
    $\tilde{\bB}_h$ evolve as follows: for any $h\in[2:H]$,
    \begin{align*}
        \tilde{\bTheta}_h=\tilde{\Psi}_h^1(\tilde{\bTheta}_{h-1},\tilde{\bZ}_h),\qquad  \tilde{\Sigma}_h=\tilde{\Psi}_h^2(\tilde{\Sigma}_{h-1}), 
    \end{align*}
    {where $\tilde{\Psi}_h^1$ is a fixed linear function, $\tilde{\Psi}_h^2$ is a fixed deterministic function, and neither depends on the strategies $\tilde g_{1:h-1}$. In particular, $\tilde\Sigma_h$ is deterministic and strategy-independent.}  
    \label{lemma: SI-conditional-mean-covariance}
\end{lemma}

According to Lemma \ref{lemma: SI-conditional-mean-covariance}, we know that $\tilde{\bTheta}_h$ depends only on $\{\tilde{\bZ}_t\}_{t=1}^h$. Hence, 
there exists a fixed linear transformation $\tilde{\Psi}_h^3$ such that  $\tilde{\bTheta}_h=\tilde{\Psi}_h^3(\tilde{\bC}_h)$. 
Then, following 
\cite{CIBLQgames}, we can construct a new problem $\tilde{\cD}^\ddag(g_{1:H}^m)$: At timestep $h\in[H]$, the state of $\tilde{\cD}^\ddag(g_{1:H}^m)$ is defined as $\tilde{\bTheta}_h\in\tilde{\cS}_h$ and the action is defined as $\tilde{\gamma}_h\in\tilde{\Gamma}_h$. The stage cost is defined as
$c_h^\ddag(\tilde{\bTheta}_h,\tilde{\gamma}_h):=\EE\Big[\tilde{\bX}_h^\top\tilde{Q}_h^1\tilde{\bX}_h+\tilde{\gamma}_h(\tilde{\bP}_h)^\top\tilde{Q}_h^2\tilde{\gamma}_h(\tilde{\bP}_h)+\mathds{1}[h=H]\tilde{\bX}_{H+1}^{\top}\tilde Q_{H+1}^1\tilde{\bX}_{H+1}\Big],
$
where the expectation is taken over $[\tilde{\bX}_h^\top~\tilde{\bP}_h^\top]^\top\sim\cN(\tilde{\bTheta}_h,\tilde\Sigma_h)$, and when $h=H$, also over the final state $\tilde{\bX}_{H+1}$ generated by the system dynamics.
The admissible strategy at timestep $h$ in $\tilde{\cD}^\ddag(g_{1:H}^m)$ is denoted by $g^\ddag_h:\tilde{\cC}_h\rightarrow \tilde{\Gamma}_h$, and $\cG_h^\ddag$ is the set of all admissible strategies.  The objective of $\tilde{\cD}^\ddag(g_{1:H}^m)$ is then {defined by}
\begin{align*}
J_{\tilde{\cD}^\ddag(g_{1:H}^m)}(g_{1:H}^\ddag):=\EE\left[\sum_{h=1}^Hc_h^\ddag(\tilde{\bTheta}_h,\tilde{\gamma}_h)\right]. 
\end{align*}  

As shown in the following theorem, the constructed problem $\tilde{\cD}^\ddag(g_{1:H}^m)$ is a Markov decision problem, and the optimal strategy of $\tilde{\cD}^\ddag(g_{1:H}^m)$ can be used to compute the team-optimal strategy of $\tilde{\cD}(g_{1:H}^m)$, and vice versa. 
\begin{theorem}
    \label{thm: virtual MDP and equivalence.}
    Let $\tilde{\cD}^\ddag(g_{1:H}^m)$ be the problem constructed from $\tilde{\cD}(g_{1:H}^m)$  satisfying the SI-CIB condition. Then, $\tilde{\cD}^\ddag(g_{1:H}^m)$ is a Markov decision problem with horizon $H$, state $\{\tilde{\bTheta}_h\}_{h\in[H]}$, control action $\{\tilde{\gamma}_h\}_{h\in[H]}$, and cost function $\{c_h^\ddag\}_{h\in[H]}$. Moreover, there exist bijections   $\varsigma_h:\tilde{\cG}_{h}\rightarrow \cG_h^\ddag$ for all $h\in[H]$ that can be defined as follows: for any strategy $\tilde{g}_h\in\tilde{\cG}_h$ and realization of common information $\tilde{C}_h\in\tilde{\cC}_h$, $\varsigma_h(\tilde{g}_h)(\tilde{C}_h)(\cdot)=\tilde{g}_h(\tilde{C}_h,\cdot)$. {These bijections also satisfy the following equality for every} strategy $\tilde{g}_{1:H}\in\tilde{\cG}_{1:H}$: $J_{\tilde{\cD}(g_{1:H}^m)}(\tilde{g}_{1:H})=J_{\tilde{\cD}^\ddag(g_{1:H}^m)}(g^\ddag_{1:H})$, where $g_h^\ddag=\varsigma_h(\tilde{g}_h), \forall h\in[H]$. 
\end{theorem}

By Theorem \ref{thm: virtual MDP and equivalence.} {and Lemma \ref{thm: linear in M} below, an optimal strategy exists and may be chosen \emph{Markovian}, with $g_h^{\ddag,\ast}$ depending on the common record only through $\tilde{\bTheta}_h$. In particular,}
for any $h\in[H]$, we can write $g_h^{\ddag,\ast}$ as $g_h^{\ddag,\ast}(\tilde{\bC}_h)=\sG_h^\ast(\tilde{\Psi}_h^3(\tilde{\bC}_h))$ for some function $\sG_h^\ast$. Furthermore, if we define the strategy $\tilde{g}_h^\ast=\varsigma_h^{-1}(g_h^{\ddag,\ast}), \forall h\in[H]$, then $\tilde{g}_{1:H}^\ast$ is an optimal strategy of $\tilde{\cD}(g_{1:H}^m)$.  Therefore, there is no loss of optimality in restricting attention to $\tilde{g}_{i,h}$ that chooses action $\tilde{\bU}_{i,h}$ depending only on $\tilde{\bTheta}_h$ and $\tilde{\bP}_{i,h}$ for any $h\in[H],i\in[n]$ in solving the problem $\tilde{\cD}(g_{1:H}^m)$. We categorize this type of strategy as \emph{common-information-based Markovian (CIB-Markovian) strategy} \cite{ashutosh2013game,CIBLQgames}.  Moreover, we can define the associated value functions for every $h\in[H]$ as
\begin{align*}
\tilde{V}_h(\tilde{\bTheta}_h):=\min_{\tilde{g}_{h:H}}\EE\left[\sum_{t=h}^H (\tilde{\bX}_t^\top \tilde{Q}_t^1\tilde{\bX}_t+\tilde{\bU}_t^\top \tilde{Q}_t^2\tilde{\bU}_t)+\tilde{\bX}_{H+1}^\top \tilde{Q}_{H+1}^1\tilde{\bX}_{H+1}\Biggiven \tilde{\bTheta}_h, \tilde{g}_{h:H}\right],
\end{align*}
which satisfy the following Bellman Equations:

\begin{equation}
\tilde{V}_h(\tilde{\bTheta}_h)=
\begin{cases}
\displaystyle\min_{\tilde{\gamma}_h\in\tilde{\Gamma}_h}\left\{c_h^\ddag(\tilde{\bTheta}_h,\tilde\gamma_h)+\EE[\tilde V_{h+1}(\tilde{\bTheta}_{h+1})\given\tilde{\bTheta}_h,\tilde\gamma_h]\right\},&h\in[H-1],\\[2mm]
\displaystyle\min_{\tilde{\gamma}_H\in\tilde{\Gamma}_H}c_H^\ddag(\tilde{\bTheta}_H,\tilde\gamma_H),&h=H.
\end{cases}
\label{equ: centralized Bellman Equation}
\end{equation}

Then, we can derive an approach to obtain the optimal strategy $\tilde{g}_{1:H}^\ast$ by dynamic programming: for all $h\in[H]$ and $\tilde{C}_h\in\tilde{\cC}_h$, $\tilde{g}_h^\ast(\tilde{C}_h,\cdot)$ minimizes the right-hand side of Equation   \eqref{equ: centralized Bellman Equation} at $\tilde{\bTheta}_h=\tilde\Psi_h^3(\tilde C_h)$. However, {there are uncountably many possible $\tilde{C}_h$ (or $\tilde{\Theta}_h$), making this approach computationally intractable without a closed-form solution.}

For computational tractability, we leverage the \emph{linearity} of the optimal strategy. Although it is known  that 
there exists a \emph{linear}  optimal strategy under PN IS (cf. Corollary \ref{cor:degenerate-PN-linearity}), and if it exists, the optimal strategy can be \emph{CIB-Markovian}  under the SI-CIB condition (cf. Theorem \ref{thm: virtual MDP and equivalence.}),  
the existence of an optimal strategy with both properties remains unclear. We show the existence of such an optimal strategy below.

\begin{lemma}
    Suppose $\tilde{\cD}(g_{1:H}^m)$ is a decentralized LQG control problem constructed from a PN JCCO problem $\cD$ that satisfies Assumptions \ref{ass: evolution rule}, \ref{ass: useless action}, and \ref{ass: non-degeneracy}{. Then,} there exists an optimal linear strategy $\tilde{g}_{1:H}^\ast$, and matrices $\{\tilde{\cE}_{i,h}^\ast\}_{i\in[n], h\in[H]}, \{\tilde{\cF}_{i,h}^\ast\}_{i\in[n], h\in[H]}$ such that the $\tilde{\gamma}_{i,h}^\ast(\cdot)$,  defined as $\tilde{\gamma}_{i,h}^\ast(\cdot):=\tilde{\cE}_{i,h}^\ast\tilde{\Theta}_h+\tilde{\cF}_{i,h}^\ast\cdot$,  
    minimizes the right-hand side of Equation~\eqref{equ: centralized Bellman Equation} at every CIB mean $\tilde\Theta_h$ reachable under a preceding strategy, and $\tilde{g}_{i,h}^\ast(\tilde{C}_h,\tilde{P}_{i,h})=\tilde{\cE}_{i,h}^\ast\tilde{\Psi}_h^3(\tilde{C}_h)+\tilde{\cF}_{i,h}^\ast\tilde{P}_{i,h}$ for any $i\in[n],h\in[H], \tilde{C}_h\in\tilde{\cC}_h,\tilde{P}_{i,h}\in \tilde{\cP}_{i,h}$. 
\label{thm: linear in M}
\end{lemma}
    Lemma \ref{thm: linear in M} shows that for each agent $i$ at timestep $h$, it suffices to compute the optimal strategy that is linear in $\tilde{\bTheta}_h$ and $\tilde{\bP}_{i,h}$. Then, for any $i\in[n],h\in[H]$, we can parameterize the strategy $\tilde{g}_{i,h}$ by two matrices $\tilde{\cE}_{i,h},\tilde{\cF}_{i,h}$, where $\tilde{g}_{i,h}(\tilde{\bC}_h,\tilde{\bP}_{i,h})=\tilde{\cE}_{i,h}\tilde{\Psi}_h^3(\tilde{\bC}_h)+\tilde{\cF}_{i,h}\tilde{\bP}_{i,h}$. Define $\tilde{\cE}_h^\ast:=\begin{bmatrix}
        (\tilde{\cE}_{1,h}^{\ast})^\top&\cdots&(\tilde{\cE}_{n,h}^{\ast})^\top
    \end{bmatrix}^\top$ and $\tilde{\cF}_h^\ast:=\diag(\tilde{\cF}_{1,h}^\ast,\cdots, \tilde{\cF}_{n,h}^\ast)$. With these definitions, we can find the optimal control strategy $\tilde{g}_{h}^\ast$ by solving for  the optimal matrices:
        \begin{align}
(\tilde{\cE}_h^\ast,\tilde{\cF}_h^\ast)\in\argmin_{\tilde{\cE}_h,\tilde{\cF}_h}
\begin{cases}
c_h^\ddag(\tilde\Theta_h,\tilde\gamma_h)+\EE[\tilde V_{h+1}(\tilde\bTheta_{h+1})\given\tilde\bTheta_h=\tilde\Theta_h,\tilde\cE_h,\tilde\cF_h],&h\in[H-1],\\
c_H^\ddag(\tilde\Theta_H,\tilde\gamma_H),&h=H.
\end{cases}\label{eq: objective_E_F}
        \end{align}
    
    As shown in the following theorem, such optimal matrices can be computed through a set of Riccati Equations with better computational tractability.  
\begin{theorem}
        \label{thm: dp_in_decLQG}
        Suppose $\tilde{\cD}(g_{1:H}^m)$ is a decentralized LQG control problem constructed from a PN JCCO problem $\cD$ satisfying Assumptions \ref{ass: evolution rule}, \ref{ass: useless action}, and \ref{ass: non-degeneracy}. For $h\in[H-1]$, the conditional mean evolves as
        $\tilde{\bTheta}_{h+1}=K_{h+1}^1\tilde{\bTheta}_h+K_{h+1}^2\tilde{\bP}_h+K_{h+1}^3\tilde{\bU}_h+K_{h+1}^4\tilde{\bY}_{h+1}$,
        where the matrices $\{K_{h+1}^k\}_{k\in[4]}$ can be computed by a Kalman filter. Moreover, the value function has a quadratic form $\tilde V_h(\tilde{\bTheta}_h)=\tilde{\bTheta}_h^\top\tilde R_h\tilde{\bTheta}_h+\tilde c_h$ for some matrix $\tilde{R}_h$ and constant $\tilde{c}_h$  that do not depend on $\tilde{\bTheta}_h$,  and $\tilde{R}_h$ can be computed through the following Riccati Equations: 
        \begin{equation}
            \begin{aligned}
            &\tilde{R}_h=(\hat{\cE}_h^{\ast})^{\top}\tilde{Q}_h^2\hat{\cE}_h^\ast+(\tilde{L}_h^1\hat{\cE}_h^\ast+\tilde{L}_h^2)^\top \tilde{R}_{h+1}(\tilde{L}_h^1\hat{\cE}_h^\ast+\tilde{L}_h^2)+\II_{x,h}^\top\tilde{Q}_h^1\II_{x,h},\qquad  \tilde{L}_h^1=K_{h+1}^3+K_{h+1}^4\tilde{E}_{h+1}\tilde{B}_h,\\
            &\tilde{L}_h^2=K_{h+1}^1+K_{h+1}^2\II_{p,h}+K_{h+1}^4\tilde{E}_{h+1}\tilde{A}_h\II_{x,h},\qquad \hat{\cE}_h^\ast=-(\tilde{Q}_h^2+(\tilde{L}_h^1)^\top\tilde{R}_{h+1}\tilde{L}_h^1)^\dag(\tilde{L}_h^{1})^\top\tilde{R}_{h+1}\tilde{L}_h^2,\\
            &\tilde L_H^1=\tilde L_H^3:=\tilde B_H,\qquad \tilde L_H^2=\tilde L_H^4:=\tilde A_H\II_{x,H},\qquad\tilde R_{H+1}=\tilde Q_{H+1}^1,
            \end{aligned}\label{eq: centralized Riccati}
            \end{equation}
        where  $\II_{x,h}$ and $\II_{p,h}$ project vectors in $\tilde\cS_h=\tilde\cX\times\tilde\cP_h$ onto $\tilde\cX$ and $\tilde\cP_h$, respectively. 
        Then, the optimal matrices $\tilde{\cE}_h,\{\tilde{\cF}_{i,h}\}_{i\in[n]}$ can be computed correspondingly as follows:
        \begin{equation}
        \begin{aligned}
            &\tilde{\cF}_{h}^\ast\in\argmin_{\tilde{\cF}_h=\diag(\tilde{\cF}_{1,h},\cdots,\tilde{\cF}_{n,h})}\tr(\tilde{\cF}_h\II_{p,h}\tilde{\Sigma}_h\II_{p,h}^\top \tilde{\cF}_h^\top\tilde{Q}_h^2)+\tr((\tilde{L}_h^3\tilde{\cF}_h\II_{p,h}+\tilde{L}_h^4)\tilde{\Sigma}_h(\tilde{L}_h^3\tilde{\cF}_h\II_{p,h}+\tilde{L}_h^4)^\top\tilde{R}_{h+1}),\\
            &\tilde{L}_h^3=\tilde{L}_h^1, \qquad \tilde{L}_h^4=K_{h+1}^2\II_{p,h}+K_{h+1}^4\tilde{E}_{h+1}\tilde{A}_h\II_{x,h},\qquad \tilde{\cE}_h^\ast=\hat{\cE}_h^\ast-\diag(\tilde{\cF}_{1,h}^\ast,\tilde{\cF}_{2,h}^\ast,\cdots,\tilde{\cF}_{n,h}^\ast)\II_{p,h}.
            \label{eq: optimize F} 
        \end{aligned}
        \end{equation}
    \end{theorem}

        \noindent\textbf{Solving JCCO with open-loop $g^m_{1:H}$:} 
        Now we are ready to solve the problem $\cD(g^m_{1:H})$. First, from Theorem \ref{thm: dp_in_decLQG}, we can solve for $\tilde{g}_{1:H}^\ast$ as $\tilde{g}_{i,h}^\ast(\tilde{\bC}_h,\tilde{\bP}_{i,h})=\tilde{\cE}_{i,h}^\ast\tilde{\Psi}_h^3(\tilde{\bC}_h)+\tilde{\cF}_{i,h}^\ast\tilde{\bP}_{i,h}$ for any $i\in[n],h\in[H]$, and obtain the optimal value of $\min_{\tilde{g}_{1:H}}J_{\tilde{\cD}(g_{1:H}^m)}(\tilde{g}_{1:H})=\EE[\tilde{V}_1(\tilde{\bTheta}_1)]:=\tilde{c}_0$. Note that  $\tilde{\bTheta}_1=\tilde{\Psi}_1^3(\tilde{\bC}_1)=K_1^1\tilde{\bC}_1, \tilde{\bC}_1=\tilde{\bZ}_1=\tilde{\chi}_1(\tilde{\bY}_1)=K_1^2\tilde{\bY}_1, \tilde{\bY}_1=\tilde{E}_1\tilde{\bX}_1+\tilde{\bW}_{1:n,1}$, where $\tilde{\bX}_1\sim \cN(\mathbf{0},\Sigma_1)$ for some matrices $K_1^1,K_1^2$. Therefore, we can write         $\tilde{c}_0=\tr(\tilde{E}_1^\top (K_1^{2})^\top (K_1^{1})^\top\tilde{R}_1K_1^1K_1^2\tilde{E}_1\Sigma_1)+\tr((K_1^{2})^\top (K_1^{1})^\top\tilde{R}_1K_1^1K_1^2\diag(\Sigma_{1,1},\cdots,\Sigma_{n,1}))+\tilde{c}_1$. Then,  by Lemma \ref{lemma: equivalence of PN and sPN}, we can obtain an optimal control strategy of $\cD(g^m_{1:H})$ as  $\bar{g}_{1:H}^{\ast}=\varphi(\tilde{g}_{1:H}^{\ast},\cD(g_{1:H}^m))$ such that $J_{\cD(g_{1:H}^m)}(\bar{g}_{1:H}^{\ast})=J_{\tilde{\cD}(g_{1:H}^m)}(\tilde{g}_{1:H}^{\ast})$. If we assign $g_{1:H}^{\ast}=\bar{g}_{1:H}^{\ast}$, then $g_{1:H}^{\ast}$ is an optimal control strategy with respect to the communication strategy $g_{1:H}^m$ in $\cD$. Let $M_{1:H}$ be the communication actions chosen by $g_{1:H}^m$ such that $M_h=g^m_h${. Then,} $J_\cD(g_{1:H}^m, g_{1:H}^{\ast})=J_{\tilde{\cD}(g_{1:H}^m)}(\tilde{g}_{1:H}^{\ast})+\sum_{h=1}^H \cK_h(g^m_h)$. Therefore, we can finally solve the JCCO problem by solving $g_{1:H}^{m,\ast}\in\argmin_{g_{1:H}^m}J_{\tilde{\cD}(g_{1:H}^m)}(\tilde{g}_{1:H}^{\ast})+\sum_{h=1}^H \cK_h(g^m_h)$.

\section{Extension to JCCO with Closed-loop Communication Strategies} 
\label{sec: closed-loop}

In this section, we extend our study to the JCCO problem with  \emph{closed-loop}  communication strategies, i.e., for any  
$h\in[H],i\in[n], g_{i,h}^m:\cI_{i,h^-}\times\cM_{1:h-1}\rightarrow \cM_{i,h}$, with associated (closed-loop) control strategies $g_{i,h}^a:\cI_{i,h^+}\times \cM_{1:h}\rightarrow \cU_{i}$. Since each agent may influence others through the additional sharing determined by the communication strategy $g_{i,h}^m$, while others may not access the input of $g_{i,h}^m$, PN IS may break. Therefore, the SI-CIB condition does not hold in general  and the common-information-based beliefs need not be Gaussian distributions that can be characterized by the finite-dimensional conditional mean and covariance. 
For the recursion at $h=1$, identify $(\bX_1,\bP_{0^+})=(\bX_1,\emptyset)$ with $\bX_1$ and set
$\bB_{0^+}:=\cN(\bm0,\Sigma_1)$, $\bTheta_{0^+}:=\bm0$, and $\Sigma_{0^+}:=\Sigma_1$.  Since the strict expansion has the same initial law, set $\tilde{\bB}_{0^+}:=\bB_{0^+}$, $\tilde{\bTheta}_{0^+}:=\bTheta_{0^+}$, $\tilde{\Sigma}_{0^+}:=\Sigma_{0^+}$, and $\gamma_0^a=\tilde\gamma_0^a:=\emptyset$.
Fortunately, we can still apply the approach designed in \cite{LTC}, which expanded the problem into a new one that satisfies the SI-CIB condition, and then solved it via dynamic programming.
To guarantee the SI-CIB condition, \cite{LTC} introduced the following assumption. 
\begin{assumption}
    \label{ass: limit_communication}
    The communication strategies take common information and past communication actions as input, i.e., $
        \forall i\in[n],h\in[H], g_{i,h}^m:\cC_{h^-}\times \cM_{1:h-1}\rightarrow \cM_{i,h}.$
\end{assumption}

To this end, we first introduce some additional notation. We denote by $\cP_{h^-}(M_{1:h-1})\subseteq \cP_{h^-}$ the set of all possible $P_{h^-}$ given $M_{1:h-1}$. Similarly, we denote by $\cP_{h^+}(M_{1:h})\subseteq \cP_{h^+}$ the set of all possible $P_{h^+}$ given $M_{1:h}$. Then, given any $M_{1:h-1}$, we can define $\bB_{h^-}\in\cP(\cX\times\cP_{h^-}(M_{1:h-1}))$ as the conditional probability measure of state $\bX_h$ and private information $\bP_{h^-}$, given the past strategies $g_{1:h-1}^m,g_{1:h-1}^a$, communication actions $M_{1:h-1}$, and common information $\bC_{h^-}$ before additional sharing. Similarly, we define $\bB_{h^+}\in\cP(\cX\times\cP_{h^+}(M_{1:h}))$ as the conditional probability measure of $\bX_h$ and $\bP_{h^+}$ given $g_{1:h}^m,g_{1:h-1}^a, \bM_{1:h}$, and $\bC_{h^+}$. Hence, we can write $\bB_{h^-}=\Pi_{h^-}(\bC_{h^-},\bM_{1:h-1}, g_{1:h-1}^m,g_{1:h-1}^a)$ and $\bB_{h^+}=\Pi_{h^+}(\bC_{h^+},\bM_{1:h}, g_{1:h}^m,g_{1:h-1}^a)$ for some functionals $\Pi_{h^-}$ and $\Pi_{h^+}$, respectively. Then, we can extend the SI-CIB condition as follows.
\begin{definition}[SI-CIB Condition of $\cD$]
We say that a JCCO problem $\cD$ satisfies the \emph{strategy-independent common-information-based belief} condition  if 
\begin{itemize}
    \item For any $h\in[H]$, any realization $C_{h^-}\in \cC_{h^-}, M_{1:h-1}\in \cM_{1:h-1}$ and any two strategies $(g_{1:h-1}^m,g_{1:h-1}^a)$ and $(g_{1:h-1}^{m,'},g_{1:h-1}^{a,'})$ that can reach $C_{h^-}$ and $M_{1:h-1}$, it holds that $\Pi_{h^-}(C_{h^-},M_{1:h-1}, g_{1:h-1}^m,g_{1:h-1}^a)=\Pi_{h^-}(C_{h^-},M_{1:h-1}, g_{1:h-1}^{m,'},g_{1:h-1}^{a,'}).$ 
    \item For any $h\in[H]$, any realization $C_{h^+}\in \cC_{h^+}, M_{1:h}\in \cM_{1:h}$ and any two strategies $(g_{1:h}^m,g_{1:h-1}^a)$ and $(g_{1:h}^{m,'},g_{1:h-1}^{a,'})$ that can reach $C_{h^+}$ and $M_{1:h}$, it holds that $
        \Pi_{h^+}(C_{h^+},M_{1:h}, g_{1:h}^m,g_{1:h-1}^a)=\Pi_{h^+}(C_{h^+},M_{1:h}, g_{1:h}^{m,'},g_{1:h-1}^{a,'}).$
\end{itemize}
\label{def: SI-CIB-closed_loop}
\end{definition}

Under Assumption \ref{ass: limit_communication}, for both timesteps  $h^-$ and $h^+$, we can expand the problem $\cD$ into a new one by adding some actions that affect other agents into the common information, similarly to Equation  \eqref{eq: expansion}. After the strict expansion, the new problem satisfies the SI-CIB condition, and the team-optimal strategy of the expanded JCCO problem can be used to obtain that of the original one. 
More importantly, 
satisfying the SI-CIB condition ensures that  $\tilde{\bB}_{h^-}$ and $\tilde{\bB}_{h^+}$ are Gaussian distributions, as shown below. 
\begin{theorem}
    \label{theorem: closed_loop satisfying SI-CIB and Gaussian}
    Let $\tilde{\cD}$ be the problem obtained by the above strict expansion from a JCCO problem $\cD$ with PN IS that satisfies Assumptions \ref{ass: evolution rule}, \ref{ass: useless action}, \ref{ass: non-degeneracy}, and \ref{ass: limit_communication}. Then,  
    $\tilde{\cD}$ is a JCCO problem satisfying the SI-CIB condition. Moreover, for any $h\in[H]$, $\tilde{\bB}_{h^-},\tilde{\bB}_{h^+}$ are Gaussian distributions. Formally, $\tilde{\bB}_{h^-}=\cN(\tilde{\bTheta}_{h^-},\tilde{\Sigma}_{h^-}), \tilde{\bB}_{h^+}=\cN(\tilde{\bTheta}_{h^+},\tilde{\Sigma}_{h^+})$, and
    \begin{equation}
        \begin{aligned}
&\tilde{\bTheta}_{h^-}=\tilde{\Psi}_{h^-}^1(\tilde{\bTheta}_{(h-1)^+},\tilde{\bM}_{1:h-1},\tilde{\bZ}_h^b),\quad \tilde{\Sigma}_{h^-}=\tilde{\Psi}_{h^-}^2(\tilde{\bM}_{1:h-1}),\\
&\qquad \tilde{\bTheta}_{h^+}=\tilde{\Psi}_{h^+}^1(\tilde{\bTheta}_{h^-},\tilde{\bM}_{1:h},\tilde{\bZ}_h^a),\quad \tilde{\Sigma}_{h^+}=\tilde{\Psi}_{h^+}^2(\tilde{\bM}_{1:h}),
    \label{equ: closed-loop evolution of mean and covariance}
        \end{aligned}
    \end{equation}
    
    \noindent  for some functions $\tilde{\Psi}_{h^-}^1,\tilde{\Psi}_{h^+}^1,\tilde{\Psi}_{h^-}^2,\tilde{\Psi}_{h^+}^2$ that do not depend on the strategies $\tilde{g}_{1:h}^m,\tilde{g}_{1:h}^a$. 
\end{theorem}

Theorem \ref{theorem: closed_loop satisfying SI-CIB and Gaussian} provides a way to solve JCCO with closed-loop communication strategies via dynamic programming. Specifically, we can conduct dynamic programming over the spaces of 
$(\tilde \bM_{1:h},\tilde\bTheta_{h^+})\in \tilde{\cM}_{1:h}\times(\tilde{\cX}\times\tilde{\cP}_{h^+}(\tilde M_{1:h}))$ and the spaces of $(\tilde \bM_{1:h-1},\tilde\bTheta_{h^-})\in \tilde{\cM}_{1:h-1}\times(\tilde{\cX}\times\tilde{\cP}_{h^-}(\tilde M_{1:h-1}))$, to solve for the 
optimal strategies $\tilde g_h^{a,\ast}$ and $\tilde g_h^{m,\ast}$, respectively.  
The overall  algorithm, under its stated attainment assumption, is tabulated as Algorithm \ref{main algorithm} in \S\ref{appendix:closed_loop}.

\section{Numerical Example and Experimental Results}

In this section, we demonstrate the proposed method through concrete JCCO examples. We first present a numerical example that illustrates the implementation of our method step by step. We then evaluate the approach experimentally and analyze its performance.
\subsection{Numerical example}
We illustrate the method using a linear-quadratic JCCO problem $\cD$ with $n=2$, $H=2$, and open-loop communication strategies. 
The baseline sharing of $\cD$ is \emph{one-step measurement delay}; that is, for any $i\in[2]$ and $h\in[2]$, $\bC_{h^-}=\{\bY_{1:h-1}\}$ and $\bP_{i,h^-}=\{\bY_{i,h}\}$. The communication action spaces are $\cM_{i,h}=\{0,1\}$ for all $i\in[2]$ and $h\in[2]$, where $\bM_{i,h}=1$ indicates $\bZ_{i,h}^a=\{\bY_{i,h}\}$ and $\bM_{i,h}=0$ indicates $\bZ_{i,h}^a=\emptyset$. The communication cost is $\cK_h(\bM_h)=0.2(\bM_{1,h}+\bM_{2,h})$. The system is defined by $\cX=\cY_i=\cU_i=\RR$ for all $i\in[2]$, with the following data:
\vspace{-2mm}
\begin{align*}
    &A_1=1.25, A_2=1.15, B_{1,1}=1,B_{2,1}=0.9,B_{1,2}=0.85, B_{2,2}=1.1, E_{1,1}=1.1,E_{2,1}=0.8,E_{1,2}=0.9, E_{2,2}=1.2,\\
    &\Sigma_{0,1}=1.5,\Sigma_{0,2}=2,\Sigma_{1,1}=1,\Sigma_{2,1}=1.5,\Sigma_{1,2}=2, \Sigma_{2,2}=0.8,\Sigma_1=2,Q_1^1=5,Q_2^1=6,Q_3^1=8, Q_1^2=Q_2^2=\II_2.
\end{align*}
\vspace{-6mm}

{\noindent The problem $\cD$ has PN IS and satisfies Assumptions \ref{ass: evolution rule}, \ref{ass: useless action}, and \ref{ass: non-degeneracy}. We solve it as follows.}

\vspace{7pt}
\noindent\textbf{Step 1:}
There are 16 possible communication strategies with $g_{i,h}^m\in\{0,1\}$ for all $h\in[2]$ and $i\in[2]$. For each $g_{1:2}^m$, we construct the associated decentralized LQG problem $\cD(g_{1:2}^m)$. 

\vspace{7pt}
\noindent\textbf{Step 2:} We expand $\cD(g_{1:2}^m)$ to $\tilde{\cD}(g_{1:2}^m)$ and compute the corresponding filter matrices and covariance matrices from the information structure of $\tilde{\cD}(g_{1:2}^m)$. To illustrate the computation, consider $g_{1,1}^m=g_{2,2}^m=1$ and $g_{2,1}^m=g_{1,2}^m=0$. This choice yields $\tilde{\bC}_1=\{\tilde{\bY}_{1,1}\},\tilde{\bC}_2=\{\tilde{\bY}_1,\tilde{\bU}_1,\tilde{\bY}_{2,2}\},\tilde{\bP}_{2,1}=\{\tilde{\bY}_{2,1}\},\tilde{\bP}_{1,2}=\{\tilde{\bY}_{1,2}\}$, and $\tilde{\bP}_{1,1}=\tilde{\bP}_{2,2}=\emptyset$.
For timestep $h=2$, $\{K_2^j\}_{j\in[4]}$ and $\tilde{\Sigma}_2$ are given by
\begin{equation}
\begin{aligned}
        &K_2^1=E_{1,2}^{\mathrm{ext}}\left(1-\frac{\mathring{P}_2E_{2,2}^2}{E_{2,2}^2\mathring{P}_{2}+\Sigma_{2,2}}\right)A_1\mathring{P}_1\frac{E_{1,1}^2\Sigma_1+\Sigma_{1,1}}{\Sigma_1\Sigma_{1,1}}
        \begin{bmatrix}
            1&0
        \end{bmatrix},K_2^2=E_{1,2}^{\mathrm{ext}}\left(1-\frac{\mathring{P}_2E_{2,2}^2}{E_{2,2}^2\mathring{P}_{2}+\Sigma_{2,2}}\right)A_1\mathring{P}_1\frac{E_{2,1}}{\Sigma_{2,1}},\\
        &K_2^3=E_{1,2}^{\mathrm{ext}}\left(1-\frac{\mathring{P}_2E_{2,2}^2}{E_{2,2}^2\mathring{P}_{2}+\Sigma_{2,2}}\right)B_1, K_2^4=E_{1,2}^{\mathrm{ext}}\frac{\mathring{P}_2E_{2,2}
        \begin{bmatrix}
            0&1
        \end{bmatrix}}{E_{2,2}^2\mathring{P}_{2}+\Sigma_{2,2}}, \tilde{\Sigma}_2=\begin{bmatrix}\frac{\mathring{P}_2\Sigma_{2,2}}{E_{2,2}^2\mathring{P}_{2}+\Sigma_{2,2}}&\frac{E_{1,2}\mathring{P}_2\Sigma_{2,2}}{E_{2,2}^2\mathring{P}_{2}+\Sigma_{2,2}}\\
           \frac{E_{1,2}\mathring{P}_2\Sigma_{2,2}}{E_{2,2}^2\mathring{P}_{2}+\Sigma_{2,2}}&\Sigma_{1,2}+\frac{E_{1,2}^2\mathring{P}_2\Sigma_{2,2}}{E_{2,2}^2\mathring{P}_{2}+\Sigma_{2,2}}
       \end{bmatrix},
\end{aligned}
\end{equation}
where $\mathring{P}_1=(\frac{1}{\Sigma_{1}}+\frac{E_{1,1}^2}{\Sigma_{1,1}}+\frac{E_{2,1}^2}{\Sigma_{2,1}})^{-1}$, $\mathring{P}_2=A_1^2\mathring{P}_1+\Sigma_{0,1}$ and $E_{1,2}^{\mathrm{ext}}=\begin{bmatrix}
    1&E_{1,2}
    \end{bmatrix}^\top$. The initial quantities $K_1^1,K_1^2$, and $\tilde{\Sigma}_1$ are computed similarly.

\vspace{7pt}
\noindent\textbf{Step 3:} Apply the backward recursion and finite-dimensional convex quadratic minimizations in Theorem \ref{thm: dp_in_decLQG} to compute $\tilde{c}_0$ and $\tilde{g}_{1:2}^\ast$ for the selected $g_{1:2}^m$, and define
$f(g_{1:2}^m):=\min_{g_{1:2}^a}J_{\cD}( g_{1:2}^m,g_{1:2}^a)=\tilde{c}_0+\cK_1(g_1^m)+\cK_2(g_2^m)$. 

\vspace{7pt}
\noindent\textbf{Step 4:} Compute $g_{1:2}^{m,\ast}\in\argmin_{g_{1:2}^m\in \cG_{1:2}^m}f(g_{1:2}^m)$. The optimizer is
    $g_{1,1}^{m,\ast}=g_{2,2}^{m,\ast}=1$ and $g_{2,1}^{m,\ast}=g_{1,2}^{m,\ast}=0$. The optimal strategy $\tilde{g}_{1:2}^\ast$ for the problem $\tilde{\cD}(g_{1:2}^{m,\ast})$ is\footnote{All numerical values are rounded to three decimal places.}
    \begin{align*}
    &\tilde{g}_{1,1}^\ast(\tilde{\bI}_{1,1})=-0.410\tilde{\bY}_{1,1},\tilde{g}_{2,1}^{\ast}(\tilde{\bI}_{2,1})=-0.219\tilde{\bY}_{1,1}-0.292\tilde{\bY}_{2,1},\\
    &\tilde{g}_{1,2}^\ast(\tilde{\bI}_{1,2})=-0.038\tilde{\bY}_{1,1}-0.019\tilde{\bY}_{2,1}-0.060\tilde{\bU}_{1,1}-0.054\tilde{\bU}_{2,1}-0.196\tilde{\bY}_{1,2}-0.200\tilde{\bY}_{2,2}, \\
    &\tilde{g}_{2,2}^\ast(\tilde{\bI}_{2,2})=-0.079\tilde{\bY}_{1,1}-0.038\tilde{\bY}_{2,1}-0.123\tilde{\bU}_{1,1}-0.110\tilde{\bU}_{2,1}-0.410\tilde{\bY}_{2,2}.
    \end{align*}
   The optimal objective of $\tilde{\cD}(g_{1:2}^{m,\ast})$ is $\tilde{c}_0=45.947$. We then compute $g_{1:2}^{a,\ast}=\varphi(\tilde{g}_{1:2}^\ast,\cD(g_{1:2}^{m,\ast}))$, where $g_1^{a,\ast}=\tilde{g}_1^\ast$, and
   \small
    \begin{align*}
        &g_{1,2}^{a,\ast}(\bI_{1,2^+})=-0.002\bY_{1,1}-0.003\bY_{2,1}-0.196\bY_{1,2}-0.200\bY_{2,2},\\
        &g_{2,2}^{a,\ast}(\bI_{2,2^+})=-0.004\bY_{1,1}-0.006\bY_{2,1}-0.410\bY_{2,2},
    \vspace{-2mm}
    \end{align*}
    \normalsize 
    The optimal objective of $\cD$ is $J_\cD(g_{1:2}^{m,\ast},g_{1:2}^{a,\ast})=46.347$.
    
\subsection{Experimental results}

\begin{figure}[!t]
    \centering
        \begin{subfigure}{0.48\linewidth}
            \centering
            \includegraphics[width=\linewidth]{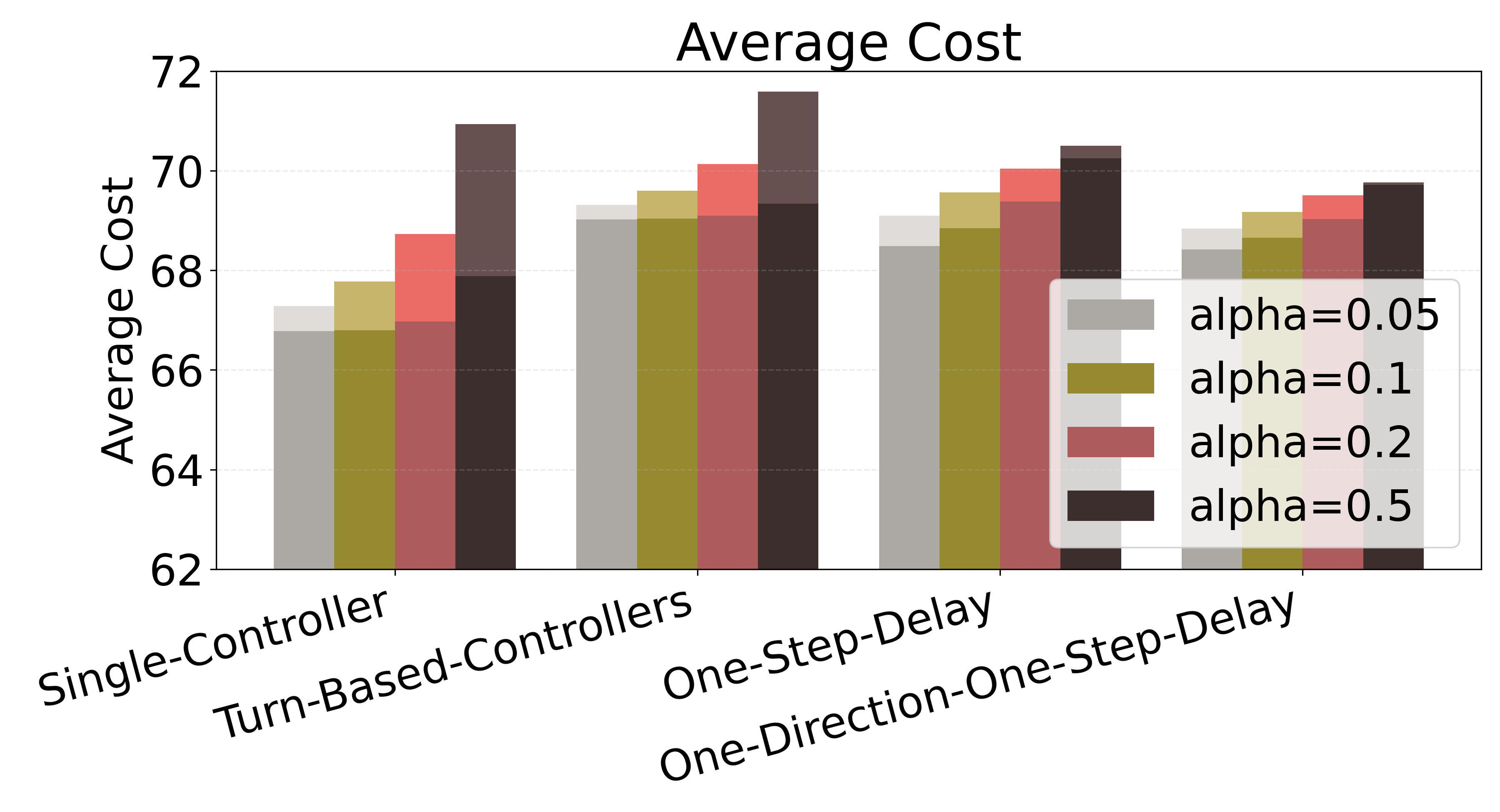}
            \caption{}
        \end{subfigure}
        \begin{subfigure}{0.48\linewidth}
            \centering
            \includegraphics[width=\linewidth]{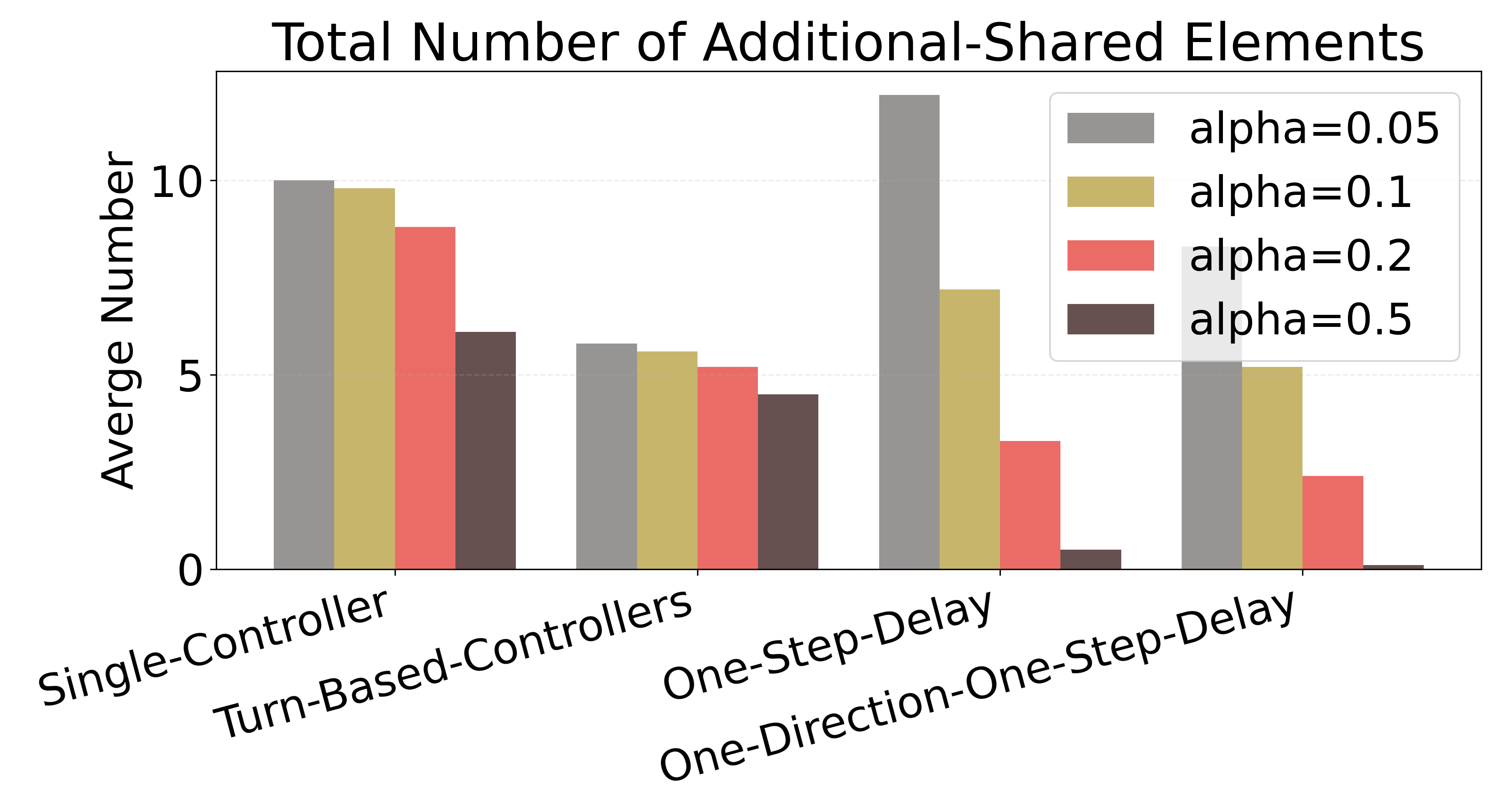}
            \caption{}
        \end{subfigure}
    \caption{Average values over 10 random seeds under different baseline sharing protocols and values of $\alpha$. Figure (a): For each bar, the dark and light portions correspond to the control cost and communication cost, respectively. Figure (b): Each value represents the total number of additionally shared elements.}
    \label{fig:bar}
\end{figure}
\begin{figure}[!t]
    \centering
        \begin{subfigure}{0.46\linewidth}
            \centering
            \includegraphics[width=\linewidth]{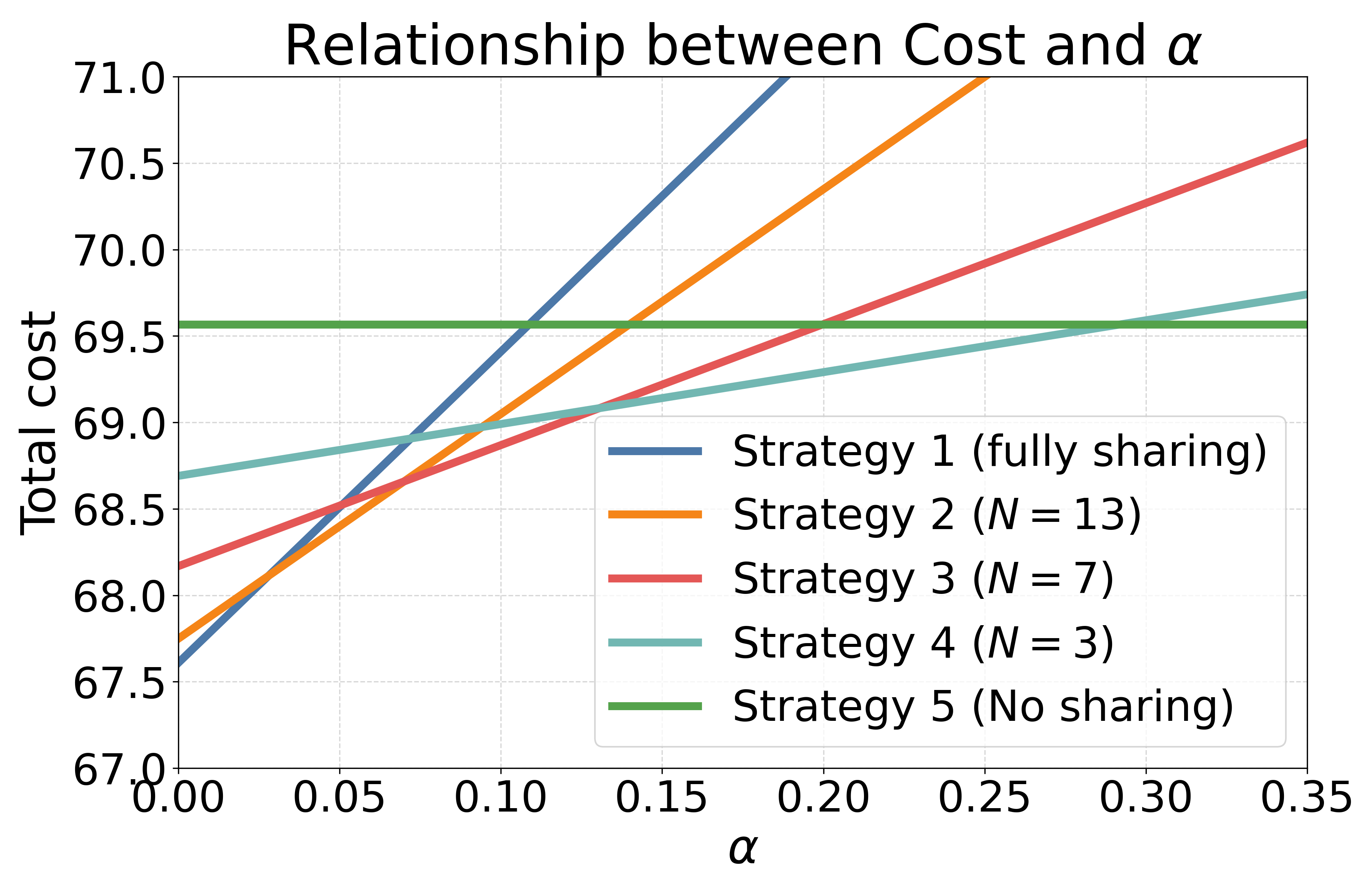}
            \caption{}
        \end{subfigure}
        \begin{subfigure}{0.46\linewidth}
            \centering
            \includegraphics[width=\linewidth]{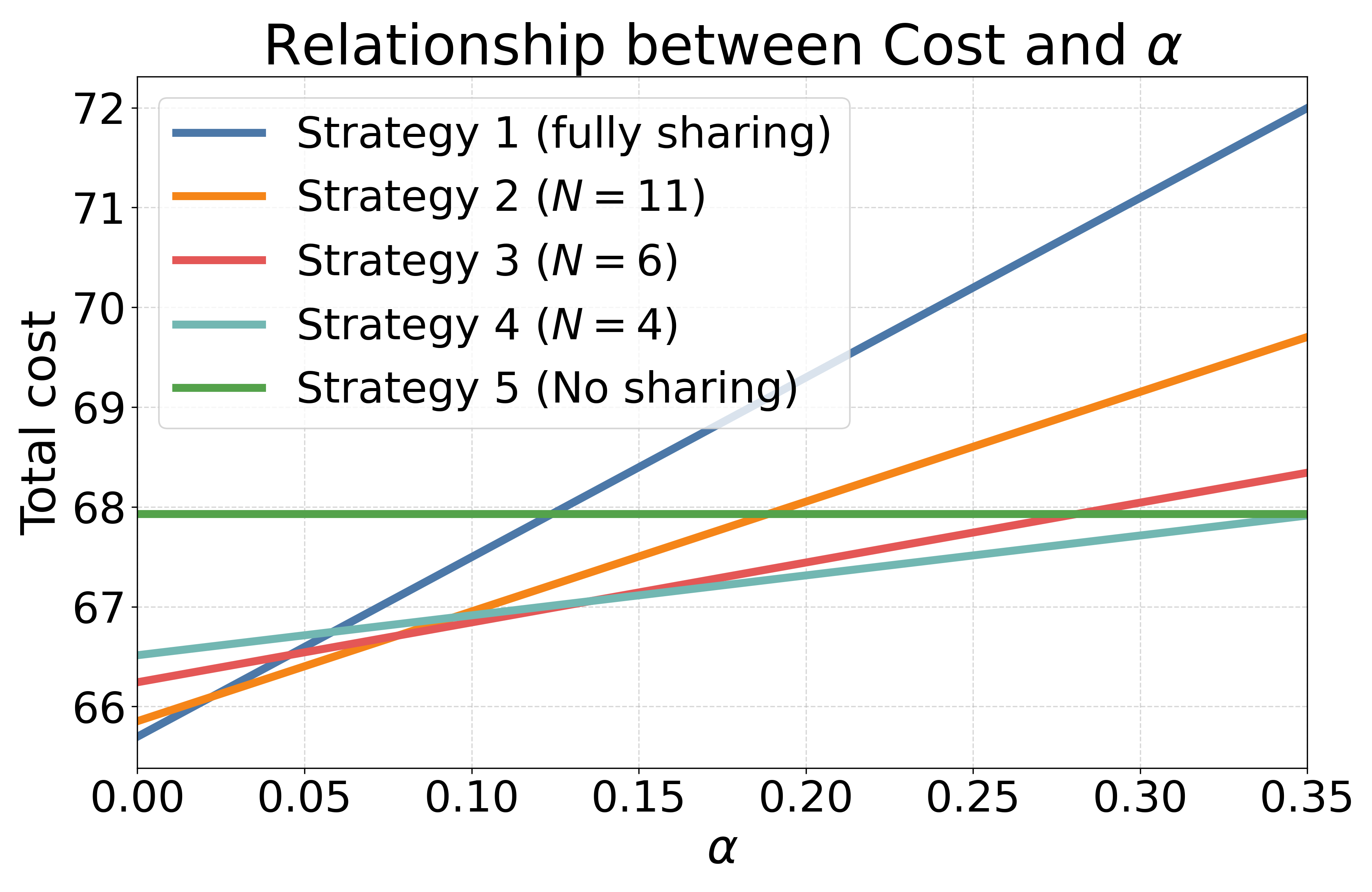}
            \caption{}
        \end{subfigure}
    \caption{Total cost under different values of $\alpha$ and communication strategies. Each line corresponds to a different communication strategy, with $N$ representing the number of elements shared through additional sharing. Figure (a): One-Step-Delay baseline sharing. Figure (b): One-Direction-One-Step-Delay baseline sharing.}
    \label{fig:line}
\end{figure}

We evaluate our approach on JCCO problems with the following four information structures. In the descriptions below, labeled observations and actions with nonpositive time indices are omitted, and $\bC_{0^+}=\bU_0=\emptyset$.
\begin{itemize}
    \item Single-Controller: Only agent 1 can control the underlying state $\bX_h$, i.e., $B_{i,h}=\mathbf{0}$ for all $h\in[H]$ and $i\in[2:n]$. Agent 1 shares all of its information with the other agents, while the remaining agents share their information with a $d$-step delay for some $d>1$; that is, $\bC_{h^-}=\bC_{(h-1)^+}\cup \{\bY_{1,h},\bU_{1,h-1},\bY_{2:n,h-d}\}$, $\bP_{1,h^-}=\emptyset$, and $\bP_{i,h^-}=\bP_{i,(h-1)^+}\cup \{\bY_{i,h}\}\setminus\{\bY_{i,h-d}\}$ for all $h\in[H]$ and $i\in[2:n]$. In the experiments, we set $n=3, H=6$.
    \item Turn-based-Controllers: For each $h\in[H]$ with the unique representation $h=kn+i$, where $k\in\mathbb{Z}_{\geq0}$ and $i\in[n]$, only agent $i$ can control the state at timestep $h$, i.e., $B_{j,h}=\mathbf{0}$ for all $j\neq i$. For $h<H$, at timestep $h+1$, agent $i$ shares its past observations and actions with the other agents; that is, $\bC_{(h+1)^-}=\bC_{h^+}\cup\{\bY_{i,h-n+1:h},\bU_{i,h}\}$. In the experiments, we set $n=3, H=6$.
    \item One-Step-Delay: Through baseline sharing, each agent shares its past observations and actions with the other agents, i.e., $\bC_{h^-}=\bC_{(h-1)^+}\cup \{\bY_{h-1},\bU_{h-1}\}$ and $\bP_{i,h^-}=\{\bY_{i,h}\}$ for all $h\in[H]$ and $i\in[n]$. In the experiments, we set $n=3, H=7$.
    \item One-Direction-One-Step-Delay: Agent 1 shares all of its information with the other agents, while the remaining agents share their past observations and actions through baseline sharing; that is, $\bC_{h^-}=\bC_{(h-1)^+}\cup \{\bY_{1,h},\bY_{2:n,h-1}, \bU_{h-1}\},\bP_{1,h^-}=\emptyset, \bP_{i,h^-}=\{\bY_{i,h}\}$ for all $h\in[H]$ and $i\in[2:n]$. In the experiments, we set $n=3, H=7$.
\end{itemize}
For additional sharing, each agent decides whether to share each element of its private information with the other agents. The communication cost is $\alpha$ times the number of elements shared additionally, where $\alpha\in \{0.05,0.1,0.2,0.5\}$. We set $\cX=\cU_i=\cY_i=\RR$ for all $i\in[n]$. The system matrices are sampled using a random seed $\beta$. Each nonzero block of the dynamics matrices $A_{1:H},B_{1:H},E_{1:H}$ is sampled uniformly from $[0.8,1.3]$, $[0.5,1.2]$, and $[0.7,1.3]$, respectively. The noise covariance matrices $\{\Sigma_{i,h}\}_{i\in[0:n],h\in[H]}$ are set to be block diagonal, and each diagonal block is sampled uniformly from $[0.5,2.0]$. The covariance matrix $\Sigma_1$ is also set to be block diagonal, with each diagonal block sampled uniformly from $[1.0,3.0]$. We run experiments for all four information structures, the four values of $\alpha$ in $\{0.05,0.1,0.2,0.5\}$, and 10 random seeds $\beta=1,\ldots,10$.

The results are presented in Figure \ref{fig:bar}, which shows that lower communication costs encourage agents to share more information, thereby reducing both the control cost and the total objective value. 
We also evaluate the total cost under different values of $\alpha$ and different communication strategies, each paired with its optimal control strategy, as presented in Figure \ref{fig:line}. The results show that when $\alpha$ is small, strategies with more additional sharing can attain lower total costs. Conversely, when $\alpha$ is large, the communication penalty makes strategies with less additional sharing preferable.

\section{Concluding Remarks}

In this paper, we formalized joint communication-control strategy optimization in decentralized multi-agent LQG systems. We identified structural conditions under which additional information sharing preserves partial nestedness and showed that excluding them can yield either nonlinearity or  non-existence of the optimal control strategies. For each fixed open-loop communication strategy, we derived Riccati-type dynamic-programming recursions for the optimal controller, which, of independent interest, can also be applied to decentralized LQG problems with partially nested information structures and output feedback, thus advancing the results in \cite{AA}. We then further extended the approach to JCCO with closed-loop communication strategies. 
Our work opens several important directions for future research, including extensions to non-Gaussian noise and non-cooperative settings, as well as relaxations of the assumptions for specific system structures such as those with factorized states.

\section*{Acknowledgement}
The authors acknowledge the valuable feedback from the anonymous reviewers of  IEEE CDC 2026, and acknowledge the support from the Army Research Office grant  W911NF-24-1-0085, the NSF CAREER Award 2443704, the AFOSR YIP Award FA 9550-25-1-0258, a JP Morgan Faculty Research Award,  and an AI Safety
Research Award from Coefficient Giving. 

\bibliographystyle{IEEEtran}
\bibliography{IEEEabrv,reference}

@article{fu2012lack,
  title={Lack of separation principle for quantized linear quadratic {G}aussian control},
  author={Fu, Minyue},
  journal={IEEE Trans. Autom. Control},
  volume={57},
  number={9},
  pages={2385--2390},
  year={2012},
  publisher={IEEE}
}

@article{maity2021optimal,
  title={Optimal controller synthesis and dynamic quantizer switching for linear-quadratic-Gaussian systems},
  author={Maity, Dipankar and Tsiotras, Panagiotis},
  journal={IEEE Trans. Autom. Control},
  volume={67},
  number={1},
  pages={382--389},
  year={2021},
  publisher={IEEE}
}

@article{peng2013event,
  title={Event-triggered communication and $H_\infty$ control co-design for networked control systems},
  author={Peng, Chen and Yang, Tai Cheng},
  journal={Automatica},
  volume={49},
  number={5},
  pages={1326--1332},
  year={2013},
  publisher={Elsevier}
}

@article{matni2016regularization,
  title={Regularization for design},
  author={Matni, Nikolai and Chandrasekaran, Venkat},
  journal={IEEE Trans. Autom. Control},
  volume={61},
  number={12},
  pages={3991--4006},
  year={2016},
  publisher={IEEE}
}

@article{zhang2006communication,
  title={Communication and control co-design for networked control systems},
  author={Zhang, Lei and Hristu-Varsakelis, Dimitrios},
  journal={Automatica},
  volume={42},
  number={6},
  pages={953--958},
  year={2006},
  publisher={Elsevier}
}

@inproceedings{nayyar2015structural,
  title={Structural results for partially nested {LQG} systems over graphs},
  author={Nayyar, Ashutosh and Lessard, Laurent},
  booktitle={American Control Conf. (ACC)},
  pages={5457--5464},
  year={2015}
}

@inproceedings{LTC,
  title={Principled Learning-to-Communicate with
Quasi-Classical Information Structures},
  author={Xiangyu Liu and Haoyi You and Kaiqing Zhang},
  booktitle={IEEE Conf. Decis. Control (CDC)},
  year={2025}
}

@article{witsenhausen1971information,
  title={On information structures, feedback and causality},
  author={Witsenhausen, Hans S},
  journal={SIAM J.  Control},
  volume={9},
  number={2},
  pages={149--160},
  year={1971},
  publisher={SIAM}
}

@article{yuksel2013jointly,
  title={Jointly optimal {LQG} quantization and control policies for multi-dimensional systems},
  author={Y{\"u}ksel, Serdar},
  journal={IEEE Trans. Autom. Control},
  volume={59}, 
  pages={1612--1617},
  year={2013},
  publisher={IEEE}
}

@article{QCstatefeedback,
  title={Optimal decentralized state-feedback control with sparsity and delays},
  author={Lamperski, Andrew and Lessard, Laurent},
  journal={Automatica},
  pages={143--151},
  year={2015},
  publisher={Elsevier}
}

@article{ho1972team,
  title={Team decision theory and information structures in optimal control problems -- Part {I}},
  author={Ho, Yu-Chi and Chu, Kai-Ching},
  journal={IEEE Trans. Autom. Control},
  volume={17},
  pages={15--22},
  year={1972},
  publisher={IEEE}
}

@article{CIBLQgames,
  title={Common information based {M}arkov perfect equilibria for linear-Gaussian games with asymmetric information},
  author={Gupta, Abhishek and Nayyar, Ashutosh and Langbort, C{\'e}dric and Basar, Tamer},
  journal={SIAM J. Control Optim.},
  volume={52},
  number={5},
  pages={3228--3260},
  year={2014},
  publisher={SIAM}
}

@article{staticteam,
  title={Team decision problems},
  author={Radner, Roy},
  journal={The Annals of Mathematical Statistics},
  volume={33},
  number={3},
  pages={857--881},
  year={1962},
  publisher={JSTOR}
}

@book{yuksel2023stochastic,
  title={Stochastic Teams, Games, and Control under Information Constraints},
  author={Y{\"u}ksel, Serdar and Ba\c{s}ar, Tamer},
  year={2023},
  publisher={Springer Nature}
}

@article{ashutosh2013team,
  title={Decentralized stochastic control with partial history sharing: A common information approach},
  author={Nayyar, Ashutosh and Mahajan, Aditya and Teneketzis, Demosthenis},
  journal={IEEE Trans. Autom. Control},
  volume={58},
  number={7},
  pages={1644--1658},
  year={2013},
  publisher={IEEE}
}

@article{ashutosh2013game,
  title={Common information based {Markov} perfect equilibria for stochastic games with asymmetric information: Finite games},
  author={Nayyar, Ashutosh and Gupta, Abhishek and Langbort, Cedric and Ba{\c{s}}ar, Tamer},
  journal={IEEE Trans. Autom. Control},
  volume={59},
  pages={555--570},
  year={2013},
  publisher={IEEE}
}

@article{liu2023tractable,
  title={Partially Observable Multiagent Reinforcement Learning with Information Sharing},
  author={Liu, Xiangyu and Zhang, Kaiqing},
  journal={SIAM J. Control Optim.},
  volume={64},
  number={2},
  pages={673--697},
  year={2026},
  publisher={SIAM}
}

@inproceedings{aditya2012information,
  title={Information structures in optimal decentralized control},
  author={Mahajan, Aditya and Martins, Nuno C and Rotkowitz, Michael C and Y{\"u}ksel, Serdar},
  booktitle={IEEE Conf. Decis. Control (CDC)},
  year={2012}
}

@inproceedings{2016learningtocommunicate,
  title={Learning to communicate with deep multi-agent reinforcement learning},
  author={Foerster, Jakob and Assael, Ioannis Alexandros and De Freitas, Nando and Whiteson, Shimon},
  booktitle={NeurIPS},
  year={2016}
}

@article{ashutoshcommunicate1,
  author  = {Sudhakara, S. and Kartik, D. and Jain, R. and Nayyar, A.},
  title   = {Optimal Communication and Control Strategies in a Cooperative Multiagent {MDP} Problem},
  journal = {IEEE Trans. Autom. Control},
  volume  = {69},
  number  = {10},
  pages   = {6959--6966},
  year    = {2024}
}

@inproceedings{ashutoshcommunicate2,
  title={Optimal communication and control strategies for a multi-agent system in the presence of an adversary},
  author={Kartik, Dhruva and Sudhakara, Sagar and Jain, Rahul and Nayyar, Ashutosh},
  booktitle={IEEE Conf. Decis. Control (CDC)},
  year={2022}
}

@inproceedings{learningpropogation,
  title={Learning multiagent communication with backpropagation},
  author={Sukhbaatar, Sainbayar and Szlam,Arthur and Fergus, Rob},
  booktitle={NeurIPS},
  year={2016}
}

@article{witsenhausen1968counterexample,
  title={A counterexample in stochastic optimum control},
  author={Witsenhausen, Hans S},
  journal={SIAM J. Control},
  volume={6},
  number={1},
  pages={131--147},
  year={1968},
  publisher={SIAM}
}

@article{matni2015communication,
  title={Communication Delay Co-Design in $\mathcal{H}_2$-Distributed Control Using Atomic Norm Minimization},
  author={Matni, Nikolai},
  journal={IEEE Trans. Control Netw. Syst.},
  volume={4},
  number={2},
  pages={267--278},
  year={2015},
  publisher={IEEE}
}

@article{AA,
  title={Sufficient statistics for linear control strategies in decentralized systems with partial history sharing},
  author={Mahajan, Aditya and Nayyar, Ashutosh},
  journal={IEEE Trans. Autom. Control},
  volume={60},
  number={8},
  pages={2046--2056},
  year={2015},
  publisher={IEEE}
}

@article{psedoinverse1,
  title   = {Structure and Stability of Discrete-Time Optimal Systems},
  author  = {Rappaport, David and Silverman, Leonard M.},
  journal = {IEEE Trans. Autom. Control},
  volume  = {16},
  number  = {3},
  pages   = {227--233},
  year    = {1971},
  month   = jun,
  doi     = {10.1109/TAC.1971.1099702}
}

@inproceedings{psedoinverse2,
  title={The generalised discrete algebraic {Riccati} equation arising in {LQ} optimal control problems: Part {I}},
  author={Ferrante, Augusto and Ntogramatzidis, Lorenzo},
  booktitle={IEEE Conf. Decis. Control (CDC)},
  pages={6394--6399},
  year={2012}
}

@book{Linear_Algebra,
  author    = {Sheldon Axler},
  title     = {Linear Algebra Done Right},
  edition   = {4},
  publisher = {Springer},
  year      = {2024}
}

@article{eckart1936approximation,
  author  = {Eckart, Carl and Young, Gale},
  title   = {The Approximation of One Matrix by Another of Lower Rank},
  journal = {Psychometrika},
  volume  = {1},
  number  = {3},
  pages   = {211--218},
  year    = {1936},
  doi     = {10.1007/BF02288367}
}

@article{kalman-filter-pseudoinverse,
  title   = {On {Kalman} filtering for conditionally Gaussian systems with random matrices},
  author  = {Chen, Han-Fu and Kumar, P. R. and van Schuppen, J. H.},
  journal = {Syst. Control Lett.},
  volume  = {13},
  number  = {5},
  pages   = {397--404},
  year    = {1989},
  doi     = {10.1016/0167-6911(89)90106-0}
}
\clearpage
\onecolumn
\appendix
\section{Decentralized Linear-Quadratic-Gaussian Control Problem}
\label{sec: DecLQG}

For $n>1$ agents, a (cooperative) decentralized linear-quadratic-Gaussian (decentralized LQG) control problem $\check{\cD}$ is described by a tuple $\langle H, \cX,\{\cY_i\}_{i\in[n]},\{\cU_{i}\}_{i\in[n]}, \{A_h\}_{h\in[H]},\{B_{i,h}\}_{i\in[n],h\in[H]}, \{E_{i,h}\}_{i\in[n],h\in[H]}, \{Q_h^1\}_{h\in[H+1]},\{Q_h^2\}_{h\in[H]}\rangle$,  where $H$ is the time horizon and   $\bX_h\in\cX=\RR^{d_x}$ is the state. At each timestep $h\in[H]$, each agent $i\in[n]$ receives a noisy observation $\bY_{i,h}\in\cY_i=\RR^{d_y^i}$ of the state $\bX_h$, and chooses a control action $\bU_{i,h}\in\cU_i=\RR^{d_u^i}$.
At timestep $h\in[H]$, 
we denote by $\bU_h=\begin{bmatrix}
    \bU_{1,h}^\top&\bU_{2,h}^\top&\cdots&\bU_{n,h}^\top
\end{bmatrix}^\top$ the joint control action of all the $n$ agents, and by $\cU=\RR^{\sum_{i=1}^nd_u^i}$ the joint control action space; we denote by $\bY_h=\begin{bmatrix}
    \bY_{1,h}^\top&\bY_{2,h}^\top&\cdots&\bY_{n,h}^\top
\end{bmatrix}^\top$ the joint observation, and by $\cY=\RR^{\sum_{i=1}^nd_y^i}$ the joint observation space. 
The system evolves as follows: for each timestep $h\in[H]$
\begin{equation}\label{equ: LQG_system_dynamics}
    \begin{aligned}
        \bX_{h+1}&=A_h\bX_h+B_h\bU_h+\bW_{0,h}=A_h\bX_h+\sum_{i=1}^nB_{i,h}\bU_{i,h}+\bW_{0,h},\\
        \bY_{i,h}&=E_{i,h}\bX_h+\bW_{i,h}, \forall i\in[n],
    \end{aligned}
\end{equation}
where $\bW_{i,h}$ is a random variable taking values in some real-vector space $\cW_{i,h}$ for all $h\in[H], i=0,1,\cdots, n$.  $A_h, B_h, E_{i,h}$  are matrices of {appropriate} dimensions, with $B_h=\begin{bmatrix}
    B_{1,h}&B_{2,h}&\cdots&B_{n,h}
\end{bmatrix}$. Here, $\bW_{i,h}\sim\cN(\mathbf{0},\Sigma_{i,h})$ for some covariance matrices $\Sigma_{i,h}\succeq 0$  
for all $i\in[0:n],h\in[H]$ and $\bX_1\sim\cN(\mathbf{0},\Sigma_1)$ for some covariance matrix  $\Sigma_1\succeq 0$. The random variables $\bX_1$ and $\{\bW_{i,h}\}_{i\in[0:n],h\in[H]}$ are mutually independent. 

At timestep $h\in[H]$,  each agent can access some information $\bI_{i,h}\subseteq\{\bY_{1:h},\bU_{1:h-1}\}
$, and the collection of all possible such available information is denoted by $\cI_{i,h}$. We use $\bI_h$ to denote the joint available information at timestep $h$.
Meanwhile, agents may share part of their information with each other, and the shared information is denoted by $\bZ_h$. 
At timestep $h$, the common information among all agents is defined as the union of all the shared information so far: $\bC_h=\cup_{t=1}^{h}\bZ_t$. The private information of agent $i$ at timestep $h$ is thus defined as $\bP_{i,h}=\bI_{i,h}\backslash \bC_h$, 
and the joint private information is denoted by $\bP_h:=[\bP_{1,h}^\top~\cdots~\bP_{n,h}^\top]^\top$. We denote by $\cC_h, \cP_h, \cP_{i,h},\cI_{i,h},\cI_{h}$ the sets of values of the random variables $\bC_h, \bP_h, \bP_{i,h},\bI_{i,h},\bI_{h}$ for each $h\in[H],i\in[n]$.

Each agent $i$ at timestep $h$ chooses the control action $\bU_{i,h}$ based on some strategy $g_{i,h}:\cI_{i,h}\rightarrow \cU_i$. We denote by $g_h:=(g_{1,h},g_{2,h},\cdots,g_{n,h})$ the joint control strategy of all the agents, and by $g_{1:h}:=(g_{1},g_{2},\cdots,g_h), \forall h\in[H]$ the sequence of joint strategies from timesteps  $1$ to $h$. We use $\cG_{i,h}$ to denote the strategy space of $g_{i,h}$, and use $\cG_h,\cG_{1:h}$ to denote the corresponding joint strategy spaces.
At each timestep $h\in[H]$, the cost is defined as $c_h=\bX_h^\top Q_h^1\bX_h+\bU_h^\top Q_h^2\bU_h$, and the cost of final timestep is $c_{H+1}=\bX_{H+1}^\top Q_{H+1}^1\bX_{H+1}$. Here, $\forall h\in[H], Q_h^1, Q_h^2\succeq 0, Q_{H+1}^1\succeq 0$ are symmetric matrices with appropriate dimensions. The objective among all the agents is thus  defined as
\begin{align*}
    J_{\check{\cD}}(g_{1:H}):=\EE\left[\sum_{h=1}^Hc_h+c_{H+1}\Bigggiven g_{1:H}\right]. 
\end{align*}
With this objective, we define the notion of team optimality for the problem $\check{\cD}$ as follows.
\begin{definition}[Team optimality]
   We call a joint strategy  $g_{1:H}^\ast\in \cG_{1:H}$ a
   \emph{team-optimal strategy} of  the decentralized LQG control problem $\check{\cD}$ if
   \begin{equation}
       \forall g_{1:H}\in\cG_{1:H},\qquad J_{\check{\cD}}(g_{1:H})\ge  J_{\check{\cD}}(g_{1:H}^\ast).
   \end{equation}
\end{definition}

\begin{corollary}[Adapted from \cite{ho1972team}]
\label{cor:degenerate-PN-linearity}
Suppose that the decentralized LQG control problem $\check{\cD}$ is PN, namely, for any $i_1,i_2\in[n]$ and $h_1,h_2\in[H]$ with $h_1<h_2$, if $\bU_{i_1,h_1}$ influences $\bI_{i_2,h_2}$, then $\bI_{i_1,h_1}\subseteq\bI_{i_2,h_2}$. Then, there exists a team-optimal strategy $g_{1:H}^\ast$ of $\check{\cD}$ such that $g_{i,h}^\ast$ is a linear function of $\bI_{i,h}$ for all $i\in[n],h\in[H]$. This holds even when the covariance matrices $\Sigma_1,\{\Sigma_{i,h}\}_{i\in[0:n],h\in[H]}$ and the cost matrices $\{Q_h^1\}_{h\in[H+1]},\{Q_h^2\}_{h\in[H]}$ are singular.
\end{corollary}

\begin{proof}
We follow \cite{ho1972team}, and only modify its analysis of the static problem to allow singular covariance and cost matrices. Let $\bxi$ contain $\bX_1$ and $\{\bW_{i,h}\}_{i\in[0:n],h\in[H]}$, which is a zero-mean Gaussian random vector with a possibly singular covariance matrix. From Equation~\eqref{equ: LQG_system_dynamics}, $\bX_{1:H+1}=\cA\bxi+\cB\bU_{1:H}$ for some fixed matrices $\cA,\cB$, and the expected total cost under $\bU_{1:H}$ is
\begin{align*}
    \cJ(\bU_{1:H}):=\EE\left[\bX_{1:H+1}^\top Q_{1:H+1}^1\bX_{1:H+1}+\bU_{1:H}^\top Q_{1:H}^2\bU_{1:H}\right],
\end{align*}

\noindent\textbf{Part 1:} we reduce $\check{\cD}$ to a static team problem \cite{staticteam}. From Equation~\eqref{equ: LQG_system_dynamics}, we can write $\bI_{i,h}=\cH_{i,h}\bxi+\sum_{(j,t)\in\cR_{i,h}}\cH_{i,h}^{j,t}\bU_{j,t}$ for some fixed matrices $\cH_{i,h}$ and $\cH_{i,h}^{j,t}\neq\mathbf{0}$, where $\cR_{i,h}\subseteq[n]\times[h-1]$ is the set of agents whose control actions influence $\bI_{i,h}$. 
The static information $\hat{\bI}_{i,h}:=\cH_{i,h}\bxi$ is the value of $\bI_{i,h}$ when all the control actions are $\mathbf{0}$, and we call $\bU_{1:H}$ \emph{static} if each $\bU_{i,h}$ is a function of $\hat{\bI}_{i,h}$. Since the partial nestedness gives $\bI_{j,t}\subseteq\bI_{i,h}$ for any $(j,t)\in\cR_{i,h}$, similarly as the standard static reduction  in \cite{ho1972team}, we can show that the control actions generated by any strategy are static, and conversely, that for any matrices $\{\cF_{i,h}\}_{i\in[n],h\in[H]}$, the static control actions $\bU_{i,h}=\cF_{i,h}\hat{\bI}_{i,h}$ are generated by a linear strategy, which recovers $\hat{\bI}_{i,h}$ by subtracting $\sum_{(j,t)\in\cR_{i,h}}\cH_{i,h}^{j,t}\bU_{j,t}$ from $\bI_{i,h}$.
Hence, it suffices to show that there exist $\{\cF_{i,h}^\ast\}_{i\in[n],h\in[H]}$ such that $\bU_{i,h}^\ast:=\cF_{i,h}^\ast\hat{\bI}_{i,h}$ minimizes $\cJ$ over all static control actions: letting $g_{1:H}^\ast$ be the linear strategy generating $\bU_{1:H}^\ast$, we then have $J_{\check{\cD}}(g_{1:H})\ge\cJ(\bU_{1:H}^\ast)=J_{\check{\cD}}(g_{1:H}^\ast)$ for any strategy $g_{1:H}$.

\noindent\textbf{Part 2:} we solve the static problem. When $\bU_{i,h}=\cF_{i,h}\hat{\bI}_{i,h}$ for all $i\in[n],h\in[H]$, the cost $\cJ$ is a convex quadratic function of $\{\cF_{i,h}\}_{i\in[n],h\in[H]}$, which may not be strictly convex since $\{Q_h^2\}_{h\in[H]}$ and the covariance matrices of $\{\hat{\bI}_{i,h}\}_{i\in[n],h\in[H]}$ may be singular. However, $\cJ$ is bounded below by $0$, so its linear term must vanish on the kernel of its quadratic term, otherwise $\cJ$ can be negative. Therefore, $\cJ$ attains its minimum at some $\{\cF_{i,h}^\ast\}_{i\in[n],h\in[H]}$. Let $\bU_{i,h}^\ast:=\cF_{i,h}^\ast\hat{\bI}_{i,h}$, $\bX_{1:H+1}^\ast:=\cA\bxi+\cB\bU_{1:H}^\ast$, and $\bG:=\cB^\top Q_{1:H+1}^1\bX_{1:H+1}^\ast+Q_{1:H}^2\bU_{1:H}^\ast$,  whose block corresponding to $\bU_{i,h}$ is denoted by $\bG_{i,h}$. For any static $\bU_{1:H}$, letting $\Delta\bU_{1:H}:=\bU_{1:H}-\bU_{1:H}^\ast$ and $\Delta\bX_{1:H+1}:=\cB\Delta\bU_{1:H}$, we have
\begin{align}
    \cJ(\bU_{1:H})-\cJ(\bU_{1:H}^\ast)=2\EE[\Delta\bU_{1:H}^\top\bG]+\EE[\Delta\bX_{1:H+1}^\top Q_{1:H+1}^1\Delta\bX_{1:H+1}+\Delta\bU_{1:H}^\top Q_{1:H}^2\Delta\bU_{1:H}].\label{eq: static cost difference}
\end{align}
For any $i\in[n],h\in[H]$, matrix $K$, and $s\in\RR$, replacing $\cF_{i,h}^\ast$ by $\cF_{i,h}^\ast+sK$ makes the right-hand side of Equation~\eqref{eq: static cost difference} equal to $2s\EE[(K\hat{\bI}_{i,h})^\top\bG_{i,h}]$ plus a nonnegative multiple of $s^2$, which is nonnegative for all $s\in\RR$ due to the optimality of $\{\cF_{i,h}^\ast\}_{i\in[n],h\in[H]}$. Hence, $\EE[(K\hat{\bI}_{i,h})^\top\bG_{i,h}]=0$ for any $K$, i.e., $\EE[\bG_{i,h}\hat{\bI}_{i,h}^\top]=\mathbf{0}$. Since $\bG_{i,h}$ and $\hat{\bI}_{i,h}$ are linear in $\bxi$, they are jointly Gaussian with zero mean, so this uncorrelatedness implies their independence even if their covariance matrix is singular. Thus, $\EE[\bG_{i,h}\given\hat{\bI}_{i,h}]=\EE[\bG_{i,h}]=\mathbf{0}$. Therefore, for any static $\bU_{1:H}$, it holds that $\EE[\Delta\bU_{1:H}^\top\bG]=\sum_{h=1}^H\sum_{i=1}^n\EE\big[\Delta\bU_{i,h}^\top\EE[\bG_{i,h}\given\hat{\bI}_{i,h}]\big]=0$, and Equation~\eqref{eq: static cost difference} gives $\cJ(\bU_{1:H})\ge\cJ(\bU_{1:H}^\ast)$, which completes the proof.
\end{proof}

\section{Deferred Details of  \S\ref{sec: hardness and assumptions}}
\label{sec: appendix_hardness_result_proof}

We start with the following auxiliary lemma. 

\begin{lemma}
    For any $n_1,n_2\in\NN$, let $Q_1\in\RR^{n_1\times n_1}$ and $Q_2\in\RR^{n_2\times n_2}$ be symmetric positive-semidefinite matrices, and let $A,B$ be matrices of compatible dimensions. Then, $\ker{(Q_1+B^\top Q_2B)}\subseteq\ker{(A^\top Q_2B)}$. \label{lemma: kernel_inclusion} 
\end{lemma}
\begin{proof}
    For any $v\in\ker{(Q_1+B^\top Q_2B)}$, we have $v^\top Q_1v+v^\top B^\top Q_2Bv=0$. Since $Q_1\succeq 0, Q_2\succeq 0$, we know $v^\top Q_1v=v^\top B^\top Q_2Bv=0$. Then, from $v^\top B^\top Q_2Bv=0$, we have $Q_2^{\frac{1}{2}}Bv=0$ and thus $Q_2Bv=0$, which further implies  that $A^\top Q_2Bv=0$ and thus $v\in \ker{(A^\top Q_2B)}$. 
    This completes the proof. 
\end{proof}
\subsubsection{Proof of Lemma \ref{lemma: nonqc}}  
\begin{proof} To prove this lemma, we leverage the following proposition.

\begin{proposition}[Adapted from Theorems 1 and 2 of \cite{witsenhausen1968counterexample}]\label{prop: W's counterexample}
    There exists a decentralized LQG control problem with $n=2, H=2$ and parameters $k$ and $\sigma_0$ such that no linear control strategy is team-optimal: 
    \begin{align*}
        &\text{All random variables are scalar, }\bX_1\sim\cN(0,\sigma_0^2), \bX_2=\bX_1+\bU_{1,1}, \bX_3=\bX_2-\bU_{2,2},\\
        &\bY_{1,1}=\bX_1, \bY_{2,1}=0,  \bY_{1,2}=0, \bY_{2,2}=\bX_2+\bW_{2,2}, \bW_{2,2}\sim\cN(0,1),\\
        &c_1=k^2\bU_{1,1}^2,c_2=0, c_3=\bX_3^2, ~~\text{no information sharing between agents.}
    \end{align*}
\end{proposition}

Then, using the parameters $k$ and $\sigma_0$  specified in Proposition \ref{prop: W's counterexample}, we construct a JCCO problem $\cD$ with $n=2, H=2$, with the system dynamics being defined as 
\begin{equation*}
    \begin{aligned}
           &\qquad  
       \bX_1\sim\cN(0,\sigma_0^2), \bX_2=\bX_1+\bU_{1,1}, \bX_3=\bX_2-\bU_{2,2},\\
       &\bY_{1,1}=\bX_1, \bY_{2,1}=0, \bY_{1,2}=0, \bY_{2,2}=\bX_2+\bW_{2,2}, \bW_{2,2}\sim\cN(0,1). 
    \end{aligned}
\end{equation*}
    For the communication part, suppose the baseline sharing is null and $
        \cM_{i,1}=\{0,1\}, \cM_{1,2}=\{0,1\}^3, \cM_{2,2}=\{0,1\}^2$, 
        such that  $\bM_{i,1}$ represents whether agent $i$ shares $\bY_{i,1}$, each digit of $\bM_{1,2}$ represents whether agent $1$ shares $\bY_{1,1},\bU_{1,1},\bY_{1,2}$, and each digit of $\bM_{2,2}$ represents whether agent $2$ shares $\bY_{2,1},\bY_{2,2}$. 
        For control costs, we set $c_1=k^2\bU_{1,1}^2, c_2=0, c_3=\bX_3^2$,  where $k$ is specified as in Proposition \ref{prop: W's counterexample}; for communication costs, we set  $\kappa_h=(\sigma_0^2+1)\mathds{1}[\bM_{h} \text{~is not all }0], \forall h\in[2]$. The information available to each agent at each timestep is defined as: $\forall i\in[2], \bI_{i,1^-}=\{\bY_{i,1}\},\bI_{i,1^+}=\bI_{i,1^-}\cup\bZ_1^a$, and $\bI_{1,2^-}=\bI_{1,1^+}\cup\{\bU_{1,1},\bY_{1,2}\}, \bI_{2,2^-}=\bI_{2,1^+}\cup\{\bY_{2,2}\}, \bI_{i,2^+}=\bI_{i,2^-}\cup\bZ_2^a$.
        Then, we can verify that:
        \begin{itemize}
            \item   \textbf{$\cD$  satisfies Assumption \ref{ass: evolution rule}.}
            \item \textbf{$\cD$ satisfies Assumption \ref{ass: useless action}:} The only zero input coefficient relevant to a later decision is the scalar $B_{2,1}=0$, and $\bU_{2,1}$ does not appear in any later information set.
            \item \textbf{$\cD$ satisfies Assumption \ref{ass: non-degeneracy}:} It holds that  $E_{1,2}B_{2,1}=B_{2,1}=0, E_{2,2}B_{1,1}\neq0$ and $B_{1,1}\neq 0$.
            \item \textbf{$\cD$ does not have PN IS:} If there is no additional sharing, agent $(1,1)$ influences agent $(2,2)$, but $\bI_{1,1^-}\nsubseteq \bI_{2,2^-}$.
        \end{itemize}

    For any communication strategy $g_{1:2}^m$, if it chooses any $\bM_{i,h}$ with any digit non-zero (namely, sharing \emph{some}  information via additional sharing), then it cannot achieve the optimum, since it will suffer from a communication cost of $\kappa_h\ge \sigma_0^2+1$, which is larger than the total expected cost when the agents choose $\bU_{1,1}=\bU_{2,2}=0$,  and they do not have any additional sharing. Therefore, the optimal communication strategy $g_{1:2}^{m,\ast}$ must yield no additional sharing. The remaining control subproblem is exactly Proposition~\ref{prop: W's counterexample}; hence every team-optimal control strategy is nonlinear. This completes the proof. 
\end{proof}

\subsubsection{Proof of Lemma \ref{lemma: useless action}}
\begin{proof}
    This proof consists of two {\bf Parts}. {\bf Part 1:} we construct a JCCO problem $\cD_1$ whose optimal control strategy is nonlinear; {\bf Part 2:} we construct another JCCO problem $\cD_2$ whose team-optimal strategy does not exist.
    
    \vspace{5pt}
        \noindent\textbf{Part 1:} Consider an $H=2,n=2$ JCCO problem $\cD_1$ with the system dynamics being  defined as 
    \begin{align*}
        &\forall h\in[2], \bX_h\in\RR^2, \bU_{1,1},\bU_{1,2}\in\RR, \bU_{2,1},\bU_{2,2}\in\RR^2, \bX_2=\bX_1,\bX_3=\bX_2-\bU_{2,2},\\
        &\bX_1\sim\cN(\mathbf{0},\II), \bY_{1,1}=\bX_1,\bY_{1,2}=\bY_{2,1}=\bY_{2,2}=0, c_1=c_2=0,c_3=\bX_3^\top\bX_3.
    \end{align*}
        For the communication part, suppose the baseline sharing is null and $\cM_{i,1}=\{0,1\}, \cM_{i,2}=\{0,1\}^3, \forall i\in[2]$, where $\bM_{i,1}$ represents whether agent $i$ shares $\bY_{i,1}$ and each digit of $\bM_{i,2}$ represents whether agent $i$ shares $\bY_{i,1},\bU_{i,1}$, and $\bY_{i,2}$, respectively. Set $\kappa_1=\mathds{1}[\bM_{1,1}=1\text{ or }\bM_{2,1}=1]$ and $\kappa_2=\mathds{1}[\bM_{1,2}\notin\{(0,0,0),(0,1,0)\}\text{ or }\bM_{2,2}\neq(0,0,0)]$.  Thus $(0,1,0)$ shares only $\bU_{1,1}$ at zero cost, while every other nonempty sharing incurs cost one. The information available to each agent at each timestep is defined as $\forall i\in[2], \bI_{i,1^-}=\{\bY_{i,1}\},\bI_{i,1^+}=\bI_{i,1^-}\cup\bZ_1^a$, and $\bI_{1,2^-}=\bI_{1,1^+}\cup\{\bU_{1,1},\bY_{1,2}\}, \bI_{2,2^-}=\bI_{2,1^+}\cup\{\bU_{2,1}, \bY_{2,2}\}, \bI_{i,2^+}=\bI_{i,2^-}\cup\bZ_2^a$.

        Then, we can verify that:
        \begin{itemize}
            \item   \textbf{$\cD_1$ satisfies Assumption \ref{ass: evolution rule}.}
            \item \textbf{$\cD_1$ has PN IS:} If there is no additional sharing, then for any $i_1,i_2\in[2]$, $\bI_{i_1,1^-}\subseteq \bI_{i_2,2^-}$ if $i_1=i_2$, and otherwise agent $(i_1,1)$ does not influence agent $(i_2,2)$. 
            \item \textbf{$\cD_1$ does not satisfy Assumption \ref{ass: useless action}:} $B_{1,1}=B_{2,1}=0$, but $\bU_{1,1}\in\bI_{1,2^-}$ and $\bU_{2,1}\in\bI_{2,2^-}$.
            \item \textbf{$\cD_1$ satisfies Assumption \ref{ass: non-degeneracy}:} For any $i\in[2]$, $E_{-i,2}B_{i,1}=B_{i,1}=0$.
        \end{itemize}
        Now, we aim to show that the optimal control strategy is nonlinear. Firstly, we can construct a strategy $(g_{1:2}^{m,\ast}, g_{1:2}^{a,\ast})$ as
        \begin{align*}
            g_{1,1}^{m,\ast}=g_{2,1}^{m,\ast}=0,g_{1,2}^{m,\ast}=(0,1,0), g_{2,2}^{m,\ast}=(0,0,0), g_{1,1}^{a,\ast}(\bI_{1,1^+})=\varphi_{bi}(\bY_{1,1}), g_{2,2}^{a,\ast}(\bI_{2,2^+})=\varphi_{bi}^{-1}(\bU_{1,1}),
        \end{align*}
        where $\varphi_{bi}$ is a Borel isomorphism that pushes $\cN(0,I_2)$ to $\cN(0,1)$ (such an isomorphism exists between atomless standard probability spaces), and the other two controls are zero.  Hence the strategy is square-integrable and has total cost zero almost surely.  Since every cost is nonnegative, it is team-optimal. Such a Borel bijection cannot be linear: every linear $\TT:\RR^2\rightarrow\RR$ has a nontrivial kernel by the rank-nullity theorem \cite{Linear_Algebra}.
        
        However, we will show that for any strategy $(g_{1:2}^m,g_{1:2}^a)$ such that $g_{1:2}^a$ 
        is linear, $(g_{1:2}^m,g_{1:2}^a)$ cannot be team-optimal, i.e., $J_{\cD_1}(g_{1:2}^m,g_{1:2}^a)>0$. We prove it by contradiction, and suppose that $(g_{1:2}^{m,'},g_{1:2}^{a,'})$ satisfies that $J_{\cD_1}(g_{1:2}^{m,'},g_{1:2}^{a,'})=0$ and $g_{1:2}^{a,'}$ is linear. Then, it must hold that $g_{1,1}^{m,'}=g_{2,1}^{m,'}=0, g_{2,2}^{m,'}=(0,0,0), g_{1,2}^{m,'}=(0,0,0)$ or $(0,1,0)$, since otherwise the communication cost is at least 1. If $g_{1,2}^{m,'}=(0,0,0)$, changing it to $(0,1,0)$ preserves the original controls as feasible prescriptions, adds no communication cost, and therefore preserves the zero value.  We may thus assume $g_{1,2}^{m,'}=(0,1,0)$.  Then $\bI_{i,1^+}=\{\bY_{i,1}\}$ for $i\in[2]$ and $\bI_{2,2^+}=\{\bY_{2,1},\bU_{2,1},\bY_{2,2},\bU_{1,1}\}$. Since $g_{1:2}^{a,'}$ is linear, we can write: 
        \begin{align*}
            \bU_{1,1}=G_{1,1}^{1,1,y}\bY_{1,1}, \quad \bU_{2,1}=G_{2,1}^{2,1,y}\bY_{2,1}, \quad\bU_{2,2}=G_{2,2}^{2,1,y}\bY_{2,1}+G_{2,2}^{2,2,y}\bY_{2,2}+G_{2,2}^{1,1,u}\bU_{1,1}+G_{2,2}^{2,1,u}\bU_{2,1},
        \end{align*}
        for some matrices $G_{1,1}^{1,1,y},G_{2,1}^{2,1,y}, G_{2,2}^{2,1,y},G_{2,2}^{2,2,y},G_{2,2}^{1,1,u}, G_{2,2}^{2,1,u}$ with proper dimensions. Then, we can write $\bU_{2,2}$ as $\bU_{2,2}=G_{2,2}^{1,1,u}G_{1,1}^{1,1,y}\bY_{1,1}$.
        Then, we can write $J_{\cD_1}(g_{1:2}^{m,'},g_{1:2}^{a,'})=\EE[((\II-G_{2,2}^{1,1,u}G_{1,1}^{1,1,y})\bX_1)^\top((\II-G_{2,2}^{1,1,u}G_{1,1}^{1,1,y})\bX_1)]=\tr((\II-G_{2,2}^{1,1,u}G_{1,1}^{1,1,y})^\top (\II-G_{2,2}^{1,1,u}G_{1,1}^{1,1,y}))$. Note that $G_{2,2}^{1,1,u}$ is a $2\times1$ matrix and $G_{1,1}^{1,1,y}$ is a $1\times2$ matrix. Therefore, we have rank($G_{2,2}^{1,1,u}G_{1,1}^{1,1,y}$)$\le1$, and then $\tr((\II-G_{2,2}^{1,1,u}G_{1,1}^{1,1,y})^\top (\II-G_{2,2}^{1,1,u}G_{1,1}^{1,1,y}))\ge 1$ due to Eckart-Young Theorem \cite{eckart1936approximation}. This means that $J_{\cD_1}(g_{1:2}^{m,'},g_{1:2}^{a,'})> 0$ and leads to the contradiction. Hence, we know that there exists a team-optimal strategy of $\cD_1$ and the control strategy of any team-optimal strategy is nonlinear, which completes the proof of {\bf Part 1}.

        \vspace{5pt}
        \noindent\textbf{Part 2:}  Consider an $n=2,H=2$ JCCO problem $\cD_2$ with the system dynamics being defined as 
    \begin{align*}
         &\forall h\in[2], i\in[2], \bX_{h}, \bX_3, \bY_{i,h}, \bU_{i,h}\in\RR, \bX_2=\bX_1, \bX_3=\bX_2-\bU_{2,2}, \\
         &\bY_{1,1}=\bX_1, \bY_{1,2}=\bY_{2,1}=\bY_{2,2}=0,  \bX_1\sim\cN(0,\frac{1}{4}), c_1=\bU_{1,1}^2,c_2=0,c_3=\bX_3^2.
    \end{align*}
     For the communication part, similarly to {\bf Part 1}, suppose the baseline sharing is null and $\cM_{i,1}=\{0,1\}, \cM_{i,2}=\{0,1\}^3, \forall i\in[2]$, where $\bM_{i,1}$ represents whether agent $i$ shares $\bY_{i,1}$ and each digit of $\bM_{i,2}$ represents whether agent $i$ shares $\bY_{i,1},\bU_{i,1}$, and $\bY_{i,2}$, respectively. We set the communication costs as $\kappa_1=\mathds{1}[\bM_{1,1}=1\text{ or }\bM_{2,1}=1], \kappa_2=\mathds{1}[\bM_{1,2}\notin\{(0,0,0),(0,1,0)\}\text{ or }\bM_{2,2}\neq(0,0,0)]$. This communication cost means that if any information except $\bU_{1,1}$ is shared, then all agents will incur a communication cost of  $1$.  The information available to each agent at each timestep is defined as $\forall i\in[2], \bI_{i,1^-}=\{\bY_{i,1}\},\bI_{i,1^+}=\bI_{i,1^-}\cup\bZ_1^a, \bI_{i,2^-}=\bI_{i,1^+}\cup\{\bU_{i,1},\bY_{i,2}\}, \bI_{i,2^+}=\bI_{i,2^-}\cup\bZ_2^a$.
     
Then, we can verify that:
        \begin{itemize}
            \item   \textbf{$\cD_2$ satisfies Assumption \ref{ass: evolution rule}.}
            \item \textbf{$\cD_2$ has PN IS:} If there is no additional sharing, then for any $i_1,i_2\in[2]$, $\bI_{i_1,1^-}\subseteq \bI_{i_2,2^-}$ if $i_1=i_2$, and otherwise agent $(i_1,1)$ does not influence agent $(i_2,2)$. 
            \item \textbf{$\cD_2$ does not satisfy Assumption \ref{ass: useless action}:} $B_{1,1}=B_{2,1}=0$, but $\bU_{1,1}\in\bI_{1,2^-}$ and $\bU_{2,1}\in\bI_{2,2^-}$.
            \item \textbf{$\cD_2$ satisfies Assumption \ref{ass: non-degeneracy}:} For any $i\in[2]$, $E_{-i,2}B_{i,1}=B_{i,1}=0$.
        \end{itemize}

       We prove it by contradiction. We assume $(g_{1:2}^{m,\ast},g_{1:2}^{a,\ast})$ is a team-optimal strategy of $\cD_2$. Firstly, it holds that $g_{1,1}^{m,\ast}=g_{2,1}^{m,\ast}=0, g_{2,2}^{m,\ast}=(0,0,0), g_{1,2}^{m,\ast}=(0,0,0)$ or $(0,1,0)$, since otherwise $J_{\cD_2}(g_{1:2}^{m,\ast},g_{1:2}^{a,\ast})\ge \EE[\kappa_1+\kappa_2\given g_{1:2}^{m,\ast},g_{1:2}^{a,\ast}]\ge 1$; however, if we consider strategy $(g_{1:2}^m,g_{1:2}^a)$ that shares nothing and chooses $\bU_{i,h}=0,\forall i\in[2],h\in[2]$, then $J_{\cD_2}(g_{1:2}^m,g_{1:2}^a)=\EE[\bX_3^2\given g_{1:2}^m,g_{1:2}^a]=\EE[\bX_1^2]=\frac{1}{4}$. 
       Secondly, we can assume $g_{1,2}^{m,\ast}=(0,1,0)$, otherwise we can change it to be $(0,1,0)$ and it is still a team-optimal strategy, since additionally sharing $\bU_{1,1}$ enlarges the $\bI_{2,2^+}$ but incurs no communication cost. Therefore, we know that under $g_{1:2}^{m,\ast}$, we have $\bI_{i,1^+}=\{\bY_{i,1}\}, \forall i\in[2], \bI_{2,2^+}=\{\bY_{2,1},\bU_{2,1},\bY_{2,2},\bU_{1,1}\}$. Thirdly, if $\PP(\bU_{1,1}=0\given g_{1:2}^{m,\ast},g_{1:2}^{a,\ast})=1$, then we consider the realization $I_{2,2^+}=\{\bY_{2,1}=0,\bU_{2,1}=g_{2,1}^{a,\ast}(0), \bY_{2,2}=0,\bU_{1,1}=0\}$, and let $U_{2,2}=g_{2,2}^{a,\ast}(I_{2,2^+})$ be the realization of $\bU_{2,2}$ when $g_{2,2}^{a,\ast}$ takes the realization $I_{2,2^+}$ as input. Then, it holds that $\PP(\bU_{2,2}=U_{2,2}\given g_{1:2}^{m,\ast},g_{1:2}^{a,\ast})=1$, and we have
       \begin{align*}
       J_{\cD_2}(g_{1:2}^{m,\ast},g_{1:2}^{a,\ast})\ge \EE[\bX_3^2\given g_{1:2}^{m,\ast},g_{1:2}^{a,\ast}]=\EE[(\bX_1-\bU_{2,2})^2\given g_{1:2}^{m,\ast},g_{1:2}^{a,\ast}]=\EE[(\bX_1-U_{2,2})^2\given g_{1:2}^{m,\ast},g_{1:2}^{a,\ast}]\ge \EE[\bX_1^2]\ge \frac{1}{4}.
       \end{align*}
       However, we consider the strategy $(g_{1:2}^{m,\ast}, g_{1:2}^{a,'})$  with $g_{1:2}^{a,'}$ being defined as $g_{1,1}^{a,'}(\bI_{1,1^+})=\frac{\bY_{1,1}}{2}, g_{2,2}^{a,'}(\bI_{2,2^+})=2\bU_{1,1}$ if $\bU_{1,1}\in \bI_{2,2^+}$ and  otherwise $0$, and choose arbitrary $g_{2,1}^{a,'},g_{1,2}^{a,'}$.  Then, we can verify that
       \begin{align*}
           J_{\cD_2}(g_{1:2}^{m,\ast},g_{1:2}^{a,'})=\EE[c_1+c_3\given g_{1:2}^{m,\ast},g_{1:2}^{a,'}]=\EE[\bU_{1,1}^2+\bX_3^2\given g_{1:2}^{m,\ast},g_{1:2}^{a,'}]=\EE[\frac{1}{4}\bY_{1,1}^2+(\bX_1-2\cdot\frac{\bY_{1,1}}{2})^2]=\EE[\frac{1}{4}\bY_{1,1}^2]=\frac{1}{16}.
       \end{align*}
       Therefore, we know that it holds that $\PP(\bU_{1,1}=0\given g_{1:2}^{m,\ast},g_{1:2}^{a,\ast})<1$, and then $\EE[\bU_{1,1}^2\given g_{1:2}^{m,\ast},g_{1:2}^{a,\ast}]>0$. Lastly, we can construct a new strategy $g_{1:2}^a$ based on $g_{1:2}^{a,\ast}$ as 
       \begin{align*}
           &g_{1,1}^a(\bI_{1,1^+})=\frac{1}{2}g_{1,1}^{a,\ast}(\bI_{1,1^+}),\\
           &g_{2,2}^a(\bI_{2,2^+})=g_{2,2}^a(\bY_{2,1},\bU_{2,1},\bY_{2,2},\bU_{1,1})=g_{2,2}^{a,\ast}(\bY_{2,1},\bU_{2,1},\bY_{2,2},2\bU_{1,1}) \\
           &\qquad\text{ if $\bI_{2,2^+}=\{\bY_{2,1},\bU_{2,1},\bY_{2,2},\bU_{1,1}\}$, otherwise 0.}
       \end{align*}
       {All other control rules remain unchanged.} This construction means that $g_{1,1}^a$ will choose $\bU_{1,1}$ to be half of the $\bU_{1,1}$ chosen by $g_{1,1}^{a,\ast}$, but $g_{2,2}^{a,\ast},g_{2,2}^a$ choose the same $\bU_{2,2}$. Therefore, we can verify that
       \begin{align*}
           J_{\cD_2}(g_{1:2}^{m,\ast},g_{1:2}^a)&=\EE[\bU_{1,1}^2\given g_{1:2}^{m,\ast},g_{1:2}^a]+\EE[(\bX_1-\bU_{2,2})^2\given g_{1:2}^{m,\ast},g_{1:2}^a]\\
           &=\EE[(\bU_{1,1}/2)^2\given g_{1:2}^{m,\ast},g_{1:2}^{a,\ast}]+\EE[(\bX_1-\bU_{2,2})^2\given g_{1:2}^{m,\ast},g_{1:2}^{a,\ast}]\\
           &<\EE[\bU_{1,1}^2\given g_{1:2}^{m,\ast},g_{1:2}^{a,\ast}]+\EE[(\bX_1-\bU_{2,2})^2\given g_{1:2}^{m,\ast},g_{1:2}^{a,\ast}]=J_{\cD_2}(g_{1:2}^{m,\ast},g_{1:2}^{a,\ast}),
       \end{align*}
       where the inequality holds because $\EE[\bU_{1,1}^2\given g_{1:2}^{m,\ast},g_{1:2}^{a,\ast}]>0$ as proved before.  This means that for any team-optimal strategy $g_{1:2}^{m,\ast},g_{1:2}^{a,\ast}$, we can construct a strategy $(g_{1:2}^{m,\ast},g_{1:2}^a)$ that is strictly better than it. Therefore, the team-optimal strategy does not exist. 
\end{proof}

\subsubsection{Proof of Lemma \ref{lemma: nondegeneracy}}\begin{proof}
    Consider a JCCO problem $\cD$ with $n=2,H=2$, parameters $k,\sigma_0$ specified in Proposition \ref{prop: W's counterexample}, and the system dynamics being defined as
    \begin{align*}
         \forall h\in[2], i\in[2], \bX_{h}, \bX_3, \bY_{i,h}, \bU_{i,h}\in\RR, \bX_2=\bX_1+\bU_{1,1}, \bX_3=\bX_2-\bU_{2,2},\\
         \bY_{1,1}=\bX_1, \bY_{1,2}=\bX_2+\bW_{1,2}, \bY_{2,1}=\bY_{2,2}=0,  \bX_1\sim\cN(0,\sigma_0^2),\bW_{1,2}\sim\cN(0,1). 
    \end{align*} 
    
    For the communication part, baseline sharing is null and
    \begin{align*}
        \cM_{i,1}=\{0,1\}, \cM_{1,2}=\{0,1\}^3, \cM_{2,2}=\{0,1\}^2,
    \end{align*}
    where $\bM_{i,1}$ represents whether agent $i$ shares $\bY_{i,1}$. Each digit of $\bM_{1,2}$ represents whether agent $1$ shares $\bY_{1,1},\bU_{1,1}$, and $\bY_{1,2}$. Each digit of $\bM_{2,2}$ represents whether agent $2$ shares $\bY_{2,1}$ and $\bY_{2,2}$. 

    For the control costs, we set $c_1=k^2\bU_{1,1}^2, c_2=0, c_3=\bX_3^2$; for the communication costs, we set $\kappa_1=(\sigma_0^2+1)\mathds{1}[\bM_{1,1}=1\text{ or }\bM_{2,1}=1], \kappa_2=(\sigma_0^2+1)\mathds{1}[\bM_{2,2}\neq (0,0) \text{~or~}\bM_{1,2}\notin\{(0,0,0),(0,0,1)\}]$. This communication cost means that if any information except $\bY_{1,2}$ is shared through additional sharing, it will incur a communication cost of $\sigma_0^2+1$. The information available to each agent at each timestep is defined as: $\forall i\in[2], \bI_{i,1^-}=\{\bY_{i,1}\}, \bI_{i,1^+}=\bI_{i,1^-}\cup \bZ_1^a,$ and $\bI_{1,2^-}=\bI_{1,1^+}\cup \{\bU_{1,1},\bY_{1,2}\}, \bI_{2,2^-}=\bI_{2,1^+}\cup\{\bY_{2,2}\}, \bI_{i,2^+}=\bI_{i,2^-}\cup \bZ_2^a$. 
        Then, we can verify that:
        \begin{itemize}
            \item   \textbf{$\cD$ satisfies Assumption \ref{ass: evolution rule}.}
            \item \textbf{$\cD$ has PN IS:} If there is no additional sharing, then for any $i_1,i_2\in[2]$, $\bI_{i_1,1^-}\subseteq \bI_{i_2,2^-}$ if $i_1=i_2$; otherwise, agent $(i_1,1)$ does not influence agent $(i_2,2)$. 
            \item \textbf{$\cD$ satisfies Assumption \ref{ass: useless action}:} The only zero input coefficient relevant to a later decision is $B_{2,1}=0$, and $\bU_{2,1}$ does not appear in any later information set.
            \item \textbf{$\cD$ does not satisfy Assumption \ref{ass: non-degeneracy}:} It holds that $E_{2,2}B_{1,1}=0$ but $B_{1,1}\neq 0$.
        \end{itemize}
        
First, no team-optimal communication strategy shares any information other than $\bY_{1,2}$ through additional sharing, since otherwise it will incur a communication cost $\kappa_h\ge \sigma_0^2+1$, which is larger than the total cost when each agent $i$ chooses $\bU_{i,h}=0, \forall h\in[2]$ and shares nothing via additional sharing.

Second,  we analyze the optimal control strategy under two types of communication strategies.

1) Suppose the optimal communication strategy $g_{1:2}^{m,\ast}$ yields that agents share $\bY_{1,2}$ through additional sharing, i.e. $\bZ_1^a=\emptyset, \bZ_2^a=\{\bY_{1,2}\}$. Then, finding a team-optimal strategy of $\cD$ can be reduced to finding an optimal strategy of the Decentralized LQG problem in Proposition \ref{prop: W's counterexample}, where the optimal control strategy is nonlinear.

2) Suppose the optimal communication strategy $g_{1:2}^{m,\ast}$ yields that agents share nothing through additional sharing, i.e.,  $\bZ_1^a=\bZ_2^a=\emptyset$. Then, let $g_{1:2}^{a,\ast}$ be the optimal control strategy, and let $g_{1:2}^{m,'}$ be the communication strategy that additionally shares $\bY_{1,2}$. Since both $g_{1:2}^{m,'}$ and $g_{1:2}^{m,\ast}$ lead to no communication cost and $g_{1:2}^{m,'}$ enlarges $\bI_{2,2^+}$, it follows that $J_{\cD}(g_{1:2}^{m,\ast},g_{1:2}^{a,\ast})\ge J_{\cD}(g_{1:2}^{m,'},g_{1:2}^{a,\ast})$, which means $(g_{1:2}^{m,'},g_{1:2}^{a,\ast})$ is an optimal strategy. Therefore, from case 1) we know that $g_{1:2}^{a,\ast}$ is nonlinear, which completes the proof.  
\end{proof}

\subsubsection{Proof of Theorem \ref{thm: preserve PN}}
\begin{proof}
    Fix any communication strategy $g_{1:H}^m\in \cG_{1:H}^m$ and let $M_{1:H}$ be the communication actions chosen by $g_{1:H}^m\in \cG_{1:H}^m$.  Then, $\cD$ becomes a   decentralized LQG problem denoted by $\cD(g_{1:H}^m)$, with the same system dynamics as $\cD$  and the following information evolution: for any $h\in[H]$
    \begin{enumerate}[(a)]
        \item The common information in $\cD(g_{1:H}^m)$ evolves as $\bar{\bC}_h=\cup_{t=1}^h \bar{\bZ}_t$, and $\bar{\bZ}_h=\bar{\chi}_h(\bar{\bP}_{h-1},\bar{\bU}_{h-1},\bar{\bY}_{h})$ for some fixed projection function $\bar{\chi}_h$. Here, $\bar{\chi}_h$ is defined as 
        \begin{align*}
        &\bar{\chi}_h(\bar{\bP}_{h-1},\bar{\bU}_{h-1},\bar{\bY}_{h})=\chi_h(\bar{\bP}_{h-1},\bar{\bU}_{h-1},\bar{\bY}_{h})\cup \phi_h(M_h,\zeta_h(\bar{\bP}_{h-1},\bar{\bU}_{h-1},\bar{\bY}_{h})).
        \end{align*}
        Recall that $\chi_h,\zeta_h,\phi_h$ are   defined in Assumption \ref{ass: evolution rule}. Such a $\bar{\chi}_h$ is a fixed projection function since $\chi_h, \{\phi_{i,h}(M_{i,h},\cdot)\}_{i\in[n]}$ are fixed projection functions, and for any $P_{h^-}\in \cP_{h^-}, \phi_h(M_h,P_{h^-})=\cup_{i=1}^n\phi_{i,h}(M_{i,h},P_{i,h^-})$. 
        
        \item The private information in $\cD(g_{1:H}^m)$ evolves as follows. For any $i\in[n], \bar{\bP}_{i,h}=\bar{\zeta}_{i,h}(\bar{\bP}_{i,h-1},\bar{\bU}_{i,h-1},\bar{\bY}_{i,h})$  for some fixed projection function $\bar{\zeta}_{i,h}$ and thus $\bar{\bP}_{h}=\bar{\zeta}_{h}(\bar{\bP}_{h-1},\bar{\bU}_{h-1},\bar{\bY}_{h})$ for some fixed projection function   $\bar{\zeta}_h$. Here,  $\bar{\zeta}_{i,h}$ is defined as 
        \begin{align*}
            &\bar{\zeta}_{i,h}(\bar{\bP}_{i,h-1},\bar{\bU}_{i,h-1},\bar{\bY}_{i,h})=\zeta_{i,h}(\bar{\bP}_{i,h-1},\bar{\bU}_{i,h-1},\bar{\bY}_{i,h})\backslash\phi_{i,h}(M_{i,h},\zeta_{i,h}(\bar{\bP}_{i,h-1},\bar{\bU}_{i,h-1},\bar{\bY}_{i,h})).
        \end{align*}
        Recall that $\zeta_{i,h},\phi_{i,h}$ are defined in Assumption \ref{ass: evolution rule}. Such a $\overline{\zeta}_{i,h}$ is a fixed projection function since $\zeta_{i,h}, \phi_{i,h}(M_{i,h},\cdot)$ are fixed projection functions.
        \item  For each agent $i\in[n]$, $\bar{\bY}_{i,h}\in\bar{\bI}_{i,h}$ for $h\in[H]$, and $\bar{\bI}_{i,h}\subseteq\bar{\bI}_{i,h+1}$ for $h\in[H-1]$. These properties still hold since additional sharing only moves some private information into common information but does not remove any information.
    \end{enumerate}
    Then, we know that $\cD(g_{1:H}^m)$ is a decentralized LQG problem with some information sharing due to both the baseline sharing and the communication actions $M_{1:H}$.     
    It remains to prove partial nestedness.  Use $~\check{}~$ for the baseline problem $\check{\cD}$.  Suppose that $\bar{\bU}_{i_1,h_1}$ influences $\bar{\bI}_{i_2,h_2}$, where $h_1<h_2$.  If $B_{i_1,h_1}=0$, the action has no state effect, and Assumption~\ref{ass: useless action} excludes its later appearance as a shared label, a contradiction.  Hence $B_{i_1,h_1}\neq0$.  Assumption~\ref{ass: non-degeneracy} gives some $j\neq i_1$ with $E_{j,h_1+1}B_{i_1,h_1}\neq0$, so $\check{\bU}_{i_1,h_1}$ influences $\check{\bY}_{j,h_1+1}\in\check{\bI}_{j,h_1+1}$.  Baseline partial nestedness yields $\check{\bI}_{i_1,h_1}\subseteq\check{\bI}_{j,h_1+1}$.  Since $j\neq i_1$, every label private to agent $i_1$ can enter agent $j$'s information only through common information due to Assumption \ref{ass: evolution rule}, then it holds $\check{\bI}_{i_1,h_1}\subseteq\check{\bC}_{h_1+1}$.  Moreover, every label in $\bar{\bI}_{i_1,h_1}\setminus\check{\bI}_{i_1,h_1}$ arrived through additional sharing and lies in $\bar{\bC}_{h_1}$, and then it holds
    \[
      \bar{\bI}_{i_1,h_1}\subseteq\bar{\bC}_{h_1+1}\subseteq
      \bar{\bC}_{h_2}\subseteq\bar{\bI}_{i_2,h_2}.
    \]
    Thus $\cD(g_{1:H}^m)$ is partially nested.  Corollary~\ref{cor:degenerate-PN-linearity} supplies a linear team optimum for every fixed schedule.  Since the open-loop schedule set is finite, a minimizing schedule exists, and its associated linear control strategy proves the final assertion.
\end{proof}

\section{Deferred Details of  \S\ref{sec: open_loop}}
\label{sec: appendix_open_loop_proof}
For the appendix derivations only, we use the following cost-free terminal convention after the last decision:
\[
\tilde{\bC}_{H+1}:=\tilde{\bC}_H\cup\{\tilde{\bU}_H,\tilde{\bX}_{H+1}\},\qquad
\tilde{\bP}_{i,H+1}:=\emptyset.
\]
Equivalently, set $\tilde{\bY}_{H+1}:=\tilde{\bX}_{H+1}$, $\tilde E_{H+1}:=I$, $\tilde{\bW}_{1:n,H+1}:=\bm0$, and $\tilde{\bZ}_{H+1}:=(\tilde{\bU}_H,\tilde{\bY}_{H+1})$. Then $\tilde{\bS}_{H+1}=\tilde{\bTheta}_{H+1}=\tilde{\bX}_{H+1}$, $\tilde\Sigma_{H+1}=0$, and $\II_{x,H+1}=I$. This convention is only a device for writing the terminal-cost calculation in the same form as the preceding backward steps; it introduces no additional decision and changes neither the strategy space nor the cost.
\subsubsection{Proof of Lemma \ref{lemma: equivalence of PN and sPN}}
\begin{proof}
    From the construction, the system dynamics and control cost are the same for both decentralized LQG problems, and it holds that $\bar{\bI}_{i,h}\subseteq \tilde{\bI}_{i,h}$ for any $i\in[n],h\in[H]$. Therefore, the agents in $\tilde{\cD}(g_{1:H}^m)$ have larger strategy spaces than those in $\cD(g_{1:H}^m)$, which implies  $\min_{\tilde{g}_{1:H}\in\tilde{\cG}_{1:H}}J_{\tilde{\cD}(g_{1:H}^m)}(\tilde{g}_{1:H})\le \min_{\bar{g}_{1:H}\in\bar{\cG}_{1:H}}J_{\cD(g_{1:H}^m)}(\bar{g}_{1:H})$.
    
    Now, we want to show that for any optimal $\tilde{g}_{1:H}^\ast$ of problem $\tilde{\cD}(g_{1:H}^m)$,  we can construct $\bar{g}_{1:H}^\ast=\varphi(\tilde{g}_{1:H}^\ast,\cD(g_{1:H}^m))$ to be a linear optimal strategy of $\cD(g_{1:H}^m)$, and $J_{\cD(g_{1:H}^m)}(\bar{g}_{1:H}^\ast)=J_{\tilde{\cD}(g_{1:H}^m)}(\tilde{g}_{1:H}^\ast)$.
    For any optimal linear strategy $\tilde{g}_{1:H}^\ast\in\argmin_{\tilde{g}_{1:H}\in \tilde{\cG}_{1:H}}J_{\tilde{\cD}(g_{1:H}^m)}(\tilde{g}_{1:H})$ of $\tilde{\cD}(g_{1:H}^m)$, we recursively construct a strategy $\bar{g}_{1:H}^\ast=\varphi(\tilde{g}_{1:H}^\ast,\cD(g_{1:H}^m))$ of $\cD(g_{1:H}^m)$. For $h=1$ and any $i\in[n]$, from the construction of $\tilde{\cD}(g_{1:H}^m)$, we know that $\tilde{\bI}_{i,1}=\bar{\bI}_{i,1}$ always holds. Then we can define $\tilde{g}_{i,1}^\ast(\tilde{I}_{i,1})=\bar{g}_{i,1}^\ast(\tilde{I}_{i,1})$ for any $i\in[n], \tilde{I}_{i,1}\in\tilde{\cI}_{i,1}=\bar{\cI}_{i,1}$. Then,  for any $i\in[n]$, $\tilde{g}_{i,1}^\ast$ and $\bar{g}_{i,1}^\ast$ output the same control action, namely,  $\tilde{\bU}_{i,1}=\bar{\bU}_{i,1}$, and $\bar{g}_{i,1}^\ast$ is linear. 
    
    Now, for any $h\in[2:H]$, assume that we have already constructed the linear strategy $\bar{g}_{1:h-1}^\ast$ from $\tilde{g}_{1:h-1}^\ast$ such that for any $i\in[n]$ and $t<h$, $\bar{g}_{i,t}^\ast$ and $\tilde{g}_{i,t}^\ast$ output the same actions. Now, we aim to construct $\bar{g}_h^\ast$ from $\tilde{g}_h^\ast$. From the construction of problem $\tilde{\cD}(g_{1:H}^m)$ from $\cD(g_{1:H}^m)$, for any $i\in[n]$, $\bar{\bI}_{i,h}\subseteq \tilde{\bI}_{i,h}$ and the only additional variables are the actions $\tilde{\bU}_{j,t}\in \tilde{\bI}_{i,h}\backslash\bar{\bI}_{i,h}$, where $j\in[n]$ and $t<h$. Furthermore, if $\tilde{\bU}_{j,t}\in \tilde{\bI}_{i,h}\backslash\bar{\bI}_{i,h}$, then $\bar{B}_{j,t}\neq \mathbf{0}$ and $\bar{\bI}_{j,t}\subseteq \bar{\bC}_h$. Since $\bar{g}_{1:h-1}^{\ast}$ outputs the same actions as $\tilde{g}_{1:h-1}^{\ast}$, we can write $\tilde{\bU}_{j,t}=\bar{\bU}_{j,t}=\bar{g}_{j,t}^{\ast}(\bar{\bI}_{j,t})$, which means that we can obtain $\tilde{\bU}_{j,t}$ from $\bar{g}_{j,t}^{\ast}$ and $\bar{\bI}_{j,t}$. Since  $\bar{g}_{j,t}^{\ast}$ is linear, we write $\bar{g}_{j,t}^{\ast}(\bar{\bI}_{j,t})=\sum_{\bar{\bY}_{k,s}\in \bar{\bI}_{j,t}}\bar{G}_{j,t}^{k,s,y}\bar{\bY}_{k,s}+\sum_{\bar{\bU}_{k,s}\in \bar{\bI}_{j,t}}\bar{G}_{j,t}^{k,s,u}\bar{\bU}_{k,s}$ for any $j\in[n]$ and $t<h$. Here, the sum $\sum_{\bar{\bY}_{k,s}\in \bar{\bI}_{j,t}}$ ranges over all $k\in[n]$ and $s\in[H]$ such that $\bar{\bY}_{k,s}\in\bar{\bI}_{j,t}$, and the other sums are interpreted similarly.
    Meanwhile, from the linearity of $\tilde{g}_{i,h}^{\ast}$ for any $i\in[n]$, we can write $\tilde{g}_{i,h}^{\ast}(\tilde{\bI}_{i,h})=\sum_{\tilde{\bY}_{j,t}\in \tilde{\bI}_{i,h}}\tilde{G}_{i,h}^{j,t,y}\tilde{\bY}_{j,t}+\sum_{\tilde{\bU}_{j,t}\in \tilde{\bI}_{i,h}}\tilde{G}_{i,h}^{j,t,u}\tilde{\bU}_{j,t}$  for some matrices $\tilde{G}_{i,h}^{j,t,y},\tilde{G}_{i,h}^{j,t,u}$. Then, we have
    \begin{align*}
        &\tilde{g}_{i,h}^{\ast}(\tilde{\bI}_{i,h})=\sum_{\tilde{\bY}_{j,t}\in \tilde{\bI}_{i,h}}\tilde{G}_{i,h}^{j,t,y}\tilde{\bY}_{j,t}+\sum_{\tilde{\bU}_{j,t}\in \tilde{\bI}_{i,h}}\tilde{G}_{i,h}^{j,t,u}\tilde{\bU}_{j,t}=\sum_{\bar{\bY}_{j,t}\in \bar{\bI}_{i,h}}\tilde{G}_{i,h}^{j,t,y}\bar{\bY}_{j,t}+\sum_{\bar{\bU}_{j,t}\in \bar{\bI}_{i,h}}\tilde{G}_{i,h}^{j,t,u}\bar{\bU}_{j,t}+\sum_{\bar{\bU}_{j,t}\in \tilde{\bI}_{i,h}\backslash\bar{\bI}_{i,h}}\tilde{G}_{i,h}^{j,t,u}\bar{\bU}_{j,t}\\
        &\qquad=\sum_{\bar{\bY}_{j,t}\in \bar{\bI}_{i,h}}\tilde{G}_{i,h}^{j,t,y}\bar{\bY}_{j,t}+\sum_{\bar{\bU}_{j,t}\in \bar{\bI}_{i,h}}\tilde{G}_{i,h}^{j,t,u}\bar{\bU}_{j,t}+\sum_{\bar{\bU}_{j,t}\in \tilde{\bI}_{i,h}\backslash\bar{\bI}_{i,h}}\tilde{G}_{i,h}^{j,t,u}\left(\sum_{\bar{\bY}_{k,s}\in \bar{\bI}_{j,t}}\bar{G}_{j,t}^{k,s,y}\bar{\bY}_{k,s}+\sum_{\bar{\bU}_{k,s}\in\bar{\bI}_{j,t}}\bar{G}_{j,t}^{k,s,u}\bar{\bU}_{k,s}\right).
    \end{align*}
    The second equality in the first line holds because $\tilde{g}_{1:h-1}^{\ast}$ and $\bar{g}_{1:h-1}^{\ast}$ output the same control actions, and the system dynamics of both problems are the same, so the observations are the same before agents take actions at timestep $h$. The equality between the first and the second lines follows from $\tilde{\bU}_{j,t}=\bar{g}_{j,t}^{\ast}(\bar{\bI}_{j,t})$ and substitution of the explicit form of $\bar{g}_{j,t}^{\ast}$. Also, if $\tilde{\bU}_{j,t}\in \tilde{\bI}_{i,h}\backslash\bar{\bI}_{i,h}$, then $\bar{\bI}_{j,t}\subseteq\bar{\bI}_{i,h}$, and any $\bar{\bY}_{k,s}\in \bar{\bI}_{j,t}, \bar{\bU}_{k,s}\in \bar{\bI}_{j,t}$ will also be included in $\bar{\bI}_{i,h}$. Therefore, we have
    \begin{align*}
        \tilde{g}_{i,h}^{\ast}(\tilde{\bI}_{i,h})&=\sum_{\bar{\bY}_{j,t}\in \bar{\bI}_{i,h}}\left(\tilde{G}_{i,h}^{j,t,y}+\sum_{\bar{\bU}_{k,s}\in\tilde{\bI}_{i,h}\backslash\bar{\bI}_{i,h}}\tilde{G}_{i,h}^{k,s,u}\mathds{1}[\bar{\bY}_{j,t}\in \bar{\bI}_{k,s}]\bar{G}_{k,s}^{j,t,y}\right)\bar{\bY}_{j,t}\\
        &\qquad+\sum_{\bar{\bU}_{j,t}\in \bar{\bI}_{i,h}}\left(\tilde{G}_{i,h}^{j,t,u}+\sum_{\bar{\bU}_{k,s}\in\tilde{\bI}_{i,h}\backslash\bar{\bI}_{i,h}}\tilde{G}_{i,h}^{k,s,u}\mathds{1}[\bar{\bU}_{j,t}\in \bar{\bI}_{k,s}]\bar{G}_{k,s}^{j,t,u}\right)\bar{\bU}_{j,t}.
    \end{align*}
    Note that in the equation above, if $\bar{\bY}_{j,t}\notin\bar{\bI}_{k,s}$, then it means that  $\bar{\bU}_{k,s}$ is not based on $\bar{\bY}_{j,t}$, and we can just assign $\bar{G}_{k,s}^{j,t,y}=\mathbf{0}$. Also, if $\bar{\bU}_{j,t}\notin\bar{\bI}_{k,s}$, then we assign $\bar{G}_{k,s}^{j,t,u}=\mathbf{0}$. Therefore, we can construct $\bar{g}_{i,h}^{ \ast}$ for any $i\in[n]$ as
    \begin{align*}
        \bar{g}_{i,h}^{\ast}(\bar{\bI}_{i,h})&=\sum_{\bar{\bY}_{j,t}\in \bar{\bI}_{i,h}}\bar{G}_{i,h}^{j,t,y}\bar{\bY}_{j,t}+\sum_{\bar{\bU}_{j,t}\in \bar{\bI}_{i,h}}\bar{G}_{i,h}^{j,t,u}\bar{\bU}_{j,t},\\
        \forall j\in[n], t\le h, \bar{G}_{i,h}^{j,t,y}&:=\tilde{G}_{i,h}^{j,t,y}+\sum_{\bar{\bU}_{k,s}\in\tilde{\bI}_{i,h}\backslash\bar{\bI}_{i,h}}\tilde{G}_{i,h}^{k,s,u}\mathds{1}[\bar{\bY}_{j,t}\in \bar{\bI}_{k,s}]\bar{G}_{k,s}^{j,t,y},\\
        \forall j\in[n], t\le h, \bar{G}_{i,h}^{j,t,u}&:=\tilde{G}_{i,h}^{j,t,u}+\sum_{\bar{\bU}_{k,s}\in\tilde{\bI}_{i,h}\backslash\bar{\bI}_{i,h}}\tilde{G}_{i,h}^{k,s,u}\mathds{1}[\bar{\bU}_{j,t}\in \bar{\bI}_{k,s}]\bar{G}_{k,s}^{j,t,u}.
    \end{align*}
    Then, $\bar{g}_h^{\ast}$ is linear and outputs the same actions as $\tilde{g}_h^{\ast}$. Recursively, we can construct $\bar{g}_h^{\ast}$ from $h=1$ to $H$ based on $\tilde{g}_{1:H}^\ast$ and the problem $\cD(g_{1:H}^m)$, and we denote this construction by a function $\varphi$, so that $\bar{g}_{1:H}^\ast=\varphi(\tilde{g}_{1:H}^\ast,\cD(g_{1:H}^m))$. 
    Since the dynamics and control costs are the same in the two problems, and $\bar{g}_{1:H}^\ast$ and $\tilde{g}_{1:H}^\ast$ output the same actions, $J_{\tilde{\cD}(g_{1:H}^m)}(\tilde{g}_{1:H}^\ast)=J_{\cD(g_{1:H}^m)}(\bar{g}_{1:H}^\ast)$. Also, we know that
    \begin{align*}
        J_{\tilde{\cD}(g_{1:H}^m)}(\tilde{g}_{1:H}^\ast)=\min_{\tilde{g}_{1:H}\in\tilde{\cG}_{1:H}}J_{\tilde{\cD}(g_{1:H}^m)}(\tilde{g}_{1:H})\le \min_{\bar{g}_{1:H}\in\bar{\cG}_{1:H}}J_{\cD(g_{1:H}^m)}(\bar{g}_{1:H})\le J_{\cD(g_{1:H}^m)}(\bar{g}_{1:H}^\ast).
    \end{align*}
    Therefore, $\bar{g}_{1:H}^\ast$ is a linear optimal strategy of $\cD(g_{1:H}^m)$.
\end{proof}
\subsubsection{Proof of Theorem \ref{thm: expansion}}
\begin{proof}
    The proof consists of three \textbf{Parts}: \textbf{Part 1:}   $\tilde{\cD}(g_{1:H}^m)$ is (strictly) PN; \textbf{Part 2:}   $\tilde{\cD}(g_{1:H}^m)$  has the information evolution rules; \textbf{Part 3:} $\tilde{\cD}(g_{1:H}^m)$ satisfies the SI-CIB condition. 

    \vspace{6pt}
    \noindent\textbf{Part 1}: We aim to prove that $\tilde{\cD}(g_{1:H}^m)$ is (strictly) partially-nested.\\
    Fix $i_1,i_2\in[n]$ and $h_1<h_2$ in $[H]$, and suppose that $\tilde{\bU}_{i_1,h_1}$ influences $\tilde{\bI}_{i_2,h_2}$. If $\tilde B_{i_1,h_1}=0$, the action has no state effect, while Assumption~\ref{ass: useless action} prevents it from appearing in any later information set, contradicting the assumed influence. Hence $\tilde B_{i_1,h_1}=\bar B_{i_1,h_1}\neq0$. Assumption~\ref{ass: non-degeneracy} gives some $i\neq i_1$ such that $\bar{\bU}_{i_1,h_1}$ influences $\bar{\bY}_{i,h_1+1}$. Since $h_1+1\le h_2$ and information is retained, $\bar{\bU}_{i_1,h_1}$ influences $\bar{\bI}_{i,h_2}$. Theorem~\ref{thm: preserve PN} therefore gives $\bar{\bI}_{i_1,h_1}\subseteq\bar{\bI}_{i,h_2}$. Since $i\neq i_1$, every information of agent $i_1$ can known by agent $i$'s information only if it lies in common information from Assumption \ref{ass: evolution rule}, so it holds $\bar{\bI}_{i_1,h_1}\subseteq\bar{\bC}_{h_2}$. Moreover, from Equation \eqref{eq: expansion}, every information in $\tilde{\bI}_{i_1,h_1}\setminus\bar{\bI}_{i_1,h_1}$ is common by construction. Thus, we have  $\tilde{\bI}_{i_1,h_1}\subseteq\tilde{\bC}_{h_2}\subseteq\tilde{\bI}_{i_2,h_2}$, and further have $\tilde{\bU}_{i_1,h_1}\in \tilde{\bC}_{h_2}\subseteq \tilde{\bI}_{i_2,h_2}$ and then prove that $\tilde\cD(g^m_{1:H})$ is strictly partially-nested.

    \vspace{6pt}
    \noindent
    \textbf{Part 2}: We want to prove $\tilde{\cD}(g_{1:H}^m)$  has the information evolution rules. From Theorem \ref{thm: preserve PN}, we know that the problem $\cD(g_{1:H}^m)$ has the information evolution rules with some projection functions $\{\bar{\chi}_h\}_{h\in[H]},\{\bar{\zeta}_{i,h}\}_{i\in[n],h\in[H]}$. 

    For any $i\in[n], h\in[H]$, because the corresponding label sets in $\tilde{\bZ}_h$ and $\tilde{\bP}_{i,h}$ are fixed, it suffices to prove that $\tilde{\bZ}_{h}\subseteq \tilde{\bP}_{h-1}\cup\{\tilde{\bU}_{h-1},\tilde{\bY}_h\}, \tilde{\bP}_{i,h}\subseteq \tilde{\bP}_{i,h-1}\cup\{\tilde{\bU}_{i,h-1},\tilde{\bY}_{i,h}\}$.
    
    Firstly, we show that the common information of $\tilde{\cD}(g_{1:H}^m)$ has the evolution rule. For any $h\in[H]$, consider the $\tilde{\bZ}_h$. If any $\tilde{\bY}_{i,t}\in \tilde{\bZ}_h, i\in[n],t\le h$, then $\tilde{\bY}_{i,t}\in\tilde{\bC}_h$ and $\tilde{\bY}_{i,t}\notin \tilde{\bC}_{h-1}$. Then $\bar{\bY}_{i,t}\in\bar{\bC}_h$ and $\bar{\bY}_{i,t}\notin \bar{\bC}_{h-1}$ since the construction does not change any observation in the information. Therefore, $\bar{\bY}_{i,t}\in \bar{\bZ}_h\subseteq \bar{\bP}_{h-1}\cup\{\bar{\bU}_{h-1},\bar{\bY}_h\}$ due to the evolution rule of problem $\cD(g_{1:H}^m)$; 
    since $\bar{\bY}_{i,t}$ is an observation label, it belongs to $\bar{\bP}_{h-1}\cup\{\bar{\bY}_h\}$, and then $\tilde{\bY}_{i,t}\in \tilde{\bP}_{h-1}\cup \{\tilde{\bY}_h\}$. 
    If any $\tilde{\bU}_{i,t}\in \tilde{\bZ}_h$, then $\bar{B}_{i,t}\neq \mathbf{0}$ from Assumption \ref{ass: useless action} and construction of $\tilde{\cD}(g_{1:H}^m)$. If $t<h-1$, we know $\bar{\bU}_{i,t}$ influences $\bar{\bX}_{t+1}$ and thus influences $\bar{\bY}_{j,t+1}$ for some $j\neq i$ due to Assumption \ref{ass: non-degeneracy}. Partial nestedness gives $\bar{\bI}_{i,t}\subseteq\bar{\bI}_{j,t+1}$. Due to $j\neq i$ and Assumption \ref{ass: evolution rule}, it holds that $\bar{\bI}_{i,t}\subseteq\bar{\bC}_{t+1}$, and the strict expansion therefore places $\tilde{\bU}_{i,t}$ in $\tilde{\bC}_{t+1}$. From $t<h-1$ we know $\tilde{\bU}_{i,t}\notin \tilde{\bZ}_h$, which leads to a contradiction. Hence, if $\tilde{\bU}_{i,t}\in \tilde{\bZ}_h$, it is only possible that $t=h-1$. In conclusion, $\tilde{\bZ}_h\subseteq \tilde{\bP}_{h-1}\cup\{\tilde{\bU}_{h-1},\tilde{\bY}_h\}$.
     
    Secondly, we show that the private information of $\tilde{\cD}(g_{1:H}^m)$ satisfies the evolution rule. For any $i\in[n], h\in[H]$, consider the $\tilde{\bP}_{i,h}$. If any $\tilde{\bY}_{i,t}\in \tilde{\bP}_{i,h},t\le h$, then $\bar{\bY}_{i,t}\in\bar{\bP}_{i,h}$  since the construction does not change any observation in the information. Therefore, $\bar{\bY}_{i,t}\in \bar{\bP}_{i,h-1}\cup \{\bar{\bY}_{i,h}\}$ due to the evolution rule of problem $\cD(g_{1:H}^m)$, and then $\tilde{\bY}_{i,t}\in \tilde{\bP}_{i,h-1}\cup \{\tilde{\bY}_{i,h}\}$. If any $\tilde{\bU}_{i,t}\in \tilde{\bP}_{i,h}$, then $\bar{B}_{i,t}\neq \mathbf{0}$ by Assumption~\ref{ass: useless action}. Then, we know $\bar{\bU}_{i,t}$ influences $\bar{\bX}_{t+1}$ and thus influences $\bar{\bY}_{j,t+1}$ for some $j\neq i$ due to Assumption \ref{ass: non-degeneracy}. Partial nestedness gives $\bar{\bI}_{i,t}\subseteq\bar{\bI}_{j,t+1}$. Due to $j\neq i$ and Assumption \ref{ass: evolution rule}, it holds that $\bar{\bI}_{i,t}\subseteq\bar{\bC}_{t+1}$, and the strict expansion therefore places $\tilde{\bU}_{i,t}$ in $\tilde{\bC}_{t+1}$. Thus, we have  $\tilde{\bU}_{i,t}\notin\tilde{\bP}_{i,h}$, which leads to a contradiction. In conclusion, $\tilde{\bP}_{i,h}\subseteq \tilde{\bP}_{i,h-1}\cup\{\tilde{\bY}_{i,h}\}$.
    
  \vspace{6pt}
\noindent\textbf{Part 3}: $\tilde{\cD}(g_{1:H}^m)$ satisfies the SI-CIB condition. We show this by induction. If $h=1$, then the belief $\tilde{\bB}_1$ does not depend on any strategy, and the SI-CIB condition holds automatically.

For each $h\ge 2$, the space
$\tilde{\cS}_h=\tilde{\cX}\times\tilde{\cP}_{1,h}\times\cdots\times\tilde{\cP}_{n,h}$
is finite-dimensional because the collection of random variables comprising each
$\tilde{\bP}_{i,h}$, $i\in[n]$, is fixed. For any fixed Borel set
$\cQ_h\subseteq\tilde{\cS}_h$ and any fixed common-information realization
$\tilde C_h\in\tilde{\cC}_h$ that can be reached by two control strategies
$\tilde g_{1:h-1}$ and $\tilde g'_{1:h-1}$, define
\[
\begin{aligned}
\PP_1=\PP\big(
\tilde{\bS}_h\in\cQ_h
\given
\tilde{\bC}_h=\tilde C_h,\tilde g_{1:h-1}
\big),~~
\PP_2=\PP\big(
\tilde{\bS}_h\in\cQ_h
\given
\tilde{\bC}_h=\tilde C_h,\tilde g'_{1:h-1}
\big).
\end{aligned}
\]
It suffices to prove that $\PP_1=\PP_2$. 

We first suppose that $\tilde g_{1:h-1}$ and $\tilde g'_{1:h-1}$ differ only at one pair $(i_1,h_1)$, where $i_1\in[n]$ and $h_1<h$.

If $\tilde B_{i_1,h_1}=\bar B_{i_1,h_1}\neq\mathbf{0}$, then
$\bar{\bU}_{i_1,h_1}$ influences $\bar{\bX}_{h_1+1}$ in
$\cD(g_{1:H}^m)$. By Assumption~\ref{ass: non-degeneracy}, there exists
$i_2\neq i_1$ such that $\bar{\bU}_{i_1,h_1}$ influences
$\bar{\bY}_{i_2,h_1+1}$. Since $\cD(g_{1:H}^m)$ has a PN IS, we have
$\bar{\bI}_{i_1,h_1}
\subseteq
\bar{\bI}_{i_2,h_1+1}$.
Due to $i_2\neq i_1$ and Assumption \ref{ass: evolution rule}, we have $
\bar{\bI}_{i_1,h_1}
\subseteq
\bar{\bC}_{h_1+1}$. 
By the construction of $\tilde{\cD}(g_{1:H}^m)$,
\[
\tilde{\bU}_{i_1,h_1}
\in
\tilde{\bC}_{h_1+1}
\subseteq
\tilde{\bC}_h.
\]
Moreover,
\[
\tilde{\bI}_{i_1,h_1}\setminus\bar{\bI}_{i_1,h_1}
\subseteq
\tilde{\bC}_{h_1},
\qquad
\bar{\bI}_{i_1,h_1}
\subseteq
\bar{\bC}_{h_1+1}
\subseteq
\tilde{\bC}_{h_1+1}
\subseteq
\tilde{\bC}_h,
\]
and hence
$\tilde{\bI}_{i_1,h_1}\subseteq\tilde{\bC}_h$.
Conditioning on $\tilde{\bC}_h=\tilde C_h$ therefore fixes both the information at which
$\tilde g_{i_1,h_1}$ and $\tilde g'_{i_1,h_1}$ are evaluated and the resulting action.
Since both strategies reach the same common-information realization, they produce the same recorded action from the same recorded information. With these quantities fixed, the system dynamics and information evolution are identical under the two strategies. Therefore, $\PP_1=\PP_2$.

If $\tilde B_{i_1,h_1}=\bar B_{i_1,h_1}=\mathbf{0}$, then
$\tilde{\bU}_{i_1,h_1}$ does not affect the system dynamics. Moreover,
Assumption~\ref{ass: useless action} prevents
$\bar{\bU}_{i_1,h_1}$ from appearing in any later information set. Since the strict expansion only adds actions with nonzero input coefficients,
$\tilde{\bU}_{i_1,h_1}$ does not appear in
$\tilde{\bC}_{h'}$ or $\tilde{\bP}_{i,h'}$ for any $h'\in[H]$ with
$h'>h_1$ and any $i\in[n]$. Thus, changing
$\tilde g_{i_1,h_1}$ affects neither the state nor any later common or private information. It follows again that $\PP_1=\PP_2$.

Finally, consider any two arbitrary $\tilde g_{1:h-1}$ and $\tilde g'_{1:h-1}$ that can both reach a realization $\tilde{C}_h$, we connect $\tilde{g}_{1:h-1}'$ and $\tilde{g}_{1:h-1}$ by finitely many intermediate strategies, i.e., $\tilde{g}_{1:h-1}^{0}:=\tilde{g}_{1:h-1},\tilde{g}_{1:h-1}^1,\cdots,\tilde{g}_{1:h-1}^K:=\tilde{g}_{1:h-1}'$, replacing the control strategy of one agent at a time. We will show that all of the intermediate strategies can reach $\tilde{C}_h$ by induction. The claim holds for $k=0$ trivially. Now suppose $\tilde{g}_{1:h-1}^{k-1}$ can reach $\tilde{C}_h$ and $\tilde{g}_{1:h-1}^k$ is constructed by replacing $\tilde{g}_{i_1,h_1}$ by $\tilde{g}_{i_1,h_1}'$ for some $i_1\in[n],h_1<h$. If $\tilde{B}_{i_1,h_1}\neq \mathbf 0$, then from the results above, we know that $\tilde{\bI}_{i_1,h_1}\cup\{\tilde{\bU}_{i_1,h_1}\}\subseteq \tilde{\bC}_h$. Let $\tilde{I}_{i_1,h_1}$ and $\tilde{U}_{i_1,h_1}$ be the realization in $\tilde{C}_h$, then it holds that $\tilde{g}_{i_1,h_1}(\tilde{I}_{i_1,h_1})=\tilde{g}_{i_1,h_1}'(\tilde{I}_{i_1,h_1})=\tilde{U}_{i_1,h_1}$ since both $\tilde{g}_{1:h-1}$ and $\tilde{g}_{1:h-1}'$ can reach $\tilde{C}_h$, and thus $\tilde{g}_{1:h-1}^{k}$ can also reach $\tilde{C}_h$; Otherwise, $\tilde{\bU}_{i_1,h_1}$ can neither influence the state nor the information of other agents from the results above, so $\tilde{g}_{1:h-1}^{k}$ can also reach $\tilde{C}_h$. Therefore, it suffices to only consider the case that the two strategies only differ in the strategy of one agent, which completes the proof.
\end{proof}

\subsubsection{Proof of Lemma \ref{lemma: SI-conditional-mean-covariance}}
{The proof follows the arguments in Appendices B and C of \cite{CIBLQgames}, with the minor modification needed for possibly singular Gaussian distributions.
\begin{proof}
We first prove Gaussianity: 

At timestep $h=1$, the state, private information, and common information are linear functions of the primitive Gaussian variables. Hence, $(\tilde{\bS}_1,\tilde{\bC}_1)$ is jointly Gaussian, possibly singular, and the conditional belief $\tilde{\bB}_1$ is Gaussian.

Fix $h\geq2$ and suppose that $\tilde{g}_{1:h-1}$ is affine.  Due to the linearity of the system dynamics, the state $\tilde{\bX}_h$, private information $\tilde{\bP}_h$, and common information $\tilde{\bC}_h$ are jointly Gaussian. Consequently, for every common information realization 
$\tilde{C}_h$ that can be reached under $\tilde{g}_{1:h-1}$, the conditional distribution of $\tilde{\bS}_h$ given $\tilde{\bC}_h=\tilde{C}_h$ is Gaussian. This conclusion remains valid when the joint distribution is singular: the conditional-Gaussian formula, with the Moore-Penrose pseudo-inverse in place of the inverse, gives a Gaussian conditional distribution whose covariance does not depend on  $\tilde{C}_h$.

If $\tilde{g}_{1:h-1}$ is not affine,  then for every common information realization 
$\tilde{C}_h$ that can be reached under $\tilde{g}_{1:h-1}$, there exist realizations $\tilde{Y}_{1:h},\tilde{U}_{1:h-1}$ such that $(\tilde{Y}_{1:h},\tilde{U}_{1:h-1}, \tilde{C}_h)$ is reachable under $\tilde{g}_{1:h-1}$. Then, we construct affine strategies $\tilde{g}_{1:h-1}'\in\tilde{\cG}_{1:h-1}$ as: $\forall i\in[n],t<h, \forall \tilde{I}_{i,t}\in\tilde{\cI}_{i,t}, \tilde{g}_{i,t}'(\tilde{I}_{i,t})\equiv\tilde{U}_{i,t}$. Then, $(\tilde{Y}_{1:h},\tilde{U}_{1:h-1}, \tilde{C}_h)$ is also reachable under $\tilde{g}_{1:h-1}'$. Since the CIB belief $\tilde{\bB}_h$ is strategy independent, its value at $\tilde C_h$ is the same under the original and comparison strategies. Thus,  $\tilde{\bB}_h$ is Gaussian under any admissible strategy $\tilde{g}_{1:h-1}$.

We next prove that $\tilde{\bTheta}_h$ and $\tilde{\Sigma}_h$ satisfy the evolution rules. The system dynamics and Equation  \eqref{equ: DecLQG_info_evolution} give the usual fixed CIB-belief update from the preceding belief, the past control laws, and $\tilde{\bZ}_h$. The SI-CIB condition removes the dependence on the control laws, so
\[
 (\tilde{\bTheta}_h,\tilde{\Sigma}_h)
 =
 \tilde F_{h-1}\bigl((\tilde{\bTheta}_{h-1},\tilde{\Sigma}_{h-1}),\tilde{\bZ}_h\bigr),
\]
where $\tilde F_{h-1}$ is a fixed transformation that does not depend on the control strategies.

To determine the covariance component, again take affine control laws. The variables $\tilde{\bS}_h$ and $\tilde{\bC}_h$ are jointly Gaussian, and the conditional covariance of jointly Gaussian variables does not depend on the realized conditioning value, including in the singular case. Similarly to the argument above, for any strategy, we can replace it by a constant affine strategy. Constant control values change only the conditional mean, not the conditional covariance. Hence,
\[
 \tilde{\Sigma}_h=\tilde{\Psi}_h^2(\tilde{\Sigma}_{h-1})
\]
for a fixed deterministic function $\tilde{\Psi}_h^2$.

Under the same affine control laws, $\EE[\tilde{\bS}_h\given\tilde{\bC}_h]$ is an affine function of $\tilde{\bC}_h$, and $\EE[\tilde{\bS}_{h-1}\given\tilde{\bC}_{h-1}]$ is an affine function of $\tilde{\bC}_{h-1}$. Combining these facts with the preceding fixed belief update gives
\[
 \tilde{\bTheta}_h=\tilde{\Psi}_h^1(\tilde{\bTheta}_{h-1},\tilde{\bZ}_h)
\]
for a fixed affine function $\tilde{\Psi}_h^1$. Since the dynamics and observation equations are homogeneous and the relevant known controls are included in $\tilde{\bZ}_h$, the constant term is zero, so $\tilde{\Psi}_h^1$ is linear. Neither update depends on the control strategies, which completes the proof.
\end{proof}}

\subsubsection{Proof of Theorem \ref{thm: virtual MDP and equivalence.}}
\begin{proof}
    \noindent\textbf{Part 1:} For $h\in[H-1]$, given $(\tilde\Theta_h,\tilde\gamma_h)$, draw $\tilde{\bS}_h\sim\cN(\tilde\Theta_h,\tilde\Sigma_h)$ and the independent Gaussian noises and then apply the dynamics, $\tilde\chi_{h+1}$, and Lemma~\ref{lemma: SI-conditional-mean-covariance}. The law of $\tilde\bTheta_{h+1}$ therefore depends only on $(\tilde\Theta_h,\tilde\gamma_h)$. At $h=H$, the final-state cost has already been included in $c_H^\ddag$. Hence,  $\tilde\cD^\ddag(g_{1:H}^m)$ is Markov.
        
        \vspace{6pt}
    \noindent\textbf{Part 2:} We aim to prove that for any strategy $\tilde{g}_{1:H}\in\tilde{\cG}_{1:H}, J_{\tilde{\cD}(g_{1:H}^m)}(\tilde{g}_{1:H})=J_{\tilde{\cD}^\ddag(g_{1:H}^m)}(g^\ddag_{1:H})$, where $g_h^\ddag=\varsigma_h(\tilde{g}_h), \forall h\in[H]$. For every $h\in[H-1]$,
        \begin{align}
    &\EE[c_h^\ddag\given g_{1:H}^\ddag]=\EE\Big[\EE[c_h^\ddag(\tilde{\bTheta}_h,g^\ddag_h(\tilde{\bC}_h))\given \tilde{\bC}_h, \tilde{\bTheta}_h=\tilde{\Psi}_h^3(\tilde{\bC}_h)]\Biggiven g_{1:h}^\ddag\Big]\label{line: virtual problem l1}\\
    &\quad=\EE\left[\EE\left[\tilde{\bX}_h^\top\tilde{Q}_h^1\tilde{\bX}_h+\begin{bmatrix}
        \tilde{\gamma}_{1,h}(\tilde{\bP}_{1,h})^\top&
        \cdots&
        \tilde{\gamma}_{n,h}(\tilde{\bP}_{n,h})^\top
    \end{bmatrix}\tilde{Q}_h^2\begin{bmatrix}
        \tilde{\gamma}_{1,h}(\tilde{\bP}_{1,h})\\
        \cdots\\
        \tilde{\gamma}_{n,h}(\tilde{\bP}_{n,h})
    \end{bmatrix}\Bigggiven \tilde{\bC}_h,\tilde{\gamma}_h=g_h^\ddag(\tilde{\bC}_h)\right]\Bigggiven g_{1:h}^\ddag\right]\label{line: virtual problem l2}\\
    &\quad=\EE\left[\EE\left[\tilde{\bX}_h^\top\tilde{Q}_h^1\tilde{\bX}_h+\begin{bmatrix}
        \tilde{\gamma}_{1,h}(\tilde{\bP}_{1,h})^\top&
        \cdots&
        \tilde{\gamma}_{n,h}(\tilde{\bP}_{n,h})^\top
    \end{bmatrix}\tilde{Q}_h^2\begin{bmatrix}
        \tilde{\gamma}_{1,h}(\tilde{\bP}_{1,h})\\
        \cdots\\
        \tilde{\gamma}_{n,h}(\tilde{\bP}_{n,h})
    \end{bmatrix}\Bigggiven \tilde{\bC}_h,\tilde{\gamma}_h(\cdot)=\tilde{g}_h(\tilde{\bC}_h,\cdot)\right]\Bigggiven g_{1:h}^\ddag\right]\label{line: virtual problem l3}\\
    &\quad=\EE\left[\EE\left[\tilde{\bX}_h^\top\tilde{Q}_h^1\tilde{\bX}_h+\begin{bmatrix}
        \tilde{\gamma}_{1,h}(\tilde{\bP}_{1,h})^\top&
        \cdots&
        \tilde{\gamma}_{n,h}(\tilde{\bP}_{n,h})^\top
    \end{bmatrix}\tilde{Q}_h^2\begin{bmatrix}
        \tilde{\gamma}_{1,h}(\tilde{\bP}_{1,h})\\
        \cdots\\
        \tilde{\gamma}_{n,h}(\tilde{\bP}_{n,h})
    \end{bmatrix}\Bigggiven  \tilde{\bC}_h,\tilde{\gamma}_h(\cdot)=\tilde{g}_h(\tilde{\bC}_h,\cdot)\right]\Bigggiven \tilde{g}_{1:h}\right]\label{line: virtual problem l4}\\
    &\quad=\EE[\tilde{\bX}_h^\top\tilde{Q}_h^1\tilde{\bX}_h+\tilde{\bU}_h^\top\tilde{Q}_h^2\tilde{\bU}_h\given \tilde{g}_{1:H}].\label{line: virtual problem l5}
\end{align}

Equations~\eqref{line: virtual problem l1} and~\eqref{line: virtual problem l5} are due to the tower property; Equation~\eqref{line: virtual problem l2} is from the definition of $c_h^\ddag$; Equation~\eqref{line: virtual problem l3} is from the definition of $\varsigma_h$; Equation~\eqref{line: virtual problem l4} is because $\tilde{\bC}_h$ has the same distribution under both strategies $\tilde{g}_{1:h}$ and $g_{1:h}^\ddag$, i.e., for any Borel set $\tilde{\sC}_h\subseteq \tilde{\cC}_h$, it holds that
\begin{align*}
    &\PP(\tilde{\bC}_h\in \tilde{\sC}_h\given g_{1:h}^\ddag)=\EE[\mathds{1}[\tilde{\bC}_h\in \tilde{\sC}_h]\given g_{1:h}^\ddag]=\EE\left[\EE\Big[\mathds{1}[\tilde{\bC}_h\in \tilde{\sC}_h]\given \tilde{\bC}_{h-1},\tilde{\gamma}_{h-1}=g_{h-1}^\ddag(\tilde{\bC}_{h-1})\Big]\given g_{1:h-1}^\ddag\right]\\
    &=\EE\left[\EE\Big[\cdots\EE\Big[\EE\Big[\mathds{1}[\tilde{\bC}_h\in \tilde{\sC}_h]\Biggiven \tilde{\bC}_{h-1},\tilde{\gamma}_{h-1}=g_{h-1}^\ddag(\tilde{\bC}_{h-1})\Big]\Biggiven \tilde{\bC}_{h-2},\tilde{\gamma}_{h-2}=g_{h-2}^\ddag(\tilde{\bC}_{h-2})\Big]\cdots\Biggiven \tilde{\bC}_1,\tilde{\gamma}_1=g_1^\ddag(\tilde{\bC}_1)\Big]\right]\\
    &=\EE\left[\EE\Big[\cdots\EE\Big[\EE\Big[\mathds{1}[\tilde{\bC}_h\in \tilde{\sC}_h]\given \tilde{\bC}_{h-1},\tilde{\gamma}_{h-1}(\cdot)=\tilde{g}_{h-1}(\tilde{\bC}_{h-1},\cdot)\Big]\given \tilde{\bC}_{h-2},\tilde{\gamma}_{h-2}(\cdot)=\tilde{g}_{h-2}(\tilde{\bC}_{h-2},\cdot)\Big]\cdots\Biggiven  \tilde{\bC}_1,\tilde{\gamma}_1(\cdot)=\tilde{g}_1(\tilde{\bC}_1,\cdot)\Big]\right]\\
    &=\EE\left[\mathds{1}[\tilde{\bC}_h\in \tilde{\sC}_h]\Biggiven \tilde{g}_{1:h}\right]=\PP(\tilde{\bC}_h\in \tilde{\sC}_h\given \tilde{g}_{1:h}).
\end{align*}
The same calculation at $h=H$, now using the definition of $c_H^\ddag$, gives
\[
\EE[c_H^\ddag\given g_{1:H}^\ddag]
=\EE[\tilde{\bX}_H^\top\tilde Q_H^1\tilde{\bX}_H+\tilde{\bU}_H^\top\tilde Q_H^2\tilde{\bU}_H+\tilde{\bX}_{H+1}^\top\tilde Q_{H+1}^1\tilde{\bX}_{H+1}\given\tilde g_{1:H}].
\]
Therefore,
\[
J_{\tilde{\cD}^\ddag(g_{1:H}^m)}(g_{1:H}^\ddag)
=\EE\left[\sum_{h=1}^H c_h^\ddag\given g_{1:H}^\ddag\right]
=J_{\tilde{\cD}(g_{1:H}^m)}(\tilde g_{1:H}),
\]
which completes the proof.
\end{proof}

\subsubsection{Proof of Lemma \ref{thm: linear in M}}
\begin{proof}
        For convenience, given any matrices $\{\bar{\cE}_{i,h}\}_{i\in[n],h\in[H]}$ and $\{\tilde{\cF}_{i,h}\}_{i\in[n],h\in[H]}$ with appropriate dimensions, we denote by $\tilde{g}_{1:H}=\Lambda(\{\bar{\cE}_{i,h}\}_{i\in[n],h\in[H]},\{\tilde{\cF}_{i,h}\}_{i\in[n],h\in[H]})$ 
        the linear strategy generated by $\{\bar{\cE}_{i,h}\}_{i\in[n],h\in[H]},\{\tilde{\cF}_{i,h}\}_{i\in[n],h\in[H]}$, namely for any $i\in[n],h\in[H], \tilde{g}_{i,h}(\tilde{\bC}_h,\tilde{\bP}_{i,h})=\bar{\cE}_{i,h}\tilde{\bC}_h+\tilde{\cF}_{i,h}\tilde{\bP}_{i,h}$.
        {Theorem \ref{thm: expansion} and Corollary~\ref{cor:degenerate-PN-linearity} give an optimal linear strategy $\tilde g_{1:H}^\ast=\Lambda(\{\bar{\cE}_{i,h}^\ast\},\{\tilde{\cF}_{i,h}^\ast\})$.}
        If we fix $\{\tilde{\cF}_{i,h}^\ast\}_{i\in[n],h\in[H]}$, then $\{\bar{\cE}_{i,h}^\ast\}_{i\in[n],h\in[H]}$ minimizes $J_{\tilde{\cD}(g_{1:H}^m)}$ over $\{\bar{\cE}_{i,h}\}_{i\in[n],h\in[H]}$ with this private-information component fixed, i.e.,
        \begin{align*}
            \{\bar{\cE}_{i,h}^\ast\}_{i\in[n],h\in[H]}\in \argmin_{\{\bar{\cE}_{i,h}\}_{i\in[n],h\in[H]}}J_{\tilde{\cD}(g_{1:H}^m)}\left(\Lambda(\{\bar{\cE}_{i,h}\}_{i\in[n],h\in[H]},\{\tilde{\cF}_{i,h}^\ast\}_{i\in[n],h\in[H]})\right).
        \end{align*}
        From Theorem \ref{thm: expansion}, we know that for any $h\in[H]$
        \begin{align*}
        \tilde{\bZ}_h=\tilde{\chi}_h(\tilde{\bP}_{h-1},\tilde{\bU}_{h-1}, \tilde{\bY}_h), \qquad \tilde{\bP}_h=\tilde{\zeta}_h(\tilde{\bP}_{h-1},\tilde{\bU}_{h-1}, \tilde{\bY}_h),
        \end{align*}
        where $\tilde{\chi}_h,\tilde{\zeta}_h$ are fixed projection functions. Therefore, there exist matrices such that
        \begin{align*}
        \tilde{\bZ}_h=\tilde{\cT}_h^1\tilde{\bP}_{h-1}+\tilde{\cT}_h^2\tilde{\bU}_{h-1}+\tilde{\cT}_h^3\tilde{\bY}_h, \qquad \tilde{\bP}_h=\tilde{\cL}_h^1\tilde{\bP}_{h-1}+\tilde{\cL}_h^2\tilde{\bU}_{h-1}+\tilde{\cL}_h^3\tilde{\bY}_h,
        \end{align*}
        for some matrices $\tilde{\cT}_h^1,\tilde{\cT}_h^2,\tilde{\cT}_h^3,\tilde{\cL}_h^1,\tilde{\cL}_h^2,\tilde{\cL}_h^3$.

        As shown in \cite{AA}, if we fix $\{\tilde{\cF}_{i,h}^\ast\}_{i\in[n],h\in[H]}$, we can construct an $H$-step centralized linear control problem $\Breve{\cD}$ based on $\tilde{\cD}(g_{1:H}^m)$, where $\Breve{\cD}$ is an LQG system with linear dynamics, linear
        observations, quadratic costs, and Gaussian noises.   For completeness, we include  the construction below.

        Specifically, we can construct the centralized LQG problem $\Breve{\cD}$  
        based on $\tilde{\cD}(g_{1:H}^m)$ and fixed $\{\tilde{\cF}_{i,h}^\ast\}_{i\in[n],h\in[H]}$, where the state $\Breve{\bX}_h\in \Breve{\cX}_h$, the action $\Breve{\bU}_h\in \Breve{\cU}_h$, and the observation $\Breve{\bY}_h\in \Breve{\cY}_h$ evolve  as follows:
        \begin{align*}
            &\forall h\in[H],~ \Breve{\cX}_h=\tilde{\cX}\times \tilde{\cP}_h, ~~~\Breve{\cU}_h=\tilde{\cU}_h, ~~~\Breve{\cY}_h=\tilde{\cZ}_h, ~~~\Breve{\bX}_{h+1}=\Breve{A}_h\Breve{\bX}_h+\Breve{B}_h\Breve{\bU}_h+\Breve{\bW}_{h},\\
            &\qquad\qquad\qquad\Breve{\bY}_h=\Breve{E}_h\Breve{\bX}_{h-1}+\Breve{F}_h\Breve{\bU}_{h-1}+\Breve{\bV}_{h}~~ (\text{for }h>1),~~ \Breve{\bY}_1=\Breve{\bV}_1,\\
            &\Breve{\bX}_{H+1}=\tilde{\bX}_{H+1},~~~
             \Breve{\bX}_h=\begin{bmatrix}
                 \tilde{\bX}_h\\
                 \tilde{\bP}_h
             \end{bmatrix}, ~~~\Breve{\bU}_h=\tilde{\bU}_h-\diag(\tilde{\cF}_{1,h}^\ast,\cdots,\tilde{\cF}_{n,h}^\ast)\tilde{\bP}_h, ~~~\Breve{\bY}_h=\tilde{\bZ}_h.
        \end{align*}
        Defining $\tilde{\cF}_h^\ast=\diag(\tilde{\cF}_{1,h}^\ast,\cdots,\tilde{\cF}_{n,h}^\ast)$, the matrices $\Breve{A}_h,\Breve{B}_h,\Breve{E}_h, \Breve{F}_h$ and noises $\Breve{\bW}_h,\Breve{\bV}_h$ are given by
        \begin{align*}
            &\forall h\in[H-1], \Breve{A}_h=\begin{bmatrix}
                \tilde{A}_h&\tilde{B}_h\tilde{\cF}_h^\ast\\
                \tilde{\cL}_{h+1}^3\tilde{E}_{h+1}\tilde{A}_h&\tilde{\cL}_{h+1}^1+\tilde{\cL}_{h+1}^2\tilde{\cF}_h^\ast+\tilde{\cL}_{h+1}^3\tilde{E}_{h+1}\tilde{B}_h\tilde{\cF}_h^\ast
                \end{bmatrix}, 
                \Breve{B}_h=\begin{bmatrix}
                    \tilde{B}_h\\\tilde{\cL}_{h+1}^2+\tilde{\cL}_{h+1}^3\tilde{E}_{h+1}\tilde{B}_h
                \end{bmatrix};\\
            &\forall h\in[2:H], \Breve{E}_h=\begin{bmatrix}
                    \tilde{\cT}_h^3\tilde{E}_h\tilde{A}_{h-1}&\tilde{\cT}_{h}^1+\tilde{\cT}_h^2\tilde{\cF}_{h-1}^\ast+\tilde{\cT}_h^3\tilde{E}_h\tilde{B}_{h-1}\tilde{\cF}_{h-1}^\ast
                \end{bmatrix}, \Breve{F}_h=\begin{bmatrix}
                    \tilde{\cT}_h^2+\tilde{\cT}_h^3\tilde{E}_h\tilde{B}_{h-1}
                \end{bmatrix};\\
            &\forall h\in[H-1],\Breve{\bW}_h=\begin{bmatrix}
                    \tilde{\bW}_{0,h}\\\tilde{\cL}_{h+1}^3(\tilde{E}_{h+1}\tilde{\bW}_{0,h}+\tilde{\bW}_{1:n,h+1})
                \end{bmatrix},\Breve{\bV}_{h+1}=\tilde{\cT}_{h+1}^3(\tilde{E}_{h+1}\tilde{\bW}_{0,h}+\tilde{\bW}_{1:n,h+1});\\
                &\Breve{A}_H=\begin{bmatrix}
                \tilde{A}_H&\tilde{B}_H\tilde{\cF}_H^\ast
                \end{bmatrix}, 
                \Breve{B}_H=\tilde{B}_H, \Breve{\bW}_H=\tilde{\bW}_{0,H}, \Breve{\bV}_1=\tilde{\cT}_1^3(\tilde{E}_1\tilde{\bX}_1+\tilde{\bW}_{1:n,1}), 
        \end{align*}
        where we recall that $\forall h\in[H-1],\tilde{\bW}_{1:n,h+1}=\begin{bmatrix}
                    \tilde{\bW}_{1,h+1}\\\cdots\\\tilde{\bW}_{n,h+1}
                \end{bmatrix}$.  The cost of each stage $h\in[H]$ is $\Breve{c}_h=\begin{bmatrix}
            \Breve{\bX}_h^\top&\Breve{\bU}_h^\top
        \end{bmatrix}\begin{bmatrix}
            \Breve{Q}_h^1&\Breve{N}_h\\\Breve{N}_h^\top&\Breve{Q}_h^2
        \end{bmatrix}\begin{bmatrix}
            \Breve{\bX}_h\\\Breve{\bU}_h
        \end{bmatrix}$ and $\Breve{c}_{H+1}=\Breve{\bX}_{H+1}^\top\Breve{Q}^1_{H+1}\Breve{\bX}_{H+1}$, where the matrices $\Breve{Q}_h^1,\Breve{Q}_h^2,\Breve{N}_h, \Breve{Q}^1_{H+1}$ are defined as
        \begin{align*}
            \forall h\in[H], \Breve{Q}_h^1=\begin{bmatrix}
                \tilde{Q}_h^1&0\\0&(\tilde{\cF}_h^{\ast})^{\top}\tilde{Q}_h^2\tilde{\cF}_h^\ast
            \end{bmatrix}, \Breve{Q}_h^2=\tilde{Q}_h^2, \Breve{N}_h=\begin{bmatrix}
                0\\(\tilde{\cF}_h^{\ast})^{\top}\tilde{Q}_h^2
            \end{bmatrix}, \Breve{Q}^1_{H+1}=\tilde{Q}^1_{H+1}.
        \end{align*}
        The strategy of $\Breve{\cD}$ at each stage $h\in[H]$ is defined by $\Breve{g}_h:\prod_{t=1}^h\Breve{\cY}_t\rightarrow \Breve{\cU}_h$ and the objective is defined as $J_{\Breve{\cD}}(\Breve{g}_{1:H})=\EE[\sum_{h=1}^{H+1}\Breve{c}_h\given \Breve{g}_{1:H}]$. Thus, an optimal strategy satisfies $\Breve{g}_{1:H}^\ast\in \argmin_{\Breve{g}_{1:H}}J_{\Breve{\cD}}(\Breve{g}_{1:H})$. From the standard result of the linear-quadratic control, we know that there exists an optimal strategy $\Breve{g}_{1:H}^\ast$ that is linear with the form $\Breve{g}_h^\ast(\Breve{\bY}_{1:h})=\Breve{\cE}_h^\ast\Breve{\bY}_{1:h}=\tilde{\cE}_h^\ast\hat{\breve{\bX}}_h, \forall h\in[H]$ and $\Breve{\cE}_h^\ast=\vect(
            \Breve{\cE}_{1,h}^\ast,\cdots,\Breve{\cE}_{n,h}^\ast)$, where $\hat{\Breve{\bX}}_h:=\EE[\Breve{\bX}_h\given \Breve{\bY}_{1:h}]$. Therefore, we can construct a linear strategy of $\tilde{\cD}(g_{1:H}^m)$ as $\tilde{g}_{1:H}^\ast:=\Lambda(\{\Breve{\cE}_{i,h}^\ast\}_{i\in[n],h\in[H]},\{\tilde{\cF}_{i,h}^\ast\}_{i\in[n],h\in[H]})$ such that $J_{\Breve{\cD}}(\Breve{g}_{1:H})=J_{\tilde{\cD}(g_{1:H}^m)}(\tilde{g}_{1:H})$. Therefore, $\tilde{g}_{1:H}^\ast:=\Lambda(\{\Breve{\cE}_{i,h}^\ast\}_{i\in[n],h\in[H]},\{\tilde{\cF}_{i,h}^\ast\}_{i\in[n],h\in[H]})$ is an optimal strategy of $\tilde{\cD}(g_{1:H}^m)$.
            
            From the construction of $\Breve{\cD}$, we know that $ \Breve{\bX}_h=\vect(\tilde{\bX}_h,\tilde{\bP}_h), \Breve{\bY}_h=\tilde{\bZ}_h$ for any $h\in[H]$, so $\hat{\Breve{\bX}}_h=\EE[\Breve{\bX}_h\given \Breve{\bY}_{1:h}]=\begin{bmatrix}
                 \EE[\tilde{\bX}_h\given \tilde{\bC}_h]\\
                 \EE[\tilde{\bP}_h\given \tilde{\bC}_h]
             \end{bmatrix}=\tilde{\bTheta}_h$, and such a constructed optimal strategy $\tilde{g}_{1:H}^\ast$ is a Markovian strategy.

             Now we aim to prove that $\tilde{\gamma}_h$ can minimize the right-hand side of Equation \eqref{equ: centralized Bellman Equation} for any reachable mean $\tilde\Theta_h$.  We define the set of all possible \emph{affine} prescriptions as $\tilde{\Gamma}_h^{l}:=\{\tilde{\gamma}_h=\{\tilde{\gamma}_{i,h}\}_{i\in[n]}\given \tilde{\gamma}_{i,h}(\cdot)=a+\tilde{\cF}_{i,h}\cdot \text{ for some $a$ and $\tilde{\cF}_{i,h}$ },\forall i\in[n]\},$ and define the value function with respect to affine prescriptions as: 
             \begin{align*}
                \tilde{V}_{H+1}^l(\tilde{\bTheta}_{H+1})=\tilde{V}_{H+1}(\tilde{\bTheta}_{H+1}),\quad \tilde{V}_h^l(\tilde{\bTheta}_h)=\min_{\tilde{\gamma}_h\in \tilde{\Gamma}_h^{l}}\EE[c_h^\ddag(\tilde{\bTheta}_h,\tilde{\gamma}_h)+\tilde{V}_{h+1}^l(\tilde{\bTheta}_{h+1})\given \tilde{\bTheta}_h, \tilde{\gamma}_h], \forall h\in[H].
             \end{align*} 
            Using the same analysis as the part 2 of the proof of Theorem \ref{thm: dp_in_decLQG}, we can show that $\tilde{V}_h^l(\tilde{\bTheta}_h)$ has a quadratic form $\tilde{V}_h^l(\tilde{\bTheta}_h)=\tilde{\bTheta}_h^\top \tilde{R}_h\tilde{\bTheta}_h+\tilde{c}_h$ for some $\tilde{R}_h\succeq0$ and $\tilde{c}_h$ and the optimal prescription under $\tilde{\Theta}_h$ is selected as  $\tilde{\gamma}_h^\ast=\{\tilde{\gamma}_{i,h}^\ast\}_{i\in[n]}, \tilde{\gamma}_{i,h}^\ast(\cdot)=\tilde{\cE}_{i,h}^\ast\tilde{\Theta}_h+\tilde{\cF}_{i,h}^\ast\cdot$, where $\tilde{R}_h,\tilde{c}_h$ are computed in Equations~\eqref{eq: evolution of R_h} and~\eqref{eq: computing constant c_h} and $\tilde{\cE}_{i,h}^\ast, \tilde{\cF}_{i,h}^\ast$ are computed in Equations~\eqref{eq: centralized Riccati} and~\eqref{eq: optimize F}. Now, we aim to show that $\tilde{V}_h^l(\tilde{\bTheta}_h)=\tilde{V}_h(\tilde{\bTheta}_h)$ for any $\tilde{\bTheta}_h$ by induction. Recall that $\tilde{V}_h(\tilde{\bTheta}_h)=\min_{\tilde{\gamma}_h\in \tilde{\Gamma}_h}\EE[c_h^\ddag(\tilde{\bTheta}_h,\tilde{\gamma}_h)+\tilde{V}_{h+1}(\tilde{\bTheta}_{h+1})\given \tilde{\bTheta}_h, \tilde{\gamma}_h]$, which allows for non-affine prescriptions. 

            At timestep $h=H+1$, we have $\tilde{V}_{H+1}^l(\tilde{\bTheta}_{H+1})=\tilde{V}_{H+1}(\tilde{\bTheta}_{H+1})$ from definition. Now, fix any $h\in[H]$ and suppose that $\tilde{V}_{h+1}^l=\tilde{V}_{h+1}$. Define $\tilde{V}_h^q(\tilde{\bTheta}_h,\tilde{\gamma}_h):=\EE[c_h^\ddag(\tilde{\bTheta}_h,\tilde{\gamma}_h)+\tilde{V}_{h+1}(\tilde{\bTheta}_{h+1})\given \tilde{\bTheta}_h, \tilde{\gamma}_h]$ and the optimal prescription $\tilde{\gamma}_h'\in\argmin_{\tilde{\gamma}_h\in \tilde{\Gamma}_h}\tilde{V}_h^q(\tilde{\bTheta}_h,\tilde{\gamma}_h)$. 

            From Lemma \ref{lemma: SI-conditional-mean-covariance}, we know that $\tilde{\bTheta}_{h+1}=\tilde{\Psi}_{h+1}^1(\tilde{\bTheta}_h,\tilde{\bZ}_{h+1})$, and $\tilde{\bS}_h\sim \cN(\tilde{\bTheta}_h, \tilde{\Sigma}_h)$, we can write $[\tilde{\bX}_h^\top,\tilde{\bP}_h^\top]^\top=\tilde{\bTheta}_h+\tilde{\omega}_h$, where $\tilde{\omega}_h\sim \cN(0, \tilde{\Sigma}_h)$. Also, from Theorem \ref{thm: expansion}, we know that $\tilde{\bZ}_h=\tilde{\chi}_h(\tilde{\bP}_{h-1},\tilde{\bU}_{h-1}, \tilde{\bY}_h)$. Together with the fact that $\tilde{\bY}_{h+1}=\tilde{E}_{h+1}(\tilde{A}_h\tilde{\bX}_h+\tilde{B}_h\tilde{\bU}_h+\tilde{\bW}_{0,h})+\tilde{\bW}_{1:n,h+1}$, we can write $\tilde{\bTheta}_{h+1}=\tilde{L}_h^2\tilde{\bTheta}_h+\tilde{L}_h^1\tilde{\bU}_h+\tilde{L}_h^4\tilde{\omega}_h+\epsilon_h$, where $\tilde{L}_h^1,\tilde{L}_h^2,\tilde{L}_h^4$ are defined in Equations~\eqref{eq: centralized Riccati} and~\eqref{eq: optimize F}, and $\epsilon_h$ is a zero-mean Gaussian noise independent of $\tilde{\bTheta}_h,\tilde{\bU}_h,\tilde{\omega}_h$. Therefore, we can write 
            \begin{align*}
                \tilde{V}_h^q(\tilde{\bTheta}_h,\tilde{\gamma}_h)=d_h(\tilde{\bTheta}_h)+\EE[\tilde{\bU}_h^\top G_h\tilde{\bU}_h+2\tilde{\bU}_h^\top(M_h\tilde{\bTheta}_h+N_h\tilde{\omega}_h)\given 
                \tilde{\bTheta}_h,\tilde{\gamma}_h],\\
                G_h:=\tilde{Q}_h^2+(\tilde{L}_h^1)^\top\tilde{R}_{h+1}\tilde{L}_h^1,\quad  M_h:=(\tilde{L}_h^1)^\top\tilde{R}_{h+1}\tilde{L}_h^2, \quad N_h:=(\tilde{L}_h^1)^\top\tilde{R}_{h+1}\tilde{L}_h^4,
            \end{align*}
            where $d_h(\tilde{\bTheta}_h)$ is a function of $\tilde{\bTheta}_h$ and does not depend on $\tilde{\gamma}_h$. 

            Now, we are ready to prove that, for any $\tilde{\Theta}_h$, any optimal prescription $\tilde{\gamma}_h'$ cannot be better than the optimal affine prescription $\tilde{\gamma}_h^\ast$, i.e., $\tilde{V}_h^q(\tilde{\Theta}_h,\tilde{\gamma}_h')\ge \tilde{V}_h^q(\tilde{\Theta}_h,\tilde{\gamma}_h^\ast)$. Let $\tilde{\bU}_h'=\tilde{\gamma}_h'(\tilde{\bP}_h)$ and $\tilde{\bU}_h^\ast=\tilde{\gamma}_h^\ast(\tilde{\bP}_h)$, and let $\Delta\tilde{\bU}_h=\tilde{\bU}_h'-\tilde{\bU}_h^\ast, \Delta\tilde{\bU}_{i,h}=\tilde{\bU}_{i,h}'-\tilde{\bU}_{i,h}^\ast$, then we have 
            \begin{align*}
                &\tilde{V}_h^q(\tilde{\Theta}_h,\tilde{\gamma}_h')-\tilde{V}_h^q(\tilde{\Theta}_h,\tilde{\gamma}_h^\ast)\\
                &\qquad=\EE[(\tilde{\bU}_h')^\top G_h\tilde{\bU}_h'+2(\tilde{\bU}_h')^\top(M_h\tilde{\Theta}_h+N_h\tilde{\omega}_h)-(\tilde{\bU}_h^\ast)^\top G_h\tilde{\bU}_h^\ast-2(\tilde{\bU}_h^\ast)^\top(M_h\tilde{\Theta}_h+N_h\tilde{\omega}_h)\given \tilde{\Theta}_h]\\
                &\qquad=2\EE[\Delta\tilde{\bU}_h^\top (G_h\tilde{\bU}_h^\ast+M_h\tilde{\Theta}_h+N_h\tilde{\omega}_h)\given \tilde{\Theta}_h]+\EE[\Delta\tilde{\bU}_h^\top G_h\Delta\tilde{\bU}_h\given \tilde{\Theta}_h].
            \end{align*}
            From Equation~\eqref{eq: centralized Riccati}, one can verify that $G_h\hat{\cE}_h^\ast+M_h=0$.
            Let $r_h=[r_{1,h},\cdots,r_{n,h}]=G_h\tilde{\bU}_h^\ast+M_h\tilde{\Theta}_h+N_h\tilde{\omega}_h$, where $r_{i,h}$ is the $i$-th block of $r_h$ with the same dimension as $\tilde{\bU}_{i,h}$. Then we have $\EE[G_h\tilde{\bU}_h^\ast+M_h\tilde{\Theta}_h+N_h\tilde{\omega}_h\given \tilde{\Theta}_h]=\EE[(G_h\hat{\cE}_h^\ast+M_h)\tilde{\Theta}_h\given \tilde{\Theta}_h]=0$. Also, since $\{\tilde{\cF}_{i,h}^\ast\}_{i\in[n]}$ is optimal, if we construct a new prescription $\tilde{\gamma}_h^{\ast,new}$ by replacing $\tilde{\gamma}_{i,h}^\ast$ as $\tilde{\gamma}_{i,h}^{\ast,new}(\tilde{\bP}_{i,h})=\tilde{\gamma}_{i,h}^{\ast}(\tilde{\bP}_{i,h})+tK(\tilde{\bP}_{i,h}-\II_{p,i,h}\tilde{\Theta}_h)$ for some $i\in[n]$, arbitrary real number $t\in \RR$ and arbitrary matrix $K$, it must hold $0\le \tilde{V}_h^q(\tilde{\Theta}_h,\tilde{\gamma}_h^{\ast,new})-\tilde{V}_h^q(\tilde{\Theta}_h,\tilde{\gamma}_h^\ast)=2t\EE[(K(\tilde{\bP}_{i,h}-\II_{p,i,h}\tilde{\Theta}_h))^\top r_{i,h}\given \tilde{\Theta}_h]+t^2\EE[(K(\tilde{\bP}_{i,h}-\II_{p,i,h}\tilde{\Theta}_h))^\top G_h^{ii}(K(\tilde{\bP}_{i,h}-\II_{p,i,h}\tilde{\Theta}_h))\given \tilde{\Theta}_h]$, where $G_h^{ii}$ is the $i$-th diagonal block of $G_h$. This can hold for any matrix $K$ and real number $t$ only if $\EE[r_{i,h}(\tilde{\bP}_{i,h}-\II_{p,i,h}\tilde{\Theta}_h)^\top\given \tilde{\Theta}_h]=\mathbf{0}$. Since $r_{i,h}$ and $\tilde{\bP}_{i,h}-\II_{p,i,h}\tilde{\Theta}_h$ are jointly Gaussian, we have $\EE[r_{i,h}\given \tilde{\Theta}_h,\tilde{\bP}_{i,h}]=\EE[r_{i,h}\given \tilde{\Theta}_h,(\tilde{\bP}_{i,h}-\II_{p,i,h}\tilde{\Theta}_h)]=0$. Since $\Delta\tilde{\bU}_{i,h}$ can be determined by $\tilde{\Theta}_h$ and $\tilde{\bP}_{i,h}$, we have 
            \begin{align*}
                &\EE[\Delta\tilde{\bU}_h^\top (G_h\tilde{\bU}_h^\ast+M_h\tilde{\Theta}_h+N_h\tilde{\omega}_h)\given \tilde{\Theta}_h]=\sum_{i=1}^n\EE[\Delta\tilde{\bU}_{i,h}^\top r_{i,h}\given \tilde{\Theta}_h]\\
                &\qquad=\sum_{i=1}^n\EE\Big[\EE[\Delta\tilde{\bU}_{i,h}^\top r_{i,h}\given \tilde{\Theta}_h,\tilde{\bP}_{i,h}]\biggiven \tilde{\Theta}_h\Big]=\sum_{i=1}^n\EE\Big[\Delta\tilde{\bU}_{i,h}^\top\EE[r_{i,h}\given \tilde{\Theta}_h,\tilde{\bP}_{i,h}]\biggiven\tilde{\Theta}_h\Big]=0.
            \end{align*}
            Therefore, we know that $\tilde{V}_h^q(\tilde{\Theta}_h,\tilde{\gamma}_h')-\tilde{V}_h^q(\tilde{\Theta}_h,\tilde{\gamma}_h^\ast)=\EE[\Delta\tilde{\bU}_h^\top G_h\Delta\tilde{\bU}_h\given \tilde{\Theta}_h]\ge 0$, which implies that $\tilde{V}_h^l(\tilde{\Theta}_h)=\tilde{V}_h^q(\tilde{\Theta}_h,\tilde{\gamma}_h^\ast)=\tilde{V}_h^q(\tilde{\Theta}_h,\tilde{\gamma}_h')=\tilde{V}_h(\tilde{\Theta}_h)$, and $\tilde{\gamma}_h^\ast$ can minimize the right-hand side of Equation \eqref{equ: centralized Bellman Equation} for any reachable mean $\tilde\Theta_h$.

        \end{proof}

\subsubsection{Proof of Theorem \ref{thm: dp_in_decLQG}}
\begin{proof}
        
        The proof consists of two \textbf{Parts}.

    \noindent
        \textbf{Part 1:} we aim to show the evolution of $\tilde{\bTheta}_h$.
        
        For every $h\in[H]$, we define $\overline{\bX}_h:=\tilde{\bS}_h=\begin{bmatrix}
            \tilde{\bX}_h\\
           \tilde{\bP}_h \end{bmatrix},\overline{\bY}_h:=\tilde{\bZ}_h$, and define
               $\overline{\bU}_h:=\tilde{\bU}_h^e=\begin{bmatrix}
                   \mathds{1}[\tilde{B}_{1,h}\neq \mathbf{0}]\tilde{\bU}_{1,h}\\
                   \mathds{1}[\tilde{B}_{2,h}\neq \mathbf{0}]\tilde{\bU}_{2,h}\\
                   \cdots\\
                   \mathds{1}[\tilde{B}_{n,h}\neq \mathbf{0}]\tilde{\bU}_{n,h}
               \end{bmatrix}$, where $\tilde{\bU}_h^e$ is the \emph{effective} control action constructed
            by replacing the $\tilde{\bU}_{i,h}$ part of $\tilde{\bU}_h$ by $\mathbf{0}$ for all $i\in[n]$ such that $\tilde{B}_{i,h}=\mathbf{0}$. We claim that such a system is a linear system with the evolution rules given by
           \begin{align*}
               &\overline{\bX}_{h+1}=\overline{A}_h\overline{\bX}_h+\overline{B}_h\overline{\bU}_h+\overline{\bW}_h,\qquad h\in[H-1],\\
               &\overline{\bY}_{h}=\overline{E}_h\overline{\bX}_{h-1}+\overline{F}_h\overline{\bU}_{h-1}+\overline{\bV}_h,\qquad h\in[2:H],
           \end{align*}
           for some matrices $\overline{A}_h,\overline{B}_h,\overline{E}_h,\overline{F}_h$ and Gaussian noises $\overline{\bW}_h,\overline{\bV}_h$ defined below.
           
           Firstly, from the system evolution, we know that 
           \begin{align*} 
               \forall h\in[H],\quad  \tilde{\bX}_{h+1}=\tilde{A}_h\tilde{\bX}_h+\tilde{B}_h\tilde{\bU}_h+\tilde{\bW}_{0,h},       \tilde{\bY}_{h}=\tilde{E}_h\tilde{\bX}_h+\tilde{\bW}_{1:n,h},
               \tilde{E}_h=\begin{bmatrix}
                   \tilde{E}_{1,h}\\\cdots\\\tilde{E}_{n,h}
               \end{bmatrix},\tilde{\bW}_{1:n,h}=\begin{bmatrix}
                   \tilde{\bW}_{1,h}\\\cdots\\
                   \tilde{\bW}_{n,h}
               \end{bmatrix}. 
           \end{align*}
           Also, from Theorem \ref{thm: expansion}, we know that $\tilde{\bZ}_h,\tilde{\bP}_h$ evolve as
        \begin{align*}
        \tilde{\bZ}_h=\tilde{\cT}_h^1\tilde{\bP}_{h-1}+\tilde{\cT}_h^2\tilde{\bU}_{h-1}+\tilde{\cT}_h^3\tilde{\bY}_h, \qquad \tilde{\bP}_h=\tilde{\cL}_h^1\tilde{\bP}_{h-1}+\tilde{\cL}_h^2\tilde{\bU}_{h-1}+\tilde{\cL}_h^3\tilde{\bY}_h,
        \end{align*}
        for some matrices $\{\tilde{\cT}_h^k\}_{k\in[3]}$ and $\{\tilde{\cL}_h^k\}_{k\in[3]}$. 
           Then, we can show that the system dynamics are defined for $h\in[H-1]$ as
           \begin{align*}
               \overline{A}_h:=\begin{bmatrix}
                   \tilde{A}_h&0\\
                   \tilde{\cL}_{h+1}^3\tilde{E}_{h+1}\tilde{A}_h&\tilde{\cL}_{h+1}^1
               \end{bmatrix},~
               \overline{B}_h:=\begin{bmatrix}
                    \tilde{B}_h\\\tilde{\cL}_{h+1}^2+\tilde{\cL}_{h+1}^3\tilde{E}_{h+1}\tilde{B}_h
                \end{bmatrix},\qquad
            \overline{\bW}_h=\begin{bmatrix}
                    \tilde{\bW}_{0,h}\\\tilde{\cL}_{h+1}^3(\tilde{E}_{h+1}\tilde{\bW}_{0,h}+\tilde{\bW}_{1:n,h+1})
                \end{bmatrix}.
           \end{align*}
           The observation quantities are defined for $t\in[2:H]$ by
           \begin{align*}
                \overline{E}_t:=\begin{bmatrix}
                    \tilde{\cT}_t^3\tilde{E}_t\tilde{A}_{t-1}&\tilde{\cT}_{t}^1
                \end{bmatrix},\quad
                \overline{F}_t:=\tilde{\cT}_t^2+\tilde{\cT}_t^3\tilde{E}_t\tilde{B}_{t-1},\quad
                \overline{\bV}_t:=\tilde{\cT}_{t}^3(\tilde{E}_{t}\tilde{\bW}_{0,t-1}+\tilde{\bW}_{1:n,t}).
           \end{align*}
           
\refstepcounter{footnote}
Note that based on Assumption \ref{ass: useless action}, one can verify that if $\tilde{B}_{i,h}=\mathbf{0}$, then $\tilde{\cT}_{i,h+1}^2=\mathbf{0},\tilde{\cL}_{i,h+1}^2=\mathbf{0}$, where $\tilde{\cT}_{i,h+1}^2$ and $\tilde{\cL}_{i,h+1}^2$ are the corresponding agent $i$'s column blocks of $\tilde{\cT}_{h+1}^2,\tilde{\cL}_{h+1}^2$. 
            Therefore, it holds that for any $h\in[H-1]$, $\overline{B}_h\tilde{\bU}_h=\overline{B}_h\tilde{\bU}_h^e, ~~\overline{F}_{h+1}\tilde{\bU}_h=\overline{F}_{h+1}\tilde{\bU}_h^e$.

            Defining $\overline{\bTheta}_h:=\EE[\overline{\bX}_h\given \overline{\bY}_{1:h},\overline{\bU}_{1:h-1}]=\EE[\tilde{\bS}_h\given \tilde{\bZ}_{1:h},\tilde{\bU}_{1:h-1}]$, the standard linear-system formulas give for $h\in[H-1]$
           \begin{align*}
               &\overline{\bTheta}_{h+1}=\overline{A}_h\overline{\bTheta}_h+\overline{B}_h\overline{\bU}_h+(\overline{A}_h\overline{\Sigma}^p_h\overline{E}_{h+1}^\top+\overline{S}_h)[\overline{E}_{h+1}\overline{\Sigma}^p_h\overline{E}_{h+1}^\top+\Sigma_{h+1}^v]^\dag(\overline{\bY}_{h+1}-\overline{E}_{h+1}\overline{\bTheta}_h-\overline{F}_{h+1}\overline{\bU}_h),\text{\footnotemark[\value{footnote}]}\\
               &\qquad\textbf{with~}\overline{\Sigma}^p_{h+1}:=\overline{A}_h\overline{\Sigma}^p_h\overline{A}_h^\top+\Sigma_h^w-(\overline{A}_h\overline{\Sigma}^p_h\overline{E}_{h+1}^\top+\overline{S}_h)[\overline{E}_{h+1}\overline{\Sigma}^p_h\overline{E}_{h+1}^\top+\Sigma_{h+1}^v]^\dag(\overline{E}_{h+1}\overline{\Sigma}^p_h\overline{A}_h^\top+\overline{S}_h^\top),
               \\
               &\Omega_1:=\diag(\Sigma_1,\tilde\Sigma_{1,1},\ldots,\tilde\Sigma_{n,1}),\quad
               M_1:=\begin{bmatrix}I&0\\ \tilde{\cL}_1^3\tilde E_1&\tilde{\cL}_1^3\end{bmatrix},\quad
               N_1:=\begin{bmatrix}\tilde{\cT}_1^3\tilde E_1&\tilde{\cT}_1^3\end{bmatrix},\\
               &\overline{\bTheta}_1=M_1\Omega_1N_1^\top(N_1\Omega_1N_1^\top)^\dag\tilde{\bZ}_1,\qquad
               \overline\Sigma_1^p=M_1\Omega_1M_1^\top-M_1\Omega_1N_1^\top(N_1\Omega_1N_1^\top)^\dag N_1\Omega_1M_1^\top,\\
               &\Sigma_h^y:=\diag(\tilde\Sigma_{1,h},\cdots,\tilde\Sigma_{n,h}),\qquad h\in[H]\\
               &\overline{S}_h:=\text{cov}(\overline{\bW}_h,\overline{\bV}_{h+1})=\begin{bmatrix}
        \tilde{\Sigma}_{0,h} \tilde{E}_{h+1}^\top\\
\tilde{\mathcal{L}}_{h+1}^3(\tilde{E}_{h+1}\tilde{\Sigma}_{0,h} \tilde{E}_{h+1}^\top
+
\Sigma_{h+1}^y)
		    \end{bmatrix}(\tilde{\cT}_{h+1}^3)^\top,\qquad h\in[H-1],\\
               &\Sigma_h^v:=\text{cov}(\overline{\bV}_h)=\tilde{\cT}_h^3(\tilde{E}_h\tilde{\Sigma}_{0,h-1}\tilde{E}_h^\top+\diag(\tilde{\Sigma}_{1,h},\cdots,\tilde{\Sigma}_{n,h}))(\tilde{\cT}_h^3)^\top,\qquad h\in[2:H],\\
               &\Sigma_h^w:=\text{cov}(\overline{\bW}_h)=\begin{bmatrix}
                   \tilde{\Sigma}_{0,h}&\tilde{\Sigma}_{0,h}\tilde{E}_{h+1}^\top (\tilde{\cL}_{h+1}^3)^\top\\
                   \tilde{\cL}_{h+1}^3\tilde{E}_{h+1}\tilde{\Sigma}_{0,h}& \tilde{\cL}_{h+1}^3(\tilde{E}_{h+1}\tilde{\Sigma}_{0,h}\tilde{E}_{h+1}^\top+\Sigma_{h+1}^y)(\tilde{\cL}_{h+1}^3)^\top
               \end{bmatrix},\qquad h\in[H-1].
           \end{align*}

\footnotetext[\value{footnote}]{Let $\Omega_{h+1}:=\overline E_{h+1}\overline\Sigma_h^p\overline E_{h+1}^\top+\Sigma_{h+1}^v$. Positivity of the joint innovation/state-error covariance implies $\ker(\Omega_{h+1})\subseteq\ker(\overline A_h\overline\Sigma_h^p\overline E_{h+1}^\top+\overline S_h)$. Hence, these Moore-Penrose gains give the canonical singular Kalman update \cite{kalman-filter-pseudoinverse}.}

            For $h\in[H-1]$, strict partial nestedness implies that if $\tilde B_{i,h}\neq0$, then both $\tilde\bI_{i,h}$ and $\tilde\bU_{i,h}$ are contained in $\tilde\bC_{h+1}$. Hence a projection $\tilde{\mathrm{Prj}}_h^{z,u}$ satisfies $\tilde\bU_h^e=\tilde{\mathrm{Prj}}_h^{z,u}\tilde\bZ_{h+1}$.   Therefore, from the definition of $\overline{\bTheta}_h,\overline{\bY}_h$ and $\overline{\bU}_h$, we can conclude
            \begin{equation}
            \label{eq: Final-kalman-filter}
            \begin{aligned}
                &\tilde{\bTheta}_{h+1}=\overline{A}_h\tilde{\bTheta}_h+\overline{B}_h\tilde{\mathrm{Prj}}_h^{z,u}\tilde{\bZ}_{h+1}+(\overline{A}_h\overline{\Sigma}^p_h\overline{E}_{h+1}^\top+\overline{S}_h)[\overline{E}_{h+1}\overline{\Sigma}^p_h\overline{E}_{h+1}^\top+\Sigma_{h+1}^v]^\dag(\tilde{\bZ}_{h+1}-\overline{E}_{h+1}\tilde{\bTheta}_h-\overline{F}_{h+1}\tilde{\mathrm{Prj}}_h^{z,u}\tilde{\bZ}_{h+1}),\\
                &\tilde{\Sigma}_{h+1}=\overline{A}_h\tilde{\Sigma}_h\overline{A}_h^\top+\Sigma_h^w-(\overline{A}_h\tilde{\Sigma}_h\overline{E}_{h+1}^\top+\overline{S}_h)[\overline{E}_{h+1}\tilde{\Sigma}_h\overline{E}_{h+1}^\top+\Sigma_{h+1}^v]^\dag(\overline{E}_{h+1}\tilde{\Sigma}_h\overline{A}_h^\top+\overline{S}_h^\top).                
            \end{aligned}
            \end{equation}

        Note that for $h<H$, Equation  \eqref{eq: Final-kalman-filter} gives precisely the functions $\tilde\Psi_{h+1}^1,\tilde\Psi_{h+1}^2$ in Lemma \ref{lemma: SI-conditional-mean-covariance}.

        Moreover, from Theorem \ref{thm: expansion}, we know that $\tilde{\bZ}_{h+1}=\tilde{\chi}_{h+1}(\tilde{\bP}_h,\tilde{\bU}_h,\tilde{\bY}_{h+1})$ for some projection function $\tilde{\chi}_{h+1}$. Then, we can replace
        $\tilde{\bZ}_{h+1}$ by a linear combination of $\tilde{\bP}_h,\tilde{\bU}_h,\tilde{\bY}_{h+1}$ and then obtain $\{K_{h+1}^k\}_{k\in[4]}$ directly. At $h=H$, since we define the terminal observation to be $\tilde\bY_{H+1}=\tilde\bX_{H+1}$, we can set $K_{H+1}^1=K_{H+1}^2=K_{H+1}^3=\mathbf{0}$ and $K_{H+1}^4=\II$. 

        \vspace{5pt}
    \noindent
        \textbf{Part 2:} we aim to show that the value function $\tilde{V}_h(\tilde{\bTheta}_h)$ has a quadratic form, and the value function and the matrices $\tilde{\cE}_h^\ast,\tilde{\cF}_{h}^\ast$ of the optimal control strategy can be computed via Riccati Equations.

        Note that the proof of this part does not depend on the results that $\tilde{\gamma}_{i,h}^\ast(\cdot):=\tilde{\cE}_{i,h}^\ast\tilde{\Theta}_h+\tilde{\cF}_{i,h}^\ast\cdot$,  
    minimizes the right-hand side of Equation~\eqref{equ: centralized Bellman Equation} at every CIB mean $\tilde\Theta_h$. Instead, it is only based on the fact the existence of the optimal control strategy with the  form $\tilde{g}_{i,h}^\ast(\tilde{\bC}_h,\tilde{\bP}_{i,h})=\tilde{\cE}_{i,h}^\ast\tilde{\Psi}_h^3(\tilde{\bC}_h)+\tilde{\cF}_{i,h}^\ast\tilde{\bP}_{i,h}$, and we can compute the optimal control strategy $\tilde{\gamma}_{i,h}^\ast$ within the strategies class of such form.

        We prove this by induction.  Under the appendix terminal convention, $\tilde{\bTheta}_{H+1}=\tilde{\bX}_{H+1}$, $\tilde\Sigma_{H+1}=0$, and $\II_{x,H+1}=I$. Thus $\tilde V_{H+1}(\tilde\bTheta_{H+1})=\tilde\bTheta_{H+1}^{\top}\tilde Q_{H+1}^1\tilde\bTheta_{H+1}$, so $\tilde R_{H+1}=\tilde Q_{H+1}^1$ and $\tilde c_{H+1}=0$. At $h=H$, this is exactly the final-state term already included in $c_H^\ddag$ in the main text. For the calculation below, we equivalently decompose $c_H^\ddag$ into the ordinary stage-$H$ cost and the auxiliary terminal value $\tilde V_{H+1}(\tilde\bX_{H+1})$.
        Now, for any given $h\in[H]$, we assume that $\tilde{V}_{h+1}(\tilde{\bTheta}_{h+1})=\tilde{\bTheta}_{h+1}^\top \tilde{R}_{h+1}\tilde{\bTheta}_{h+1}+\tilde{c}_{h+1}$ for some $\tilde{R}_{h+1}\succeq0$ and $\tilde{c}_{h+1}$.  From Lemma \ref{lemma: SI-conditional-mean-covariance}, we know that the CIB belief $\tilde{\bB}_h$ is a Gaussian distribution $\cN(\tilde{\bTheta}_h,\tilde{\Sigma}_h)$, so we can write $\tilde{\bS}_h=\tilde{\bTheta}_h+\tilde{\omega}_h$, where $\tilde{\omega}_h\sim \cN(\mathbf{0},\tilde{\Sigma}_h)$. From the definition of $\tilde{\bS}_h$, we can write $\tilde{\bX}_h=\II_{x,h} \tilde{\bS}_h, \tilde{\bP}_h=\II_{p,h}\tilde{\bS}_h$ for the projection matrices $\II_{x,h}$ and $\II_{p,h}$.  Also, for any $i\in[n]$, since the action of agent $i$ has the form $\tilde{\bU}_{i,h}=\tilde{\cE}_{i,h}\tilde{\bTheta}_h+\tilde{\cF}_{i,h}\tilde{\bP}_{i,h}$, we can write the joint control action as $\tilde{\bU}_h=\tilde{\cE}_h\tilde{\bTheta}_h+\tilde{\cF}_h\tilde{\bP}_h$ with $\tilde{\cE}_{h}:=\begin{bmatrix}
        \tilde{\cE}_{1,h}\\\cdots\\\tilde{\cE}_{n,h}\end{bmatrix}$ and $\tilde{\cF}_h:=\diag(\tilde{\cF}_{1,h},\cdots, \tilde{\cF}_{n,h})$. Therefore, $\tilde{\bU}_h=\tilde{\cE}_h\tilde{\bTheta}_h+\tilde{\cF}_h\II_{p,h}(\tilde{\bTheta}_h+\tilde{\omega}_h)=(\tilde{\cE}_h+\tilde{\cF}_h\II_{p,h})\tilde{\bTheta}_h+\tilde{\cF}_h\II_{p,h}\tilde{\omega}_h$. Now, we define $\hat{\cE}_h:=\tilde{\cE}_h+\tilde{\cF}_h\II_{p,h}$, and then optimizing $(\tilde{\cE}_h,\{\tilde{\cF}_{i,h}\}_{i\in[n]})$ is equivalent to optimizing $(\hat{\cE}_h,\{\tilde{\cF}_{i,h}\}_{i\in[n]})$ with respect to the corresponding spaces. This is because for any $(\tilde{\cE}_h,\{\tilde{\cF}_{i,h}\}_{i\in[n]})$, we can find the corresponding $(\hat{\cE}_h,\{\tilde{\cF}_{i,h}\}_{i\in[n]})$, and vice versa. Then, we can write the stage-cost as
        \begin{align*}
            &\min_{\tilde{\cE}_h,\{\tilde{\cF}_{i,h}\}_{i\in[n]}}\EE[\tilde{\bX}_h^\top \tilde{Q}_h^1\tilde{\bX}_h+\tilde{\bU}_h^\top \tilde{Q}_h^2\tilde{\bU}_h\given \tilde{\bTheta}_h=\tilde{\Theta}_h,  \tilde{\cE}_h,\{\tilde{\cF}_{i,h}\}_{i\in[n]}]\\
            &\qquad=\min_{\hat{\cE}_h,\{\tilde{\cF}_{i,h}\}_{i\in[n]}}\EE[\tilde{\bX}_h^\top \tilde{Q}_h^1\tilde{\bX}_h+\tilde{\bU}_h^\top \tilde{Q}_h^2\tilde{\bU}_h\given \tilde{\bTheta}_h=\tilde{\Theta}_h, \hat{\cE}_h,\tilde{\cF}_{h}]\\
            &\qquad=\min_{\hat{\cE}_h,\{\tilde{\cF}_{i,h}\}_{i\in[n]}}\tilde{\Theta}_h^\top\II_{x,h}^\top\tilde{Q}_h^1\II_{x,h}\tilde{\Theta}_h+\tilde{\Theta}_h^\top \hat{\cE}_h^\top \tilde{Q}_h^2\hat{\cE}_h\tilde{\Theta}_h+\tr(\II_{x,h}\tilde{\Sigma}_h\II_{x,h}^\top \tilde{Q}_h^1)+\tr(\tilde{\cF}_h\II_{p,h}\tilde{\Sigma}_h\II_{p,h}^\top \tilde{\cF}_h^\top\tilde{Q}_h^2).
        \end{align*}
        For $h<H$, Theorem \ref{thm: expansion} gives the information evolution $
        \tilde{\bZ}_{h+1}=\tilde{\chi}_{h+1}(\tilde{\bP}_h,\tilde{\bU}_h,\tilde{\bY}_{h+1})$, and Lemma \ref{lemma: SI-conditional-mean-covariance} gives $\tilde{\bTheta}_{h+1}=\tilde{\Psi}_{h+1}^1(\tilde{\bTheta}_h,\tilde{\bZ}_{h+1})$. Then, we can write $\tilde{\bTheta}_{h+1}=\tilde{\Psi}_{h+1}^4(\tilde{\bTheta}_h,\tilde{\bP}_h,\tilde{\bU}_h,\tilde{\bY}_{h+1})$ for some linear function $\tilde{\Psi}_{h+1}^4$. Therefore, we can write $\tilde{\bTheta}_{h+1}=K_{h+1}^1\tilde{\bTheta}_h+K_{h+1}^2\tilde{\bP}_h+K_{h+1}^3\tilde{\bU}_h+K_{h+1}^4\tilde{\bY}_{h+1}$ for some matrices $K_{h+1}^1,K_{h+1}^2,K_{h+1}^3,K_{h+1}^4$. {At $h=H$, the same expression follows from the appendix terminal convention and the terminal matrices stated above.} Also, we know that $\tilde{\bP}_h=\II_{p,h}(\tilde{\bTheta}_h+\tilde{\omega}_h), \tilde{\bU}_h=(\tilde{\cE}_h+\tilde{\cF}_h\II_{p,h})\tilde{\bTheta}_h+\tilde{\cF}_h\II_{p,h}\tilde{\omega}_h, \tilde{\bY}_{h+1}=\tilde{E}_{h+1}\tilde{\bX}_{h+1}+\tilde{\bW}_{1:n,h+1}=\tilde{E}_{h+1}(\tilde{A}_h\tilde{\bX}_{h}+\tilde{B}_h\tilde{\bU}_h+\tilde{\bW}_{0,h})+\tilde{\bW}_{1:n,h+1}$, where $\tilde{E}_{h+1}=\begin{bmatrix}
            \tilde{E}_{1,h+1}\\\cdots\\\tilde{E}_{n,h+1}
        \end{bmatrix}, \tilde{\bW}_{1:n,h+1}=\begin{bmatrix}
            \tilde{\bW}_{1,h+1}\\\cdots\\\tilde{\bW}_{n,h+1}
        \end{bmatrix}$. Therefore, we can write $\tilde{\bTheta}_{h+1}$ as
        \begin{align*}
        \tilde{\bTheta}_{h+1}&=K_{h+1}^1\tilde{\bTheta}_h+K_{h+1}^2\II_{p,h}(\tilde{\bTheta}_h+\tilde{\omega}_h)+K_{h+1}^3(\hat{\cE}_h\tilde{\bTheta}_h+\tilde{\cF}_h\II_{p,h}\tilde{\omega}_h)\\
        &\qquad+K_{h+1}^4(\tilde{E}_{h+1}(\tilde{A}_h\II_{x,h}(\tilde{\bTheta}_h+\tilde{\omega}_h)+\tilde{B}_h(\hat{\cE}_h\tilde{\bTheta}_h+\tilde{\cF}_h\II_{p,h}\tilde{\omega}_h)+\tilde{\bW}_{0,h})+\tilde{\bW}_{1:n,h+1}),\\
        &=(K_{h+1}^1+K_{h+1}^2\II_{p,h}+K_{h+1}^3\hat{\cE}_h+K_{h+1}^4\tilde{E}_{h+1}\tilde{A}_h\II_{x,h}+K_{h+1}^4\tilde{E}_{h+1}\tilde{B}_h\hat{\cE}_h)\tilde{\bTheta}_h+K_{h+1}^4\tilde{\bW}_{1:n,h+1}\\
        &\qquad+(K_{h+1}^2\II_{p,h}+K_{h+1}^3\tilde{\cF}_h\II_{p,h}+K_{h+1}^4\tilde{E}_{h+1}(\tilde{A}_h\II_{x,h}+\tilde{B}_h\tilde{\cF}_h\II_{p,h}))\tilde{\omega}_h+K_{h+1}^4\tilde{E}_{h+1}\tilde{\bW}_{0,h}.
        \end{align*}
        Since $\tilde{\omega}_h$ depends only on primitive randomness through time $h$, it is independent of  $\tilde{\bW}_{0,h}$
        and $\tilde{\bW}_{1:n,h+1}$. Then, we can define the parameters as
        \begin{align*}
            \tilde{L}_h^1&:=K_{h+1}^3+K_{h+1}^4\tilde{E}_{h+1}\tilde{B}_h,~~
            \tilde{L}_h^2:=K_{h+1}^1+K_{h+1}^2\II_{p,h}+K_{h+1}^4\tilde{E}_{h+1}\tilde{A}_h\II_{x,h},\\
            \tilde{L}_h^3&:=K_{h+1}^3+K_{h+1}^4\tilde{E}_{h+1}\tilde{B}_h,~~
            \tilde{L}_h^4:=K_{h+1}^2\II_{p,h}+K_{h+1}^4\tilde{E}_{h+1}\tilde{A}_h\II_{x,h},
        \end{align*}
        and have
        \begin{align*}
            \tilde{\bTheta}_{h+1}=(\tilde{L}_h^1\hat{\cE}_h+\tilde{L}_h^2)\tilde{\bTheta}_h+(\tilde{L}_h^3\tilde{\cF}_h\II_{p,h}+\tilde{L}_h^4)\tilde{\omega}_h+K_{h+1}^4\tilde{E}_{h+1}\tilde{\bW}_{0,h}+K_{h+1}^4\tilde{\bW}_{1:n,h+1}.
        \end{align*}
        Thus, from the induction hypothesis, we have
        \begin{align*}
            &\EE[\tilde{V}_{h+1}(\tilde{\bTheta}_{h+1})\given \tilde{\bTheta}_h,\tilde{\cE}_h,\tilde{\cF}_h]=
            \EE[\tilde{V}_{h+1}(\tilde{\bTheta}_{h+1})\given \tilde{\bTheta}_h,\hat{\cE}_h,\tilde{\cF}_h]=
            \EE[\tilde{\bTheta}_{h+1}^\top \tilde{R}_{h+1}\tilde{\bTheta}_{h+1}+\tilde{c}_{h+1}\given \tilde{\bTheta}_h,\tilde{\cE}_h,\tilde{\cF}_h]\\
            &\qquad =\tilde{\bTheta}_h^\top (\tilde{L}_h^1\hat{\cE}_h+\tilde{L}_h^2)^\top \tilde{R}_{h+1}(\tilde{L}_h^1\hat{\cE}_h+\tilde{L}_h^2)\tilde{\bTheta}_h+\tr((\tilde{L}_h^3\tilde{\cF}_h\II_{p,h}+\tilde{L}_h^4)\tilde{\Sigma}_h(\tilde{L}_h^3\tilde{\cF}_h\II_{p,h}+\tilde{L}_h^4)^\top\tilde{R}_{h+1})\\
            &\qquad\qquad+\tr(K_{h+1}^4\tilde{E}_{h+1}\Sigma_{0,h}\tilde{E}_{h+1}^\top (K_{h+1}^4)^\top\tilde{R}_{h+1})+\tr(K_{h+1}^4\Sigma_{h+1}^y(K_{h+1}^4)^\top\tilde{R}_{h+1})+\tilde{c}_{h+1}.
        \end{align*}
        If we define the following quantities: 
        \begin{align*}
            &J_h^1(\tilde{\bTheta}_h):=\tilde{\bTheta}_h^\top(\II_{x,h}^\top\tilde{Q}_h^1\II_{x,h}+\hat{\cE}_h^\top\tilde{Q}_h^2\hat{\cE}_h+(\tilde{L}_h^1\hat{\cE}_h+\tilde{L}_h^2)^\top \tilde{R}_{h+1}(\tilde{L}_h^1\hat{\cE}_h+\tilde{L}_h^2))\tilde{\bTheta}_h,\\
            &J_h^2:=\tr(\tilde{\cF}_h\II_{p,h}\tilde{\Sigma}_h\II_{p,h}^\top \tilde{\cF}_h^\top\tilde{Q}_h^2)+\tr((\tilde{L}_h^3\tilde{\cF}_h\II_{p,h}+\tilde{L}_h^4)\tilde{\Sigma}_h(\tilde{L}_h^3\tilde{\cF}_h\II_{p,h}+\tilde{L}_h^4)^\top\tilde{R}_{h+1}),\\
            &J_h^3:=\tr(\II_{x,h}\tilde{\Sigma}_h\II_{x,h}^\top \tilde{Q}_h^1)+\tr(K_{h+1}^4\tilde{E}_{h+1}\Sigma_{0,h}\tilde{E}_{h+1}^\top (K_{h+1}^4)^\top\tilde{R}_{h+1})+\tr(K_{h+1}^4\Sigma_{h+1}^y(K_{h+1}^4)^\top\tilde{R}_{h+1})+\tilde{c}_{h+1},
        \end{align*}
        then we can write $\EE[\tilde{\bX}_h^\top \tilde{Q}_h^1\tilde{\bX}_h+\tilde{\bU}_h^\top \tilde{Q}_h^2\tilde{\bU}_h+\tilde{V}_{h+1}(\tilde{\bTheta}_{h+1})\given \tilde{\bTheta}_h, \tilde{\cE}_h,\tilde{\cF}_h]=J_h^1(\tilde{\bTheta}_h)+J_h^2+J_h^3$. 
        Note that $J_h^1(\tilde{\bTheta}_h)$ only depends on $\tilde{\bTheta}_h$ and $\hat{\cE}_h$, $J_h^2$ only depends on $\tilde{\cF}_h$, and $J_h^3$ is just a constant.
        Both minimizations below admit solutions.  Indeed, the coefficient matrix
        $G_h:=\tilde Q_h^2+(\tilde L_h^1)^\top\tilde R_{h+1}\tilde L_h^1$ is positive semidefinite and Lemma~\ref{lemma: kernel_inclusion} gives
        $\operatorname{range}((\tilde L_h^1)^\top\tilde R_{h+1}\tilde L_h^2)\subseteq\operatorname{range}(G_h)$. Hence the canonical least-squares gain $-G_h^\dagger(\tilde L_h^1)^\top\tilde R_{h+1}\tilde L_h^2$ minimizes the first quadratic for every mean vector simultaneously.  The private-gain objective is a bounded-below finite-dimensional convex quadratic, so its normal equations are consistent and it also has a minimizer.
        Therefore, if $\hat{\cE}_h^\ast,\tilde{\cF}_h^\ast$ satisfy   
        \begin{align*}
            \hat{\cE}_h^\ast\in\argmin_{\hat{\cE}_h}J_h^1(\tilde{\bTheta}_h),~~ \forall \tilde{\bTheta}_h\in \tilde{\cS}_h,\qquad \tilde{\cF}_h^\ast\in\argmin_{\tilde{\cF}_h}J_h^2,
        \end{align*}
        then $\tilde{\cE}^\ast_h:=\hat{\cE}_h^\ast-\tilde{\cF}_h^\ast\II_{p,h}$ and $\tilde{\cF}_h^\ast$ will satisfy 
        \begin{align*}
            (\tilde{\cE}_h^\ast,\tilde{\cF}_h^\ast)\in\argmin_{\tilde{\cE}_h,\tilde{\cF}_h}\EE[\tilde{\bX}_h^\top \tilde{Q}_h^1\tilde{\bX}_h+\tilde{\bU}_h^\top \tilde{Q}_h^2\tilde{\bU}_h+\tilde{V}_{h+1}(\tilde{\bTheta}_{h+1})\given \tilde{\bTheta}_h, \tilde{\cE}_h,\tilde{\cF}_h].
        \end{align*}
        Firstly, we intend to minimize $J_h^1(\tilde{\bTheta}_h)$ and find an optimal $\hat{\cE}_h^\ast$.
        Since $\tilde{Q}_h^1,\tilde{Q}_h^2,\tilde{R}_{h+1}\succeq 0$, we can write 
        \begin{align*}
            J_h^1(\tilde{\bTheta}_h)=\norm{(\tilde{Q}_h^1)^\frac{1}{2}\II_{x,h}\tilde{\bTheta}_h}_2^2+\norm{(\tilde{Q}_h^2)^\frac{1}{2}\hat{\cE}_h\tilde{\bTheta}_h}_2^2+\norm{\tilde{R}_{h+1}^\frac{1}{2}(\tilde{L}_h^1\hat{\cE}_h+\tilde{L}_h^2)\tilde{\bTheta}_h}_2^2,
        \end{align*}
        which is a convex function of $\hat{\cE}_h$ given $\tilde{\bTheta}_h$. Therefore, we can compute the optimal $\hat{\cE}_h$ via the first-order condition.
        Taking the gradient of $J_h^1(\tilde{\bTheta}_h)$ with respect to $\hat{\cE}_h$, we obtain
        \begin{align*}
            \nabla_{\hat{\cE}_h}J_h^1(\tilde{\bTheta}_h)=(2(\tilde{Q}_h^2+(\tilde{L}_h^1)^\top\tilde{R}_{h+1}\tilde{L}_h^1)\hat{\cE}_h+2(\tilde{L}_h^1)^\top\tilde{R}_{h+1}\tilde{L}_h^2)\tilde{\bTheta}_h\tilde{\bTheta}_h^{\top}.
        \end{align*}
        Since it holds  that $\ker{(\tilde{Q}_h^2+(\tilde{L}_h^1)^\top\tilde{R}_{h+1}\tilde{L}_h^1)}\subseteq \ker{((\tilde{L}_h^2)^\top\tilde{R}_{h+1}\tilde{L}_h^1)}$ due to Lemma \ref{lemma: kernel_inclusion}, and $\tilde{Q}_h^2+(\tilde{L}_h^1)^\top\tilde{R}_{h+1}\tilde{L}_h^1\succeq 0$, we know that this gradient is zero if we choose
        \begin{align*}
            \hat{\cE}_h^\ast=-(\tilde{Q}_h^2+(\tilde{L}_h^1)^\top\tilde{R}_{h+1}\tilde{L}_h^1)^\dag(\tilde{L}_h^1)^\top\tilde{R}_{h+1}\tilde{L}_h^2. 
        \end{align*}
        Then, we know that such a  $\hat{\cE}_h^\ast$ minimizes $J_h^1(\tilde{\bTheta}_h)$, and we can compute
        \begin{align}
            \tilde{R}_h=\II_{x,h}^\top\tilde{Q}_h^1\II_{x,h}+(\hat{\cE}_h^{\ast})^{\top}\tilde{Q}_h^2\hat{\cE}_h^\ast+(\tilde{L}_h^1\hat{\cE}_h^\ast+\tilde{L}_h^2)^\top \tilde{R}_{h+1}(\tilde{L}_h^1\hat{\cE}_h^\ast+\tilde{L}_h^2).
            \label{eq: evolution of R_h}
        \end{align}
        This is a sum of positive-semidefinite Gram terms, so $\tilde R_h\succeq0$, closing the induction.
Note that the use of pseudo-inverse here 
precisely corresponds to 
the  proof of \emph{generalized}  Riccati Equations in \cite{psedoinverse1,psedoinverse2}.

        Secondly, we intend to minimize $J_h^2$ and find an optimal $\tilde{\cF}_h^\ast$. Hence, the optimal $\tilde{\cF}_h^\ast$ can be found via the following optimization problem
        \begin{align}
        \tilde{\cF}_{h}^\ast&\in\argmin_{\tilde{\cF}_h=\diag(\tilde{\cF}_{1,h},\cdots,\tilde{\cF}_{n,h})}\Big\{\tr(\tilde{\cF}_h\II_{p,h}\tilde{\Sigma}_h\II_{p,h}^\top \tilde{\cF}_h^\top\tilde{Q}_h^2)+\tr((\tilde{L}_h^3\tilde{\cF}_h\II_{p,h}+\tilde{L}_h^4)\tilde{\Sigma}_h(\tilde{L}_h^3\tilde{\cF}_h\II_{p,h}+\tilde{L}_h^4)^\top\tilde{R}_{h+1})\Big\}.
        \label{eq: minimization of private part}
        \end{align}
        Then, plugging in the optimal $\tilde{\cF}_h^\ast$, we can get 
        \begin{equation}
            \label{eq: computing constant c_h}
        \begin{aligned}
            \tilde{c}_h&=
            \tr(\tilde{\cF}_h^\ast\II_{p,h}\tilde{\Sigma}_h\II_{p,h}^\top (\tilde{\cF}_h^\ast)^\top\tilde{Q}_h^2)+\tr((\tilde{L}_h^3\tilde{\cF}_h^\ast\II_{p,h}+\tilde{L}_h^4)\tilde{\Sigma}_h(\tilde{L}_h^3\tilde{\cF}_h^\ast\II_{p,h}+\tilde{L}_h^4)^\top\tilde{R}_{h+1})
            +\tr(K_{h+1}^4\tilde{E}_{h+1}\Sigma_{0,h}\tilde{E}_{h+1}^\top (K_{h+1}^4)^\top\tilde{R}_{h+1})\\
            &\qquad+\tr(K_{h+1}^4\Sigma_{h+1}^y(K_{h+1}^4)^\top\tilde{R}_{h+1})+\tr(\II_{x,h}\tilde{\Sigma}_h\II_{x,h}^\top \tilde{Q}_h^1)+\tilde{c}_{h+1},
        \end{aligned}
        \end{equation}
        
        which completes \textbf{Part 2}.

        Combining \textbf{Part 1} and \textbf{Part 2}, we complete the proof.
        
        \end{proof}
        ~\par

        \noindent\textbf{Solution of Equation  \eqref{eq: minimization of private part}:} From the fact that $\tilde{\Sigma}_h,\tilde{Q}_h^2,\tilde{R}_{h+1}\succeq 0$, we can write
        \begin{align*}
            J_h^2=\norm{(\tilde{Q}_h^2)^{\frac{1}{2}}\tilde{\cF}_h\II_{p,h}(\tilde{\Sigma}_h)^{\frac{1}{2}}}_{F}^2+\norm{(\tilde{R}_{h+1})^{\frac{1}{2}}(\tilde{L}_h^3\tilde{\cF}_h\II_{p,h}+\tilde{L}_h^4)(\tilde{\Sigma}_h)^{\frac{1}{2}}}_{F}^2,
        \end{align*} 
        where $\norm{X}_F$ denotes the Frobenius norm of $X$ for any matrix $X$.
        Therefore, we know that $J_h^2$ is a convex function with respect to $\tilde{\cF}_h$, and the optimal $\tilde{\cF}_h^\ast$ can be computed via the first-order condition. 
        
        By taking the gradient of $J_h^2$ with respect to $\tilde{\cF}_h$, we obtain
        \begin{align*}
            \nabla_{\tilde{\cF}_h}J_h^2=2(\tilde{Q}_h^2\tilde{\cF}_h\II_{p,h}\tilde{\Sigma}_h\II_{p,h}^\top+(\tilde{L}_h^3)^\top\tilde{R}_{h+1}\tilde{L}_h^3\tilde{\cF}_h\II_{p,h}\tilde{\Sigma}_h\II_{p,h}^\top+(\tilde{L}_h^3)^\top\tilde{R}_{h+1}\tilde{L}_h^4\tilde{\Sigma}_h\II_{p,h}^\top). 
        \end{align*}
        Then, we partition $\nabla_{\tilde{\cF}_h}J_h^2$ into $n\times n$ blocks and set its diagonal blocks to zero, namely, let $[\nabla_{\tilde{\cF}_h}J_h^2]_{i,i}=\mathbf{0}$ for all $i\in[n]$. Here the block partitions follow their natural row and column variables: $[\tilde Q_h^2]_{i,k}$ and $[T_h^1]_{i,k}$ have dimensions $\dim(\tilde\cU_i)\times\dim(\tilde\cU_k)$, $[\hat\Sigma_h]_{k,i}$ has dimension $\dim(\tilde\cP_k)\times\dim(\tilde\cP_i)$, and $[T_h^2]_{i,i}$ and $[\nabla_{\tilde\cF_h}J_h^2]_{i,i}$ have dimensions $\dim(\tilde\cU_i)\times\dim(\tilde\cP_i)$. We define $\hat\Sigma_h:=\II_{p,h}\tilde\Sigma_h\II_{p,h}^\top$, $T_h^1:=(\tilde L_h^3)^\top\tilde R_{h+1}\tilde L_h^3$, and $T_h^2:=(\tilde L_h^3)^\top\tilde R_{h+1}\tilde L_h^4\tilde\Sigma_h\II_{p,h}^\top$, where $\hat\Sigma_h$ is the covariance matrix of $\tilde\bP_h$. Then, for any $i\in[n]$, it holds
        \begin{align*}
            [\nabla_{\tilde{\cF}_h}J_h^2]_{i,i}=2\left(\sum_{k=1}^n\big([\tilde{Q}_h^2]_{i,k}+[T_h^1]_{i,k}\big)\tilde{\cF}_{k,h}[\hat{\Sigma}_h]_{k,i}+[T_h^2]_{i,i}\right),
        \end{align*}
        and we can obtain the optimal $\tilde{\cF}_h$ by solving the {equations} $$\sum_{k=1}^n\big([\tilde{Q}_h^2]_{i,k}+[T_h^1]_{i,k}\big)\tilde{\cF}_{k,h}[\hat{\Sigma}_h]_{k,i}+[T_h^2]_{i,i}=\mathbf{0}.
        $$

\section{Deferred Details of \S \ref{sec: closed-loop}}\label{appendix:closed_loop}
For the backward recursion in this appendix only, we use a cost-free post-decision terminal convention:
\[
\bC_{(H+1)^-}:=\bC_{H^+}\cup\{\bU_H,\bX_{H+1}\},\qquad \bP_{(H+1)^-}:=\emptyset,
\]
and the same convention with every quantity tilded for the strict expansion. Thus $\tilde{\bB}_{(H+1)^-}=\delta_{\tilde{\bX}_{H+1}}$, $\tilde{\bTheta}_{(H+1)^-}=\tilde{\bX}_{H+1}$, and $\tilde\Sigma_{(H+1)^-}=\mathbf0$. This is only terminal bookkeeping after all decisions and changes neither the feasible strategies nor the objective.
\subsection{Deferred details of results}
{For $h\in[H]$ and $M_{1:h}\in\cM_{1:h}$, let $\Gamma_{i,h}(M_{1:h})$ denote the admissible prescriptions $\gamma_{i,h}^a:\cP_{i,h^+}(M_{1:h})\to\cU_i$, and write $\gamma_h^a$ and $\Gamma_h^a(M_{1:h})$ for their joint versions.}
\begin{lemma}
    {At a realized common history, let $\mu_{i,h}(P_{i,h^-}):=g_{i,h}^m(C_{h^-}\cup P_{i,h^-},M_{1:h-1})$ and let $\mu_h$ be the joint communication prescription.  If $\cD$ satisfies Assumptions \ref{ass: evolution rule} and \ref{ass: useless action}, then}
    \begin{equation}
        \bB_{h^-}=\Phi_h^1(\bB_{(h-1)^+},\bZ_h^b,\gamma_{h-1}^a),\qquad
        \bB_{h^+}=\Phi_h^2(\bB_{h^-},\bZ_h^a,\bM_h,\mu_h).
    \end{equation}
    Under Assumption \ref{ass: limit_communication}, $\mu_h$ is constant in the private information; conditional on $\bM_h$, it drops out, giving $\bB_{h^+}=\Phi_h^2(\bB_{h^-},\bZ_h^a,\bM_h)$.
    \label{lemma: closed_loop_belief evolution}
\end{lemma}
\begin{proof}
    {Given $\bB_{(h-1)^+}$, the fixed model maps, independent noises, and $\gamma_{h-1}^a$ determine the joint law of $(\bX_h,\bP_{h^-},\bZ_h^b)$; conditioning on $\bZ_h^b$ gives $\Phi_h^1$. For the post-sharing update, impose $\bM_h=\mu_h(\bP_{h^-})$ and $\bZ_h^a=\phi_h(\bM_h,\bP_{h^-})$; conditioning on $(\bM_h,\bZ_h^a)$ gives $\Phi_h^2(\bB_{h^-},\bZ_h^a,\bM_h,\mu_h)$. Under Assumption~\ref{ass: limit_communication}, $\mu_h$ is constant in $\bP_{h^-}$, so conditioning on $\bM_h$ adds no private-information likelihood and $\mu_h$ may be omitted.}
\end{proof}

Given any JCCO problem $\cD$ with closed-loop communication strategies, we can expand it into another JCCO problem $\tilde{\cD}$ (with notation system $~\tilde{}~$). First, for any $h\in[H]$, we define the set $\Xi_{h^-}:=\{(i,t)\given i\in[n],t<h, \bI_{i,t^+}\subseteq \bC_{h^-} \text{ under any communication strategy, and }B_{i,t}\neq \mathbf{0}\}$ and $\Xi_{h^+}:=\{(i,t)\given i\in[n],t<h, \bI_{i,t^+}\subseteq \bC_{h^+} \text{ under any communication strategy, and }B_{i,t}\neq \mathbf{0}\}$. The system variables are identified across the two problems; in particular, $\tilde{\bX}_{1:H+1}:=\bX_{1:H+1}$, $\tilde{\bY}_{j,h}:=\bY_{j,h}$, and $\tilde{\bU}_{j,h}:=\bU_{j,h}$ for every $j\in[n]$ and $h\in[H]$. Then, we can expand $\cD$ to $\tilde{\cD}$ as follows: for any $h\in[H]$
\begin{equation}
\begin{aligned}    \tilde{\bC}_{h^-}=\bC_{h^-}\cup \{\tilde{\bU}_{i,t}\given i\in[n],t<h,(i,t)\in \Xi_{h^-} \}, \tilde{\bC}_{h^+}=\bC_{h^+}\cup \{\tilde{\bU}_{i,t}\given i\in[n],t<h, (i,t)\in \Xi_{h^+}\},\\
    \forall i\in[n], \tilde{\bP}_{i,h^-}=\bP_{i,h^-}\backslash\{\tilde{\bU}_{j,t}\given j\in[n],t<h,(j,t)\in \Xi_{h^-}\}, \tilde{\bP}_{i,h^+}=\bP_{i,h^+}\backslash\{\tilde{\bU}_{j,t}\given j\in[n],t<h, (j,t)\in \Xi_{h^+}\}. 
\end{aligned}
    \label{eq: closed_loop_expansion}
\end{equation}
Then, we have the following lemma. 
\begin{lemma}
    Let $\cD$ be a JCCO problem with PN IS {that satisfies}  Assumptions \ref{ass: evolution rule}, \ref{ass: useless action}, \ref{ass: non-degeneracy}, and \ref{ass: limit_communication}, and let $\tilde{\cD}$ be the JCCO expanded  from $\cD$ according to Equation  \eqref{eq: closed_loop_expansion}. Then, $\tilde{\cD}$ satisfies Assumption \ref{ass: evolution rule}, and the two problems have the same set of achievable costs. Hence either problem admits a team-optimal strategy if and only if the other does. Moreover, there exists a function $\varphi^c$ such that for any team-optimal strategy $\tilde{g}_{1:H}^{m,\ast},\tilde{g}_{1:H}^{a,\ast}$ of  $\tilde{\cD}$, $(g_{1:H}^{m,\ast},g_{1:H}^{a,\ast})=\varphi^c(\tilde{g}_{1:H}^{m,\ast},\tilde{g}_{1:H}^{a,\ast},\cD)$ is a team-optimal strategy of $\cD$, and  
    $J_{\cD}(g_{1:H}^{m,\ast},g_{1:H}^{a,\ast})=J_{\tilde{\cD}}(\tilde{g}_{1:H}^{m,\ast},\tilde{g}_{1:H}^{a,\ast})$.
    \label{lemma: closed_loop_expansion_equivalence}
\end{lemma}
\begin{proof}
    This proof consists of two {\bf Parts}: {\bf Part 1}:  $\tilde{\cD}$ satisfies Assumption \ref{ass: evolution rule}; {\bf Part 2}:  we can construct optimal strategies of $\cD$ from the optimal strategies of $\tilde{\cD}$, and the optimal values of the two problems are the same.

    \vspace{5pt}
    \noindent\textbf{Part 1:}
    To begin with, we discuss a property of the problem $\tilde{\cD}$. For any $i\in[n],h\in[H]$, if $B_{i,h}=\mathbf{0}$, then from Assumption \ref{ass: useless action}, $\bU_{i,h}\notin \bI_{j,(h')^-}, \bU_{i,h}\notin \bI_{j,(h')^+},\forall h'\in[H]\text{ with }h'>h,\ j\in[n]$. Also from expansion, we know that $\tilde{\bU}_{i,h}=\bU_{i,h}$ will never be added into common information. Therefore, $\tilde{\bU}_{i,h}\notin \tilde{\bI}_{j,(h')^-}, \tilde{\bU}_{i,h}\notin \tilde{\bI}_{j,(h')^+},\forall h'\in[H]\text{ with }h'>h,\ j\in[n]$.

    For $h\in[H-1]$, if $B_{i,h}\neq \mathbf{0}$, Assumption \ref{ass: non-degeneracy} gives $j\neq i$ such that $\bU_{i,h}$ influences $\bY_{j,h+1}$ and hence $\bI_{j,(h+1)^-}$. Partial nestedness gives $\bI_{i,h^-}\subseteq \bI_{j,(h+1)^-}$. Since $j\neq i$, every label private to agent $i$ can enter agent $j$'s information only through common information due to Assumption \ref{ass: evolution rule}, and then it holds $\bI_{i,h^-}\subseteq \bC_{(h+1)^-}$. Moreover, $\bI_{i,h^+}\setminus\bI_{i,h^-}\subseteq\bZ_h^a\subseteq\bC_{(h+1)^-}$, so $\bI_{i,h^+}\subseteq\bC_{(h+1)^-}$. Consequently $(i,h)\in\Xi_{(h')^-}\cap\Xi_{(h')^+}$ for every $h'>h$, and the expansion adds $\tilde{\bU}_{i,h}$ to $\tilde{\bC}_{(h+1)^-}$ immediately. No such argument is needed for $h=H$, because Equation   \eqref{eq: closed_loop_expansion} only adds actions with $t<h\le H$.

    Then, based on the property discussed above, we show that $\tilde{\cD}$ satisfies Assumption \ref{ass: evolution rule}.
    \begin{enumerate}[(a)]
        \item 
        Since $\cD$ satisfies Assumption \ref{ass: evolution rule}, $\bZ_h^b=\chi_h(\bP_{(h-1)^+},\bU_{h-1},\bY_{h})$ for a fixed projection function $\chi_h$. 
        We claim that $\tilde{\bZ}_h^b\subseteq \tilde{\bP}_{(h-1)^+}\cup\{\tilde{\bU}_{h-1},\tilde{\bY}_h\}$:
        This is because it holds that $\bZ_h^b\subseteq \bP_{(h-1)^+}\cup\{\bU_{h-1},\bY_h\}$; from the property discussed above, $\tilde{\bZ}_h^b\backslash\bZ_h^b$ only consists of some control actions at timestep $h-1$; from expansion, $(\bP_{(h-1)^+}\backslash\tilde{\bP}_{(h-1)^+})\subseteq \tilde{\bC}_{(h-1)^+}$, so $(\bP_{(h-1)^+}\backslash\tilde{\bP}_{(h-1)^+})\cap \tilde{\bZ}_h^b=\emptyset$. 
        Also, since $\chi_h$ is a projection function, and the sets $\Xi_{h^-},\Xi_{h^+}$ are predefined and are not affected by the realization of each random variable,  there exists a fixed projection function $\tilde{\chi}_h$ such that $\tilde{\bZ}_h^b=\tilde{\chi}_h(\tilde{\bP}_{(h-1)^+},\tilde{\bU}_{h-1},\tilde{\bY}_h)$.
       \item For each realization $M_{i,h}$, obtain $\tilde\phi_{i,h}(M_{i,h},\cdot)$ from the projection $\phi_{i,h}(M_{i,h},\cdot)$ by deleting exactly the labels already moved to $\tilde{\bC}_{h^-}$ by the expansion.  The retained labels lie in $\tilde{\bP}_{i,h^-}$, while the deleted set is fixed by $\Xi_{h^-},\Xi_{h^+}$ and is independent of all realizations.  Hence $\tilde\phi_{i,h}(M_{i,h},\cdot)$ is a fixed projection and $\tilde{\bZ}_{i,h}^a=\tilde\phi_{i,h}(\tilde{\bM}_{i,h},\tilde{\bP}_{i,h^-})$. Stacking over agents gives $\tilde{\bZ}_h^a=\tilde\phi_h(\tilde{\bM}_h,\tilde{\bP}_{h^-})$.
        \item 
        Since $\cD$ satisfies Assumption \ref{ass: evolution rule}, for any $i\in[n],\bP_{i,h^-}=\zeta_{i,h}(\bP_{i,(h-1)^+},\bU_{i,h-1},\bY_{i,h})$ for a fixed projection function $\zeta_{i,h}$. 
        We claim that $\tilde{\bP}_{i,h^-}\subseteq \tilde{\bP}_{i,(h-1)^+}\cup\{\tilde{\bU}_{i,h-1},\tilde{\bY}_{i,h}\}$:
        This is because it holds that $\bP_{i,h^-}\subseteq \bP_{i,(h-1)^+}\cup\{\bU_{i,h-1},\bY_{i,h}\}$; from expansion, it holds that $\tilde{\bP}_{i,h^-}\subseteq \bP_{i,h^-}$, and for any $t<h$ such that $(i,t)\in \Xi_{(h-1)^+}$, it must hold that $(i,t)\in \Xi_{h^-}$, which means if $\tilde{\bU}_{i,t}\in \bP_{i,(h-1)^+}\backslash\tilde{\bP}_{i,(h-1)^+}$, then $\tilde{\bU}_{i,t}\in \bP_{i,h^-}\backslash\tilde{\bP}_{i,h^-}$, so $\tilde{\bU}_{i,t}\notin \tilde{\bP}_{i,h^-}$. 
        Also, since $\zeta_{i,h}$ is a projection function, and the sets $\Xi_{h^-},\Xi_{h^+}$ are predefined and are not affected by the realization of each random variable,  there exists a fixed projection function $\tilde{\zeta}_{i,h}$ such that $\tilde{\bP}_{i,h^-}=\tilde{\zeta}_{i,h}(\tilde{\bP}_{i,(h-1)^+},\tilde{\bU}_{i,h-1},\tilde{\bY}_{i,h})$. Therefore, there exists a function $\tilde{\zeta}_h$ such that $\tilde{\bP}_{h^-}=\tilde{\zeta}_{h}(\tilde{\bP}_{(h-1)^+},\tilde{\bU}_{h-1},\tilde{\bY}_{h})$.
        \item For each $i\in[n], \tilde{\bP}_{i,h^+}=\tilde{\bP}_{i,h^-}\backslash\tilde{\bZ}_{i,h}^a$. 
        \item For each $i\in[n]$, the construction of $\tilde{\cD}$ only adds some actions into $\tilde{\bC}_{h^-},\tilde{\bC}_{h^+}$ and does not change the observations. Thus $\tilde{\bI}_{i,h^-}\subseteq\tilde{\bI}_{i,h^+}$ and $\tilde{\bY}_{i,h}\in\tilde{\bI}_{i,h^-}$ for $h\in[H]$, while $\tilde{\bI}_{i,h^+}\subseteq\tilde{\bI}_{i,(h+1)^-}$ for $h\in[H-1]$.
    \end{enumerate}
    
    \noindent\textbf{Part 2:}
    From the construction of $\tilde{\cD}$, we know that the system dynamics and communication cost are the same for both $\cD$ and $\tilde{\cD}$, and it holds that $\bI_{i,h^-}\subseteq \tilde{\bI}_{i,h^-}$ and $\bI_{i,h^+}\subseteq \tilde{\bI}_{i,h^+}$ for any $i\in[n],h\in[H]$. So the agents in $\tilde{\cD}$ have larger strategy spaces than those in $\cD$. Thus every cost achievable in $\cD$ is achievable in $\tilde{\cD}$.

    More strongly, for every feasible strategy $(\tilde g_{1:H}^m,\tilde g_{1:H}^a)$ of $\tilde\cD$, we recursively construct a feasible strategy $(g_{1:H}^m,g_{1:H}^a)$ of $\cD$ with the same samplewise communication and control actions. For $h=1$ and any $i\in[n]$, note that $\bI_{i,1^-}=\tilde{\bI}_{i,1^-},\bI_{i,1^+}=\tilde{\bI}_{i,1^+}$ always holds. Then, set $g_{i,1}^{m}:=\tilde{g}_{i,1}^{m}$ and $g_{i,1}^{a}:=\tilde{g}_{i,1}^{a}$. It is immediate that these two strategies output the same communication and control actions.

    Now assume that $\tilde{g}_{1:h-1}^{m},\tilde{g}_{1:h-1}^{a}$ and $g_{1:h-1}^{m},g_{1:h-1}^{a}$ output the same actions.  The only additional variables in $\tilde{\bC}_{h^-}\backslash\bC_{h^-}$ are actions $\tilde{\bU}_{j,t}$ with $j\in[n]$ and $t<h$. For each such action, the expansion gives $\bI_{j,t^+}\subseteq \bC_{h^-}$ under any additional sharing. Meanwhile, since $g_{j,t}^{a}$ outputs the same control action as $\tilde{g}_{j,t}^{a}$, we can recover $\tilde{\bU}_{j,t}=g_{j,t}^{a}(\bI_{j,t^+},\bM_{1:t})$. Therefore, we can construct $\tilde{\bC}_{h^-}$ from $\bC_{h^-},\bM_{1:h-1}$ and $g_{1:h-1}^{a}$ by recovering all possible $\tilde{\bU}_{j,t}\in \tilde{\bC}_{h^-}\backslash\bC_{h^-}$. We then construct $g_{i,h}^{m}(\bC_{h^-},\bM_{1:h-1})=\tilde{g}_{i,h}^{m}(\tilde{\bC}_{h^-},\tilde{\bM}_{1:h-1})$, and these two strategies will output the same communication action. 

    Similarly, assume that $\tilde{g}_{1:h}^{m},\tilde{g}_{1:h-1}^{a}$ and $g_{1:h}^{m},g_{1:h-1}^{a}$ output the same actions. For any $i$, the only additional variables in $\tilde{\bI}_{i,h^+}\backslash\bI_{i,h^+}$ are actions $\tilde{\bU}_{j,t}$ with $j\in[n]$ and $t<h$. For each such action, the expansion gives $B_{j,t}\neq \mathbf{0}$ and $\bI_{j,t^+}\subseteq \bC_{h^+}$ under any additional sharing. Meanwhile, since $g_{j,t}^{a}$ outputs the same control action as $\tilde{g}_{j,t}^{a}$, we can recover $\tilde{\bU}_{j,t}=g_{j,t}^{a}(\bI_{j,t^+},\bM_{1:t})$. Therefore, we can construct $\tilde{\bI}_{i,h^+}$ from $\bI_{i,h^+},\bM_{1:h}$ and $g_{1:h-1}^{a}$  by recovering all possible $\tilde{\bU}_{j,t}\in \tilde{\bI}_{i,h^+}\backslash\bI_{i,h^+}$. We then construct $g_{i,h}^{a}(\bI_{i,h^+},\bM_{1:h})=\tilde{g}_{i,h}^{a}(\tilde{\bI}_{i,h^+},\tilde{\bM}_{1:h})$, and these two strategies will output the same control action. 

    Recursion defines $(g_{1:H}^{m},g_{1:H}^{a})=\varphi^c(\tilde g_{1:H}^{m},\tilde g_{1:H}^{a},\cD)$.  The two problems then have identical samplewise actions, dynamics, and costs, so
    $J_{\cD}(g_{1:H}^{m},g_{1:H}^{a})=J_{\tilde\cD}(\tilde g_{1:H}^{m},\tilde g_{1:H}^{a})$ for every feasible expanded strategy. Thus every cost achievable in $\tilde\cD$ is achievable in $\cD$. Together with the preceding inclusion, the achievable cost sets coincide, and team optimality and its existence transfer in both directions.
\end{proof}

Now, we introduce the following theorem as a full version of Theorem \ref{theorem: closed_loop satisfying SI-CIB and Gaussian}.
\begin{theorem}[Full version of Theorem \ref{theorem: closed_loop satisfying SI-CIB and Gaussian}]
    \label{theorem: full_version_closed_loop satisfying SI-CIB and Gaussian}
    Let $\cD$ be a JCCO problem with PN IS that satisfies Assumptions \ref{ass: evolution rule}, \ref{ass: useless action}, \ref{ass: non-degeneracy}, and \ref{ass: limit_communication},  and let $\tilde{\cD}$ be the JCCO expanded from $\cD$. Then $\tilde{\cD}$ is a JCCO problem satisfying the SI-CIB condition, and Assumptions \ref{ass: evolution rule}, \ref{ass: useless action}, \ref{ass: non-degeneracy}, \ref{ass: limit_communication}. Moreover, for any $h\in[H]$, $\tilde{\bB}_{h^-},\tilde{\bB}_{h^+}$ admit Gaussian distributions. Formally, $\tilde{\bB}_{h^-}=\cN(\tilde{\bTheta}_{h^-},\tilde{\Sigma}_{h^-}), \tilde{\bB}_{h^+}=\cN(\tilde{\bTheta}_{h^+},\tilde{\Sigma}_{h^+})$, and
    \begin{equation}
    \label{equ: full_version_closed-loop evolution of mean and covariance}
        \begin{aligned}
\tilde{\bTheta}_{h^-}&=\tilde{\Psi}_{h^-}^1(\tilde{\bTheta}_{(h-1)^+},\tilde{\bM}_{1:h-1},\tilde{\bZ}_h^b),~~ \tilde{\bTheta}_{h^+}=\tilde{\Psi}_{h^+}^1(\tilde{\bTheta}_{h^-},\tilde{\bM}_{1:h},\tilde{\bZ}_h^a),\\
            \tilde{\Sigma}_{h^-}&=\tilde{\Psi}_{h^-}^2(\tilde{\bM}_{1:h-1}),~~\qquad\qquad \qquad\tilde{\Sigma}_{h^+}=\tilde{\Psi}_{h^+}^2(\tilde{\bM}_{1:h}),
        \end{aligned}
    \end{equation} 
    for some functions $\tilde{\Psi}_{h^-}^1,\tilde{\Psi}_{h^+}^1,\tilde{\Psi}_{h^-}^2,\tilde{\Psi}_{h^+}^2$ that do not depend on the strategies $\tilde{g}_{1:h}^m,\tilde{g}_{1:h}^a$. Therefore, it holds that
    \begin{equation}
    \label{equ: full_version_closed-loop mean and covariance}
    \tilde{\bTheta}_{h^-}=\tilde{\Psi}_{h^-}^3(\tilde{\bC}_{h^-},\tilde{\bM}_{1:h-1}), \tilde{\bTheta}_{h^+}=\tilde{\Psi}_{h^+}^3(\tilde{\bC}_{h^+},\tilde{\bM}_{1:h}),
    \end{equation}
    for some functions $\tilde{\Psi}_{h^-}^3,\tilde{\Psi}_{h^+}^3$ that do not depend on the strategies $\tilde{g}_{1:h}^m,\tilde{g}_{1:h}^a$.
\end{theorem}
\begin{proof}
    This proof consists of three {\bf Parts}: {\bf Part 1:} $\tilde{\cD}$ satisfies Assumptions \ref{ass: evolution rule}, \ref{ass: useless action}, \ref{ass: non-degeneracy}, and \ref{ass: limit_communication}; {\bf Part 2:} $\tilde{\cD}$ satisfies the SI-CIB condition;  {\bf Part 3:} Beliefs $\tilde{\bB}_{h^-},\tilde{\bB}_{h^+}$ admit Gaussian distributions, and the mean and covariance satisfy Equations~\eqref{equ: full_version_closed-loop evolution of mean and covariance} and~\eqref{equ: full_version_closed-loop mean and covariance}.
    Throughout this proof, untilded quantities belong to the original problem $\cD$, whereas tilded quantities belong to its strict expansion $\tilde{\cD}$.
    
    \vspace{6pt}
    \noindent\textbf{Part 1:}
    From Lemma \ref{lemma: closed_loop_expansion_equivalence}, we know that $\tilde{\cD}$ satisfies Assumption \ref{ass: evolution rule}; Since for any $i\in[n],h\in[H], t<h$, we add control actions $\tilde{\bU}_{j,t}$ into common information at timestep $h$ in the expansion procedure only if $\tilde{B}_{j,t}=B_{j,t}\neq \mathbf{0}$, $\tilde{\cD}$ satisfies Assumption \ref{ass: useless action}; Since expansion procedure does not change the system dynamics and $\tilde{E}_{i,h}=E_{i,h}$, $\forall i\in[n],h\in[H]$, $\tilde{\cD}$ satisfies Assumption \ref{ass: non-degeneracy}; After expansion, $\tilde{g}_{i,h}^m,\forall i\in[n],h\in[H]$ can still only take $(\tilde{\bC}_{h^-},\tilde{\bM}_{1:h-1})$ as input, and $\tilde{\cD}$ satisfies Assumption \ref{ass: limit_communication}.

    \vspace{6pt}
    \noindent\textbf{Part 2:} 
    Condition on a reachable common history and message sequence. Communication rules then add no private-information likelihood. If $\tilde B_{i,t}=0$, Assumption~\ref{ass: useless action} makes the corresponding control irrelevant to later states and information. If $\tilde B_{i,t}\neq0$, the argument in Part~1 gives $\tilde{\bI}_{i,t^+}\subseteq\tilde{\bC}_{(t+1)^-}$, and the strict expansion also places $\tilde{\bU}_{i,t}$ there; the control rule's message-history input is among the separately conditioned messages. Thus past strategy rules impose only compatibility with conditioned variables and do not change either conditional belief, almost surely. Hence $\tilde\cD$ satisfies SI-CIB.
    
    \vspace{5pt}
    \noindent\textbf{Part 3:}
    For a fixed message sequence, conditioning on the common history fixes the strategy inputs and outputs identified in Part~2. After these compatibility relations are removed, the remaining model variables are affine functions of the primitive Gaussian variables. Hence $\tilde{\bB}_{h^-}$ and $\tilde{\bB}_{h^+}$ are Gaussian almost surely, possibly singular. Conditioned control values affect only affine offsets, so the covariances depend only on $\tilde{\bM}_{1:h-1}$ and $\tilde{\bM}_{1:h}$, respectively; thus $\tilde\Sigma_{h^-}=\tilde\Psi_{h^-}^2(\tilde{\bM}_{1:h-1})$ and $\tilde\Sigma_{h^+}=\tilde\Psi_{h^+}^2(\tilde{\bM}_{1:h})$. At $H+1$, the synthetic terminal disclosure gives $\tilde{\bB}_{(H+1)^-}=\delta_{\tilde{\bX}_{H+1}}$, $\tilde{\bTheta}_{(H+1)^-}=\tilde{\bX}_{H+1}$, and $\tilde\Sigma_{(H+1)^-}=\mathbf0$.

    Because $\tilde{\cD}$ satisfies Assumptions \ref{ass: evolution rule}, \ref{ass: useless action}, and \ref{ass: limit_communication}, Lemma \ref{lemma: closed_loop_belief evolution} implies that for each $h\in[H]$, 
    \begin{align*}
\tilde{\bB}_{h^-}=\tilde{\Phi}_h^1(\tilde{\bB}_{(h-1)^+}, \tilde{\bZ}_h^b,\tilde{\gamma}_{h-1}^a),\qquad  \tilde{\bB}_{h^+}=\tilde{\Phi}_h^2(\tilde{\bB}_{h^-},\tilde{\bZ}_h^a,\tilde{\bM}_h),
    \end{align*}
    for some functions  $\tilde{\Phi}_h^1,\tilde{\Phi}_h^2$. Meanwhile, we know that $\tilde{\bB}_{h^-}, \tilde{\bB}_{h^+}$ are Gaussian distributions that can be characterized by the means $\tilde{\bTheta}_{h^-},\tilde{\bTheta}_{h^+}$ and the covariances $\tilde{\Sigma}_{h^-},\tilde{\Sigma}_{h^+}$, respectively. Then, it holds that $(\tilde{\bTheta}_{h^-},\tilde{\Sigma}_{h^-})=\tilde{\Phi}_{h^-}^3(\tilde{\bTheta}_{(h-1)^+}, \tilde{\Sigma}_{(h-1)^+},
    \tilde{\bZ}_h^b,\tilde{\gamma}_{h-1}^a), (\tilde{\bTheta}_{h^+},\tilde{\Sigma}_{h^+})=\tilde{\Phi}_{h^+}^3(\tilde{\bTheta}_{h^-}, \tilde{\Sigma}_{h^-},
    \tilde{\bZ}_h^a,\tilde{\bM}_h)$ for some functions  $\tilde{\Phi}_{h^-}^3,\tilde{\Phi}_{h^+}^3$.  From $\tilde{\Sigma}_{h^-}=\tilde{\Psi}_{h^-}^2(\tilde{\bM}_{1:h-1}),\tilde{\Sigma}_{h^+}=\tilde{\Psi}_{h^+}^2(\tilde{\bM}_{1:h})$, we have 
    \begin{align*}
        \tilde{\bTheta}_{h^-}=\tilde{\Psi}_{h^-}^1(\tilde{\bTheta}_{(h-1)^+},\tilde{\bM}_{1:h-1},\tilde{\bZ}_h^b,\tilde{\gamma}_{h-1}^a),
        \tilde{\bTheta}_{h^+}=\tilde{\Psi}_{h^+}^1(\tilde{\bTheta}_{h^-},\tilde{\bM}_{1:h},\tilde{\bZ}_h^a),
    \end{align*}
    for some functions $\tilde{\Psi}_{h^-}^1,\tilde{\Psi}_{h^+}^1$. Furthermore, since $\tilde{\cD}$ satisfies the SI-CIB condition, $\tilde{\bTheta}_{h^-}$ does not depend on $\tilde{\gamma}_{h-1}^a$, and $\tilde{\Psi}_{h^-}^1,\tilde{\Psi}_{h^+}^1$ do not depend on the strategies $(\tilde{g}_{1:h}^m,\tilde{g}_{1:h}^a)$. Finally, since $\tilde{\bC}_{h^-}=(\cup_{t=1}^{h}\tilde{\bZ}_t^b)  \cup (\cup_{t=1}^{h-1}\tilde{\bZ}_t^a)$, and $\tilde{\bC}_{h^+}=(\cup_{t=1}^{h}\tilde{\bZ}_t^b)  \cup (\cup_{t=1}^{h}\tilde{\bZ}_t^a)$, we know that $
    \tilde{\bTheta}_{h^-}=\tilde{\Psi}_{h^-}^3(\tilde{\bC}_{h^-},\tilde{\bM}_{1:h-1}), \tilde{\bTheta}_{h^+}=\tilde{\Psi}_{h^+}^3(\tilde{\bC}_{h^+},\tilde{\bM}_{1:h}),$ for some functions $\tilde{\Psi}_{h^-}^3,\tilde{\Psi}_{h^+}^3$ that do not depend on the strategies  $\tilde{g}_{1:h}^m,\tilde{g}_{1:h}^a$. This completes the proof. 
\end{proof}

For $\tilde{\cD}$, the common-information coordinator states are $(\tilde{\bTheta}_{h^-},\tilde{\bM}_{1:h-1})$ before sharing and $(\tilde{\bTheta}_{h^+},\tilde{\bM}_{1:h})$ after sharing. Lemma~\ref{lemma: closed_loop_belief evolution}, SI-CIB, and the coordinator-policy correspondence give Algorithm~\ref{main algorithm}. We assume that all displayed Bellman minima have choices defining admissible strategies.

\subsection{Algorithm to solve PN JCCO with closed-loop communication strategies}

Expand $\cD$ into $\tilde{\cD}$ by Equation   \eqref{eq: closed_loop_expansion}. Under the preceding assumption, Theorem~\ref{theorem: full_version_closed_loop satisfying SI-CIB and Gaussian} and Algorithm~\ref{main algorithm} yield a team-optimal strategy of $\tilde{\cD}$, which Lemma~\ref{lemma: closed_loop_expansion_equivalence} maps to one of $\cD$.

\begin{algorithm}[!h]
    \caption{Dynamic Programming for SI-CIB JCCO with Closed-loop Communication Strategies}
    \label{main algorithm}
    \begin{algorithmic}[1]
    \REQUIRE JCCO $\tilde{\cD}$ satisfying the SI-CIB condition and Assumptions \ref{ass: evolution rule}, \ref{ass: useless action}, \ref{ass: non-degeneracy}, and \ref{ass: limit_communication}.
    \FOR{each $\tilde{M}_{1:H}\in\tilde{\cM}_{1:H}$, and each realization $\tilde{\Theta}_{(H+1)^-}=\tilde{X}_{H+1}\in \tilde{\cX}$}
    \STATE $\tilde{V}_{(H+1)^-}(\tilde{\Theta}_{(H+1)^-},\tilde{M}_{1:H})\leftarrow \tilde{\Theta}_{(H+1)^-}^\top \tilde{Q}_{H+1}^1\tilde{\Theta}_{(H+1)^-}$
    \ENDFOR
    \FOR{$h=H$ to 1}
    \FOR{each $\tilde{M}_{1:h}\in\tilde{\cM}_{1:h}$, and each realization $\tilde{\Theta}_{h^+}\in \tilde{\cX}\times\tilde{\cP}_{h^+}(\tilde{M}_{1:h})$}
    \STATE Choose $\tilde{\gamma}_h^{a,\ast}\in \argmin_{\tilde{\gamma}_h^a\in\tilde{\Gamma}_h^a(\tilde{M}_{1:h})}$\\
    $\qquad\EE[\tilde{\bX}_h^\top \tilde{Q}_h^1\tilde{\bX}_h+\tilde{\bU}_h^\top \tilde{Q}_h^2\tilde{\bU}_h+\tilde{V}_{(h+1)^-}(\tilde{\bTheta}_{(h+1)^-},\tilde{M}_{1:h})\given \tilde{\bTheta}_{h^+}=\tilde{\Theta}_{h^+},\tilde{\bM}_{1:h}=\tilde{M}_{1:h},\tilde{\gamma}_h^a]$ 
    \STATE $\tilde{V}_{h^+}(\tilde{\Theta}_{h^+},\tilde{M}_{1:h})\leftarrow$\\
    $\qquad \EE[\tilde{\bX}_h^\top \tilde{Q}_h^1\tilde{\bX}_h+\tilde{\bU}_h^\top \tilde{Q}_h^2\tilde{\bU}_h+\tilde{V}_{(h+1)^-}(\tilde{\bTheta}_{(h+1)^-},\tilde{M}_{1:h})\given \tilde{\bTheta}_{h^+}=\tilde{\Theta}_{h^+},\tilde{\bM}_{1:h}=\tilde{M}_{1:h}, \tilde{\gamma}_h^{a,\ast}]$ 
    \FOR{each $\tilde{C}_{h^+}$ such that $\tilde{\Psi}_{h^+}^3(\tilde{C}_{h^+},\tilde{M}_{1:h})=\tilde{\Theta}_{h^+}$}
    \FOR{$i\in[n]$}
    \STATE $\tilde{g}_{i,h}^{a,\ast}(\tilde{C}_{h^+},\cdot,\tilde{M}_{1:h})\leftarrow\tilde{\gamma}_{i,h}^{a,\ast}(\cdot)$
    \ENDFOR
    \ENDFOR
    \ENDFOR
    \FOR{each $\tilde{M}_{1:h-1}\in\tilde{\cM}_{1:h-1}$, and each realization $\tilde{\Theta}_{h^-}\in \tilde{\cX}\times\tilde{\cP}_{h^-}(\tilde{M}_{1:h-1})$}
    \STATE Choose $\tilde{M}_h^\ast\in \argmin_{\tilde{M}_h\in \tilde{\cM}_h}\EE[\tilde{\cK}_h(\tilde{M}_h)+\tilde{V}_{h^+}(\tilde{\bTheta}_{h^+}, \tilde{M}_{1:h})\given \tilde{\bTheta}_{h^-}=\tilde{\Theta}_{h^-}, \tilde{\bM}_{1:h-1}=\tilde{M}_{1:h-1}, \tilde{\bM}_h=\tilde{M}_h]$ 
    \STATE $\tilde{V}_{h^-}(\tilde{\Theta}_{h^-}, \tilde{M}_{1:h-1})\leftarrow \EE[\tilde{\cK}_h(\tilde{M}_h^\ast)+\tilde{V}_{h^+}(\tilde{\bTheta}_{h^+}, (\tilde{M}_{1:h-1},\tilde{M}_h^\ast))\given \tilde{\bTheta}_{h^-}=\tilde{\Theta}_{h^-},\tilde{\bM}_{1:h-1}=\tilde{M}_{1:h-1}, \tilde{\bM}_h=\tilde{M}_h^\ast]$ 
    \FOR{each $\tilde{C}_{h^-}$ such that $\tilde{\Psi}_{h^-}^3(\tilde{C}_{h^-},\tilde{M}_{1:h-1})=\tilde{\Theta}_{h^-}$}
    \FOR{$i\in[n]$}
    \STATE $\tilde{g}_{i,h}^{m,\ast}(\tilde{C}_{h^-},\tilde{M}_{1:h-1})\leftarrow \tilde{M}_{i,h}^\ast$
    \ENDFOR
    \ENDFOR
    \ENDFOR
    \ENDFOR
    \RETURN $(\tilde{g}^{m,\ast}_{1:H},\tilde{g}^{a,\ast}_{1:H})$ 
    \end{algorithmic}
\end{algorithm}

\end{document}